\pdfoutput=1

\documentclass[%
    reprint,
    showpacs,preprintnumbers,
    showkeys,
    amsmath,amssymb,
    aps,
    prb,
    floatfix,
    longbibliography,
]{revtex4-2}

\usepackage[english]{babel}
\usepackage[utf8]{inputenc}
\usepackage{graphicx}
\usepackage{bm}
\usepackage{hyperref}
\usepackage{xcolor}
\usepackage{epsfig}
\usepackage{amsmath}
\usepackage{mathtools}
\usepackage{amsfonts}
\usepackage{amssymb}
\usepackage{amsthm}
\usepackage{tensor}
\usepackage{tikz}
\usepackage{dsfont}
\usepackage[all]{nowidow}
\usepackage{pifont}
\usepackage{algorithm}
\usepackage{algpseudocode}
\usepackage[capitalise]{cleveref}
\usepackage[export]{adjustbox}
\usepackage{xspace}
\usepackage{nicefrac}
\let\doi\undefined
\usepackage{doi}
\usepackage{lipsum}
\usepackage{wasysym}

\makeatletter
\AddToHook{cmd/appendix/before}{%
    \crefalias{section}{appendix}%
    \crefalias{subsection}{appendix}
    \crefalias{subsubsection}{appendix}
}
\makeatother

\usepackage[letterspace=130]{microtype}

\definecolor{refColor}{HTML}{0376E9}
\definecolor{figColor}{HTML}{E90303}
\definecolor{urlColor}{HTML}{0bcb9a}

\hypersetup{
    unicode=false,                 
    pdftoolbar=true,               
    pdfmenubar=true,               
    pdffitwindow=false,            
    pdfstartview={FitH},           
    pdftitle={},                   
    pdfauthor={Nicolai Lang},      
    pdfsubject={Quantum physics},  
    pdfcreator={Nicolai Lang},     
    pdfproducer={Nicolai Lang},    
    pdfkeywords={},                
    pdfnewwindow=true,             
    colorlinks=true,               
    linkcolor=figColor,            
    citecolor=refColor,            
    filecolor=magenta,             
    urlcolor=urlColor              
}

\newcommand{\bra}[1]{\mathinner{\langle{#1}|}}
\newcommand{\ket}[1]{\mathinner{|{#1}\rangle}}

\newcommand{\BraKet}[2]{\langle #1 | #2 \rangle}

\renewcommand{\vec}[1]{{\boldsymbol{#1}}}

\newcommand{\tr}[1]{\operatorname{Tr}\left[#1\right]}

\newcommand{\Tr}[2]{\operatorname{Tr}_{#2}\left[#1\right]}

\newcommand{\const}{\textrm{const}}

\newcommand{\A}{\mathcal{A}}
\renewcommand{\H}{\mathcal{H}}

\newcommand{\TT}{\mathbb{T}}
\renewcommand{\S}{\mathcal{S}}

\renewcommand{\L}{\mathcal{L}}

\newcommand{\sub}[1]{{\text{\tiny\textnormal{#1}}}}
\newcommand{\spn}[1]{\operatorname{span}\{\,#1\,\}}

\newcommand{\Sym}[1]{\operatorname{Sym}\left(#1\right)}
\newcommand{\Rb}{{r_\sub{B}}}

\newcommand{\id}{\mathds{1}}
\newcommand{\del}{\partial}

\newcommand{\Z}{\mathbb{Z}}

\newcommand{\G}{\mathcal{G}}
\newcommand{\GG}{\mathcal{G}_{\{1\}}}
\newcommand{\GTG}{{\mathcal{G}\otimes G}}

\newcommand{\SG}{{\square G}}

\newcommand{\Vi}{V_l}
\newcommand{\Vs}{V_s}
\newcommand{\Ei}{E_l}
\newcommand{\Es}{E_s}
\newcommand{\Ds}{\Delta_s}
\newcommand{\GC}{\mathcal{G}_C}
\newcommand{\VC}{V_C}
\newcommand{\EC}{E_C}
\newcommand{\DC}{\Delta_C}
\newcommand{\GCl}{\mathcal{G}_\text{Cl}}
\newcommand{\VCl}{V_\text{Cl}}
\newcommand{\ECl}{E_\text{Cl}}
\newcommand{\DCl}{\Delta_\text{Cl}}
\newcommand{\ZD}{\mathbb{Z}_3}
\renewcommand{\AC}{\mathcal{A}_C}
\newcommand{\U}{\Theta}

\newcommand{\etal}{\emph{et\,al.}\xspace}

\newcommand{\Ac}{{\bar{A}}}

\newcommand{\dA}{{\partial A}}

\newcommand{\anti}[1]{\bar{#1}}
\newcommand{\anyon}[1]{\text{\bfseries\sffamily{#1}}}

\newcommand{\FmatD}[1]{F^{\anyon{DDD}}_{\,\anyon{#1}}}
\newcommand{\RmatD}[1]{R^{\anyon{DD}}_{\anyon{#1}}}

\newcommand{\Hom}[1]{\operatorname{Hom}(#1)}
\newcommand{\aut}[1]{\operatorname{Aut}\left(#1\right)}
\newcommand{\Aut}[1]{\operatorname{Aut}\left(#1\right)}

\newcommand{\Id}{\operatorname{id}}

\newcommand{\nel}{\mathrm{e}}

\newcommand{\double}[1]{\mathcal{D}(#1)}

\newcommand{\hex}{{\scriptsize \varhexagon}}
\newcommand{\sq}{{\scriptsize \square}}

\newcommand{\CaptionMark}[1]{\textbf{#1}}

\newtheorem{theorem}{Theorem}

\newtheorem{lemma}{Lemma}
\newtheorem{proposition}{Proposition}

\graphicspath{{fig/}}

\crefname{paragraph}{Paragraph}{Paragraphs}

\begin{document}

\title{Quantum doubles in symmetric blockade structures}

\author{Hans Peter B\"uchler}
\email{hans-peter.buechler@itp3.uni-stuttgart.de}
\author{Tobias F. Maier}
\author{Simon Fell}
\author{Nicolai Lang}
\thanks{All authors contributed equally to the theoretical analysis and the writing of this manuscript.}
\affiliation{%
    Institute for Theoretical Physics III 
    and Center for Integrated Quantum Science and Technology,\\
    University of Stuttgart, 70550 Stuttgart, Germany
}

\date{\today}


\begin{abstract}
  Exactly solvable models of topologically ordered phases with non-abelian
  anyons typically require complicated many-body interactions which do not
  naturally appear in nature. This motivates the ``inverse problem'' of quantum
  many-body physics: given microscopic systems with experimentally realistic
  two-body interactions, how to design a Hamiltonian that realizes a desired
  topological phase? Here we solve this problem on a platform motivated by
  Rydberg atoms, where elementary two-level systems couple via simple blockade
  interactions. Within this framework, we construct Hamiltonians that realize
  topological orders described by non-abelian quantum double models. We
  analytically prove the existence of topological order in the ground state,
  and present efficient schemes to prepare these states. We also introduce
  protocols for the controlled adiabatic braiding of anyonic excitations to
  probe their non-abelian statistics. Our construction is generic and applies
  to quantum doubles $\double{G}$ for arbitrary finite groups $G$. We
  illustrate braiding for the simplest non-abelian quantum double
  $\double{S_3}$.
\end{abstract}

\maketitle


The ground state phase diagrams of quantum many-body systems at zero
temperature can be extremely rich. Of special interest are quantum phases with
properties that are unique to quantum systems. One of the most intriguing
examples are topologically ordered phases of two-dimensional systems which are
characterized by their pattern of long-range
entanglement~\cite{Kitaev2006a,Levin2006,Wen2013,Wen2017}. It is the
entanglement structure of these gapped ground states which entails anyonic
statistics of excitations and robust ground state degeneracies on topologically
non-trivial manifolds. 
While topological phases with abelian anyonic excitations have useful
applications as quantum error correction
codes~\cite{Kitaev2003,Bombin2006,Terhal2015}, phases with non-abelian
excitations are of special interest due to their higher-dimensional braid group
representations, which makes them potential substrates for topological quantum
computing~\cite{Freedman2002,Nayak2008,Wang2010}. While there is general
consensus that some fractional quantum Hall states are natural examples of
topological orders with abelian anyons~\cite{Bartolomei2020Science,
Nakamura2020DirectObservation}, 
the appearance and realization of non-abelian phases in fractional quantum Hall
states or artificial matter is much more challenging~\cite{Kiryl2015,
ReviewStern2010NonAbelian, Nayak2008}. Especially the experimental detection of
their characteristic entanglement structure (e.g., by probing the anyonic
braiding statistics) is still an open problem.
Meanwhile, on the theory side, there are thoroughly explored models that give
rise to a large variety of topological orders with non-abelian anyons.
Examples are Chern-Simon theories for fractional quantum Hall
states~\cite{Zhang1989,Lopez1991,Halperin1993}, Kitaev's quantum double
models~\cite{Kitaev2003}, and string-net
condensates~\cite{Levin2005,Fidkowski2009}, as well as other minimalistic
models~\cite{Kitaev2006,Lang2015}.  Unfortunately, most of these models rely on
non-trivial multi-body interactions which do not naturally appear in nature --
a major roadblock to experimentally explore these topological orders.
This motivates the inverse problem of quantum many-body physics: provided a
platform with experimentally realistic interactions and tunability of
elementary degrees of freedom, is it possible -- and if so how -- to engineer
quantum systems that naturally realize interesting quantum many-body ground
states?
In this paper, we consider systems characterized by a two-body blockade
interaction and show how these can be used to systematically engineer
microscopic models that realize a large class of non-abelian topological
orders. 

In recent years, a major goal has been to realize and probe topological phases
in artificial matter. In condensed matter settings, $p$-wave superconductors
are promising for the realization of Majorana zero modes, either on the
boundaries of wires~\cite{Lutchyn2010,Fidkowski2011} or in the core of
vortices~\cite{Read2000,Ivanov2001}, while recent progress with two-dimensional
van der Waals materials opens a pathway towards fractional Chern
insulators~\cite{Regnault2011,Sun2011,Tang2011,Jia2024}.  The framework of
quantum simulation provides another promising approach, where a variety of
systems based on cold atomic and molecular gases have been put forward.
First theoretical proposals focused on the realization of Kitaev's honeycomb
model~\cite{Kitaev2006} in optical lattices~\cite{Duan2003} or using a spin
toolbox realized by polar molecules~\cite{Micheli2006}. Other proposals target
the realization of fractional quantum Hall states with rotating
gases~\cite{Cooper2008}, Majorana modes in double wires~\cite{Kraus2013} and
$p$-wave superfluids~\cite{Buehler2014}, and bosonic fractional Chern
insulators with polar molecules~\cite{Yao2013} and Rydberg
atoms~\cite{Weber2022}.
In particular Rydberg atoms have emerged as promising platform to study
topological phases, with the first experimental observation of a symmetry
protected topological phase in one-dimension~\cite{Leseleuc2019}, and attempts
to probe a $\mathbb{Z}_2$ spin liquid with toric code topological
order~\cite{Verresen2021,Semeghini2021}. (Although the experimentally observed
signatures are most likely due to dynamical state preparation, rather than
ground state properties of the engineered
Hamiltonian~\cite{Sahay2022,Giudici2022}.) The advantage of the Rydberg
platform is the high flexibility to arrange atoms in arbitrary
two-~\cite{Nogrette2014,Barredo_2016,Bernien_2017} and three-dimensional
geometries~\cite{Barredo2018}, local optical access to individual atoms, a
large freedom to select internal states to engineer microscopic Hamiltonians,
as well as strong van der Waals and dipolar exchange interactions between
different Rydberg levels~\cite{Jaksch2000,Saffman_2010,Barredo2015}.  Notably,
the strong van der Waals coupling can often be modeled by a simple blockade
interaction~\cite{Bernien_2017,Pichler2018a,Turner2018,Lin2019,Bluvstein2021,Ebadi2022,Nguyen2022,Stastny2023a}:
a strong (infinite) interaction on distances shorter than a (tunable) blockade
radius, and a vanishing interaction on larger distances. The simplicity of this
coupling makes Rydberg atoms a versatile platform for the bottom-up design of
artificial quantum matter.

In this paper, we use a framework inspired by the Rydberg platform to design
microscopic Hamiltonians with ground states that are in the topological phase
of Kitaev's paradigmatic quantum double models (of which the toric code is the
simplest example~\cite{Kitaev2003}).
The latter are characterized by a finite group $G$ and, for non-abelian groups,
their topological order supports non-abelian anyonic excitations. Our framework
is based on microscopic two-level systems, arranged in a periodic
two-dimensional structure, with local detunings, local transverse fields, and
simple two-body blockade interactions.  
The approach presented here is the natural extension to arbitrary groups $G$ of
the construction presented in Ref.~\cite{Maier2025} for the abelian group
$\mathbb{Z}_2$.
The main idea is that blockade interactions can be abstractly described by
vertex-weighted \emph{blockade graphs}, and the design of these graphs can be
guided by two crucial insights. First, ground states map to
\emph{maximum-weight independent sets} and satisfy local constraints that are
encoded in the topology of these graphs. And second, blockade graphs can
feature \emph{local graph automorphisms} that translate to local unitary
symmetries of the Hamiltonian, which, in turn, enforce strong quantum
fluctuations within the subspace of states that satisfy the local constraints.
While closely related models with local symmetries have been recently
proposed~\cite{Chamon_2020,Yang2021,Green_2023,Yu_2024,Yu2025}, the
Hamiltonians introduced here allow for a rigorous poof that their ground state
is in the topological phase of the prescribed quantum double $\double{G}$ for
weak transverse fields. 
Furthermore, a spectral gap to flux anyons can be proved in the thermodynamic
limit, while a finite charge gap is expected as well, but much more challenging
to show rigorously \cite{Gap2025}.
Leveraging this framework, we propose an efficient scheme to adiabatically
prepare these ground states, together with a controllable procedure to prepare
states with localized anyonic excitations. Finally, we propose an efficient
protocol for the adiabatic braiding of anyons to experimentally probe their
non-abelian statistics. 
While our construction is generic and works for arbitrary groups $G$, we
illustrate the braiding protocol for the simplest non-abelian quantum double
$\double{S_3}$. The proposed construction and protocols pave the way towards
probing non-abelian topological orders in artificial matter with realistic
two-body interactions. While our approach is inspired by the Rydberg platform,
our formulation is platform-agnostic and allows for alternative realizations,
e.g., with polar molecules in optical tweezers~\cite{Ruttley2025} or
superconducting qubits connected by microwave cavities acting as a bus to
mediate blockade interactions~\cite{Blais2021}. 

\begin{figure*}[tb]
    \centering
    \includegraphics[width=1.0\linewidth]{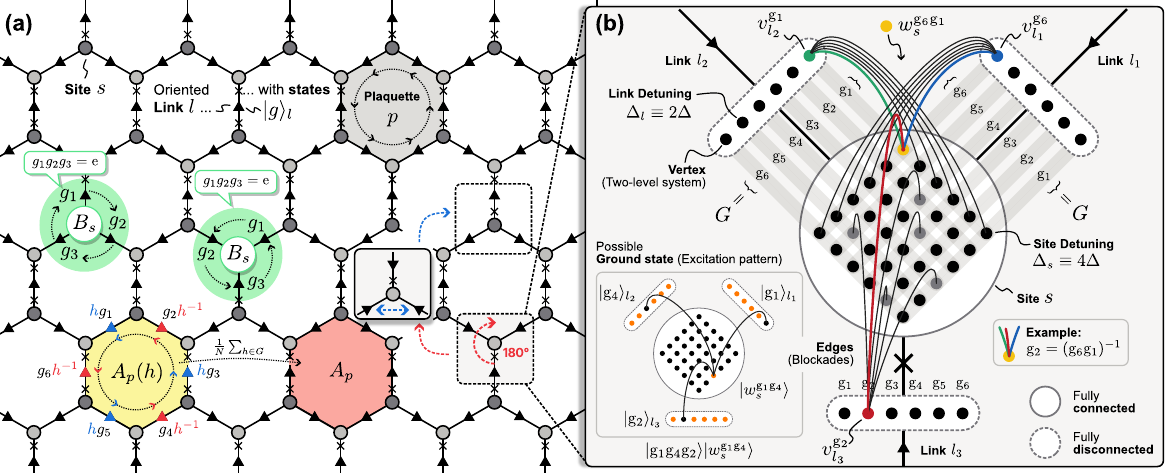}
    \caption{%
        \CaptionMark{Conventions and Construction.}
        (a) We construct quantum doubles \eqref{eq:QD} for a finite group $G$
        on the trivalent and bipartite honeycomb lattice. By convention, the
        links of the lattice are assigned an orientation (solid arrows).  Every
        link $l$ is associated with a $N=|G|$-dimensional quantum system with
        one state $\ket{g}_l$ for each group element $g\in G$. Each plaquette
        $p$ is assigned a counter-clockwise orientation (dotted arrows). With
        this convention, a plaquette $p$ is on the left (right) of a bounding
        link $l$ [write $l\in{}_p\!\!\uparrow$ ($l\in\uparrow_{\!p}$)], if the
        link's orientation is parallel (antiparallel) to the orientation of the
        plaquette. On each site $s$, we define a projector $B_s$ that singles
        out states that satisfy the no-flux condition $g_1g_2g_3=\nel$, where
        $\nel\in G$ denotes the identity and the multiplication sequence
        depends on the sublattice (two green sites).  On each plaquette, there
        are operators $A_p(h)$ that act by left/right multiplication on the
        group elements on the bounding links (blue/red). $A_p(h)$ acts by left
        (right) multiplication if the plaquette orientation aligns
        (counter-aligns) with the link orientation. This construction ensures
        that $A_p(h)$ commutes with all site constraints $g_1g_2g_3=\nel$
        enforced by $B_s$.  Summing over all group elements yields the
        projector $A_p$.
        (b) Microscopically, the quantum double is realized by a blockade
        Hamiltonian \eqref{eq:H} encoded by a blockade graph
        $\G=(V,E,W)$. Depicted is an example for $N=6$ with group
        elements $\mathrm{g_i}$ for $i=1,\ldots,6$ and exemplary group product
        $\mathrm{g_2}=(\mathrm{g_6}\mathrm{g_1})^{-1}$. Note that roman symbols
        like $\mathrm{g_2}$ label \emph{specific group elements}, whereas
        italic symbols like $g_3$ are used as \emph{variables}. For the
        construction it is convenient to mark one of the three edges at each
        site (crosses); this choice has no physical consequence. Then one
        places $N$ two-level systems (vertices) on each link $l$ which are not
        in blockade which each other; each is assigned a group element and
        labeled by $v_l^{g}$. Additionally, there are $N^2$ two-level systems
        on the site labeled by pairs of group elements and denoted by
        $w_s^{g_1g_2}$; these site systems are all in blockade which each other
        (blockades not shown).  The crucial part is how the link vertices are
        connected by edges (blockades) with the site vertices (solid arcs, only
        a few are shown). This construction is explained in the text and
        depends on the orientation of the site (= the sublattice) and the
        marked edge (crosses). The inset shows an exemplary classical ground
        state $\ket{\mathrm{g_1g_4g_2}}\ket{w_s^{\mathrm{g_1g_4}}}$ that
        satisfies all blockades and realizes the state with constraint
        $\mathrm{g_1g_4g_2}=\mathrm{e}$. Note that there is only one two-level
        system excited on the site and all but one on each link.  Although the
        construction seems to break the three-fold rotation symmetry of a site
        (via the edge marked with a cross), the cyclic symmetry of the
        constraint $g_1g_2g_3=\nel$ ensures that the constructed blockade graph
        is completely symmetric under rotations by $120^\circ$. To obtain the
        blockade graph on the other sublattice, one can rotate the shown site
        by $180^\circ$ and swap two of the three edges, thereby inverting the
        orientation of the multiplication around the vertex.
    }
    \label{fig:fig1}
\end{figure*}

\section{The model}
\label{sec:model}

We consider extensive \emph{blockade structures} $\G$ of two-level systems
arranged in space. We denote the total Hilbert space of such a structure by
$\H_\G$. The two-level systems $i$ are subject to a uniform transverse field
$\Omega$ and local detunings $\Delta_i$, and interact via an isotropic Blockade
potential~\cite{Bernien_2017,Pichler2018a,Turner2018,Lin2019,Bluvstein2021,Ebadi2022,Nguyen2022,Stastny2023a};
i.e., their interaction vanishes at large distances and saturates at a value
$U_0$ on distances shorter than a blockade radius~$\Rb$:
\begin{align}
    U(r):=
    \begin{cases}
        0&\text{for}\;r>\Rb\,,\\
        U_0&\text{for}\;r\leq\Rb\,.
    \end{cases}
    \label{eq:U}
\end{align}
Note that in the context of blockade interactions, $U_0$ is often set to
infinity. Here we keep it as free but large parameter which has no effect on
our results but makes rigorous statements easier to prove. 
The simplicity of the blockade potential~\eqref{eq:U} suggests encoding the
spatial arrangement of a structure $\G$ by a \emph{vertex-weighted blockade
graph} $\G\equiv (V,E,W)$: The two-level systems $i\in V$ form the
\emph{vertices} $V$ of the graph, so that the Hilbert space of the structure
has the form $\H_\G=(\mathbb{C}^2)^{\otimes|V|}$. The detunings are interpreted
as the \emph{weights} $W \equiv \{ \Delta_{i}\}$ of the vertices, and the
\emph{edges} $e=\{i,j\}\in E$ of the graph denote pairs of two-level systems
that are in blockade (i.e., are separated by less than the blockade radius
$\Rb$).

With these conventions, the Hamiltonian associated to a blockade
structure/graph $\G$ has the form
\begin{align}
    H_\G &= H_\G^0 + \Omega  \sum_{i\in V} \sigma_i^x\nonumber\\
    \text{with}\quad
    H_\G^0 &= U_0\sum_{\{i,j\}\in E} n_in_j - \sum_{i\in V}  \Delta_i  n_i
    \,,
    \label{eq:H}
\end{align}
where $\sigma_i^{\alpha}$ denotes the Pauli matrices for each two-level system
and $n_{i} = (1-\sigma_i^{z})/2$ is the projector onto the state $\ket{1}_i$. 
The summation of the interaction term in \eqref{eq:H} runs over all pairs
$\{i,j\}\in E$ of sites which are connected by an edge of the blockade graph
$\G$. Clearly there exists a one-to-one correspondence between Hamiltonians of
the form~\eqref{eq:H} and vertex-weighted graphs $\G$ for every transverse
field $\Omega$. However, note that not every abstract graph can be realized as
blockade graph of a spatial structure in two or three dimensions.
The Hamiltonian \eqref{eq:H} belongs to the class of transverse field Ising
models in the presence of a site dependent longitudinal field $\Delta_i$. The
complexity and versatility of this family of Hamiltonians (e.g., to realize
non-abelian topological phases) is hidden in the spatial arrangement of the
two-level systems, and therefore the choice which pairs of two-level systems
are in blockade.

We now present a generic construction of a family of blockade graphs $\G$ such
that the ground states of the corresponding Hamiltonians \eqref{eq:H} realize
all topological phases of Kitaev's quantum double
models~\cite{Kitaev2003,Bombin2008,Cui2015,Komar2017,Cui_2020,Simon2023},
defined on a honeycomb lattice with trivalent sites (\cref{fig:fig1}).
These models are derived from a finite group $G$ of order $N = |G|$. To each
element $g\in G$ a quantum state  $\ket{g}$ is assigned on every \emph{link} of
the lattice; thus the Hilbert space of the quantum double on a periodic
honeycomb lattice with $L$ unit cells is $\H_G = (\mathbb{C}^N)^{\otimes 3 L}$.
In addition, we assign an orientation to each link; it is convenient to choose
all links incoming (outgoing) on alternating lattice sites, see
\cref{fig:fig1}~(a). Then, the Hamiltonian of the quantum double model can be
written as
\begin{equation}
    H_G
    = - J_s  \sum_{\text{Sites $s$}} B_{s}  - J_p \sum_{\text{Faces $p$}}  
    \underbrace{%
        \frac{1}{N}\sum_{h\in G} A_{p}(h)
    }_{A_{p }}
    \label{eq:QD}
\end{equation}
with $J_s>0$ and $J_p>0$. (Note that we are using the convention of
Simon~\cite{Simon2023} -- which is formulated on the \emph{dual} lattice of the
original model introduced by Kitaev~\cite{Kitaev2003}.)

The \emph{site operators} $B_s$ are local projectors onto configurations where
the three states $\ket{g_1}$, $\ket{g_2}$ and $\ket{g_3}$ on the links adjacent
to site $s$ obey the ``no-flux'' constraint $g_1 g_2 g_3 = \nel$ with $\nel\in
G$ the identity of the group $G$. 
Note that in general the group $G$ is non-abelian and therefore the order of
multiplication is important -- here we follow the convention with clockwise
multiplication on sites with outward pointing arrows, and counter-clockwise
multiplication on sites with inward pointing arrows [\cref{fig:fig1}~(a)].  
The \emph{plaquette operators} $A_p(h)$ flip between such configurations by
changing each state $\ket{g_l}$ on links $l\in p$ bounding plaquette $p$ to
$\ket{h g_l}$ or $\ket{g_l h^{-1}}$: if the arrow on the link is parallel to
the counter-clockwise orientation of the loop surrounding the plaquette, the
action is $\ket{h g_l}$ on this link (= the plaquette is on the left of the
link arrow); otherwise, the action is $\ket{g_l h^{-1}}$ (= the plaquette is on
the right of the link arrow), see \cref{fig:fig1}~(a). The sum $A_p$ of
$A_p(h)$ over all group elements $h\in G$ defined in \cref{eq:QD} is then again
a projector.

It is straightforward to show that the projectors $B_s$ and $A_{p}$ all commute
with each other, and the ground state of Hamiltonian \eqref{eq:QD} on a planar
patch with open boundaries (with suitable boundary conditions) is the (unique)
\emph{equal-weight superposition} of all configurations that satisfy the local
constraints imposed by the site terms $B_s$. 
For the abelian group $G= \Z_2$, this model yields the toric code on a
honeycomb lattice, whereas for a general group $G$, the ground state is the
fixpoint wave function of a topological phase characterized by the quantum
double $\double{G}$ of the group $G$~\cite{Dijkgraaf1991,Majid1998,Kitaev2003}.
Notably, for non-abelian groups $G$, these topological phases feature
non-abelian anyonic excitations and can be used for universal topological
quantum computation~\cite{Cui2015}.

Our next goal is to describe a construction $G\mapsto\G$ of a blockade graph
$\G$ for an arbitrary finite group $G$, such that the ground state of the
associated blockade Hamiltonian~\eqref{eq:H} for weak $\Omega\ll\Delta_i,U_0$
is in the topological phase of the quantum double Hamiltonian~\eqref{eq:QD}.
We first explain the rationale of our approach and then describe the detailed
construction below.
We start with $\Omega=0$ and construct a blockade graph $\G=(V,E,W)$
such that there is a one-to-one correspondence between the degenerate ground
states of $H_\G$ and $H_G$ for $J_p=0$ (this ground state manifold is
extensively degenerate).
Crucially, the construction of $\G$ ensures the existence of a group of local
graph automorphisms that translate to local symmetries of the Hamiltonian
$H_\G$. The generators of this local symmetry act on the ground state space of
$H_\G$ exactly like the operators $A_p(h)$ act on the ground states space of
$H_G$ (for $J_p=0$). 
Then we turn on a weak field $\Omega$ and show that the (now unique) ground
state of $H_\G$ exhibits topological order and is in the same phase as the
ground state of the quantum double $H_G$ for finite $J_p>0$. 
Thus the core idea for this realization of Kitaev's quantum double models is to
implement the site terms $B_s$ via diagonal two-body interactions, while the
terms $A_p(h)$ appear as a local symmetry of the microscopic Hamiltonian.
Together with a uniform transverse field $\Omega$, the latter implies the
perturbative generation of terms $A_p(h)$ in the low-energy effective
Hamiltonian of the system.
Note that for $G=\Z_2$ this approach leads to the model presented in
Ref.~\cite{Maier2025}, and large parts of the proof of topological order
presented there can be straightforwardly transferred to the general
construction for arbitrary groups $G$ presented here.

We now describe the construction of the graph $\G$ for the Hamiltonian $H_\G$
in detail. To this end, we consider a finite group $G$ with $N=|G|$ elements.
As shown in \cref{fig:fig1}~(b), we distinguish between two-level systems
placed on the \emph{links} to implement the logical states $\ket{g}_l$, and
two-level systems on the \emph{sites} to realize the site constraints. 
For simplicity, we start with the construction of a blockade graph $\G_s$ for a
single site $s$ of the honeycomb lattice and its three adjacent links, see
\cref{fig:fig1}~(b). On the links we place $N$ two-level systems ($N$
\emph{vertices} in graph language). To each vertex on link $l$ we associate a
unique group element $g\in G$ and denote this vertex by $v_l^{g}$. This yields
a Hilbert space of dimension $2^N$ on each link and we identify the states
$\ket{g}_l$ of the quantum double as basis of an $N$-dimensional subspace
spanned by 
\begin{equation}
    \ket{g}_l \equiv 
    |\,1\,1\ldots\,\underset{\mathclap{\substack{\uparrow\\\text{Vertex $v_l^{g}$}}}}{0}\,\ldots 1\,1\,\rangle_l\,,
\end{equation}
i.e., all two-level systems on the link are excited to state $\ket{1}$ except
for vertex $v_l^{g}$ which is in the de-excited state $\ket{0}$. There are no
edges (= blockades) connecting the vertices on the links among themselves, and
we choose the detunings uniformly $\Delta_l\equiv \Delta_{i} = \Delta$ on all
vertices of the link. 
Next, we place $N^2$ two-level systems with detuning $\Delta_s \equiv
\Delta_{i}=4 \Delta$ on the \emph{site}, such that every pair has a distance
smaller than the blockade radius $\Rb$. Hence the blockade graph on a site with
$N^2$ vertices is \emph{fully connected} and has uniform weight. 

The last and most important step is to define the edges of the blockade graph
that connect the $3\times N$ vertices on the three links with the $N^2$
vertices on their common site. 
To this end, we label the three links adjacent to site $s$ by $l_1, l_2, l_3$
with order as indicated in \cref{fig:fig1}~(a), i.e., clockwise for outgoing
arrows and anti-clockwise for incoming arrows. Furthermore, we assign to each
(ordered!) pair of group elements $g_1,g_2\in G$ a unique vertex on the site
$s$ and denote it by $w_{s}^{g_1 g_2}$. Then the vertex $v_{l_1}^{g_1}$ on the
first link connects to the $N$ vertices $w_{s}^{g_1h}$ for all $h\in G$, while
the vertex $v_{l_2}^{g_2}$ on the second link connects to all vertices
$w_{s}^{h g_2}$ for $h \in G$. Finally, the vertex $v_{l_3}^{g_3}$ on the third
link connects to all vertices $w_s^{g_1 g_2}$ which satisfy the condition
$g_3=(g_1 g_2)^{-1}$. Note that for every $g_3 \in G$ there are exactly $N$
vertices on the site that satisfy this condition. 
In summary, each vertex $w_s^{g_1g_2}$ on a site has an edge with one vertex on
each adjacent link: $v_{l_1}^{g_1}$, $v_{l_2}^{g_2}$, and $v_{l_3}^{g_3}$, with
the three group elements satisfying $g_1 g_2 g_3=\nel$. It is important to
point out that this construction is invariant under cyclic permutations of the
links since $g_1 g_2 g_3= g_2 g_3 g_1 = g_3 g_1 g_3$, i.e., the construction is
invariant under the choice of labeling. In particular, the (apparent)
distinction of one of the three links [\cref{fig:fig1}~(b)] is an artefact of
the construction and not reflected in the graph. Furthermore, a mapping between
sites with incoming arrows and sites with outgoing arrows is possible by
exchanging the labeling between two links.

For such a single site with three adjacent links, the ground states of the
Hamiltonian $H_{\G_s}^0$ (with $U_0 > \Delta_s$) are characterized by a single
vertex $w_s^{g_1g_2}$ in state $\ket{1}$ and all other vertices on the site in
state $\ket{0}$ due to the on-site blockade interactions; we denote this state
by $\ket{w_{s}^{g_1g_2}}$.
Meanwhile, on each adjacent link, the (unique) vertex connected to the excited
vertex $w_{s}^{g_1 g_2}$ is in state $\ket{0}$ due to the blockade, while all
other vertices on the link are in state $\ket{1}$. Therefore each link is in
state $\ket{g_l}_l$ with the constraint $g_1 g_2 g_3 = \nel$.  Hence all states
in the degenerate ground state manifold of a single site can be written as
$\ket{g_1,g_2,g_3}\ket{w_{s}^{g_1g_2}}$ with $g_1 g_2 g_3 = \nel$, and there is
a one-to-one correspondence to the eigenstates of the projection operator $B_s$
with eigenvalue $+1$ of the quantum double~\eqref{eq:QD}. The ground state
energy of $H^0_{\G_s}$ for a single site with three links is $E=-(3N+1)\Delta$.
It is important to stress that here the full Hilbert space of a site with three
links is $2^{3N + N^2}$-dimensional and therefore much larger than the Hilbert
space of a quantum double model with dimension $N^3$. In particular the
one-to-one mapping is only valid for \emph{ground states}, while our blockade
model has a much richer excitation structure. 

Before we extend this analysis to the full honeycomb lattice, we discuss the
\emph{graph automorphisms} of such a single site (and its three adjacent
links). Automorphisms of vertex-weighted graphs are permutations of vertices
that map adjacent vertices to adjacent vertices with the same weights. The set
of all automorphisms of a graph forms its \emph{automorphism group} (with
concatenation of automorphisms as multiplication). Due to the high symmetry of
the construction, the automorphism group of the graph $\G_s$ turns out to be
$(G\times G )\rtimes \operatorname{Aut}(G)$ with $\operatorname{Aut}(G)$ the
group of group automorphisms of $G$. In the following, we focus on an important
subgroup of these graph automorphisms, a discussion of the full automorphism
group can be found in \cref{sec:aut_single}.  

Since each vertex on a link is associated with a unique group element $g\in G$,
the group $G$ induces a permutation $\varphi_{l}(h,k)$ of vertices on link $l$
via
\begin{equation}
    \varphi_{l}(h,k):
    v_{l}^{g} \mapsto v_l^{h g k^{-1}}  
    \quad \text{for every $h,k \in G$.}
    \label{eq:varphi}
\end{equation}
Note that these permutations map the states $\ket{g}_l$ on link $l$ to the
states $\ket{h g k^{-1}}_l$ on the same link. 
Correspondingly, there is a permutation $\phi_{s}(h_1,h_2,h_3)$ of the vertices
on site $s$ defined via
\begin{equation}
    \phi_{s}(h_1,h_2,h_3): 
    w_{s}^{g_1 g_2} \mapsto w_{s}^{h_1 g_1 h_2^{-1} \: h_2 g_2 h_3^{-1}} 
    \label{eq:phi}
\end{equation}
for every triple $h_1,h_2,h_3\in G$.
The permutations \eqref{eq:varphi} and \eqref{eq:phi} allow us to define a
permutation $\Phi_s$ which acts on the vertices of site $s$ and all three
adjacent links, and turns out to be a graph automorphism of $\G_s$ parametrized
by three group elements:
\begin{align}
    \Phi_s(h_1,h_2,h_3) :=
    &\;\varphi_{l_1}(h_1,h_2) 
    \cdot \varphi_{l_2}(h_2,h_3)
    \cdot \varphi_{l_3}(h_3,h_1)
    \nonumber\\
    &\quad\times\phi_s(h_1,h_2,h_3)\,.
    \label{eq:Phi}
\end{align}
Crucially, the automorphism $\Phi_s$ leaves the constraint $g_1 g_2 g_3 = \nel$
invariant. Note that every permutation of vertices (two-level systems) induces
a unitary operator on the Hilbert space; by abuse of notation, we denote
permutations and induced unitaries with the same symbol. Then the (unitary
action of) $\Phi_s$ maps ground states of $H^0_{\G_s}$ onto each other. In
addition, these operations generate a single \emph{orbit}, i.e., starting from
an arbitrary ground state $\ket{g_1,g_2,g_3}\ket{w_s^{g_1g_2}}$ one can reach
any other ground state by applying these graph automorphisms. According to the
definition put forward in Ref.~\cite{Maier2025}, this makes the blockade
structure $\G_s$ \emph{fully-symmetric}, a special feature that ensures that
the ground states of $H_{\G_s}$ for $\Omega \neq 0$ contain \emph{equal-weight
superpositions} of the degenerate ground states of $H^0_{\G_s}$.
Of special interest in the following are graph automorphism
$\Phi_s(h)\equiv\Phi_s(\nel,h,\nel)$ with two elements $h_i$ equal to the
identity~$\nel$. Under these permutations, the state $\ket{g}_{l_3}$ on one
link (here $l_3$) remains invariant, whereas the other two links (here $l_1$
and $l_2$) transform as $\Phi_s(h)\ket{g}_{l_1}=\ket{gh^{-1}}_{l_1}$ and
$\Phi_s(h)\ket{g}_{l_2}=\ket{h g}_{l_2}$, respectively. These automorphisms
allow us to construct local \emph{plaquette automorphisms} below.

We close this section by constructing the blockade graph $\G$ on the full
honeycomb lattice. The most important aspect is that the orientation of the
sites of the honeycomb lattice alternates between \emph{clockwise} (outgoing
arrows) and \emph{anti-clockwise} (incoming arrows), \cref{fig:fig1}~(a). The
construction of both types of sites follows the recipe detailed above: For each
site, we label the adjacent links according to its orientation, and connect the
vertices on the links to the vertices on the site as explained above. 
The site graphs are then joined by identifying the vertices on common links. To
ensure that this construction leads to a gapped ground state manifold that
realizes the intended constraints on every site, one adds up the detunings of
vertices that are shared between sites. This leads to the link detunings
$\Delta_l=2\Delta$ in the bulk of the honeycomb lattice. A mathematically
rigorous discussion of this construction (dubbed \emph{amalgamation}) can be
found in Ref.~\cite{Stastny2023a}. Note that for \emph{open} boundaries, the
above procedure leads to link detunings $\Delta_l=2\Delta$ in the bulk, but
only $\Delta_l=1\Delta$ for dangling links on the boundary.

%

In summary, this construction provides the required graph $\G$ describing the
Hamiltonian $H_\G$ in \cref{eq:H}. The remainder of this paper is dedicated to
studying the ground state properties of $H_\G$ and discussing local
modifications needed for braiding anyons.

\begin{figure}[tb]
    \centering
    \includegraphics[width=1.0\linewidth]{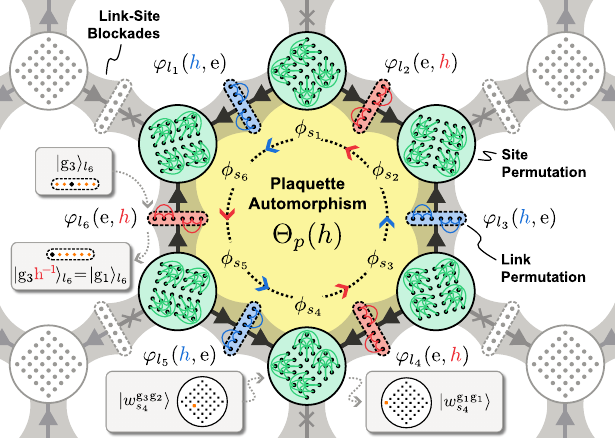}
    \caption{%
        \CaptionMark{Plaquette automorphisms.}
        The blockade graph $\G$ (shaded gray) constructed in \cref{fig:fig1}
        allows for local automorphisms $\Theta_p(h)$ for each $h\in G$ that
        affect only the vertices on links and sites bounding a single plaquette
        $p$. The automorphism decomposes into a product of permutations of
        vertices $\varphi_l$ and $\phi_s$ on links $l\in p$ and sites $s\in p$,
        respectively. The link permutations depend on whether the link
        orientation is parallel (anti-parallel) to the orientation of the
        plaquette (blue and red arrows). The definition of these permutations
        is given in the text. Describing the site permutations $\phi_s$ is most
        convenient if the site vertices are labeled by the pairs of group
        elements of the two adjacent edges that bound the plaquette.  (To apply
        the construction detailed in \cref{sec:model}, the radial edge on every
        site of the plaquette is identified with $l_3$ [marked by a cross];
        this is indicated by rotated and mirrored arrays of site vertices.) As
        a blockade graph automorphism, $\Theta_p(h)$ induces a symmetry of the
        blockade Hamiltonian $H_\G$ on the full Hilbert space.  As such, the
        ground state manifold remains invariant, and the representation induced
        by the link permutations acts by left and right group multiplications
        on the ground states $\ket{\vec g}$ in $\H^0_\G$.  The insets show
        exemplary actions of the permutations on states (black vertices:
        $\ket{0}$, orange vertices: $\ket{1}$). Note that links (and sites)
        that are not adjacent to $p$ are unaffected by the automorphism (gray).
    }
    \label{fig:fig2}
\end{figure}

\section{Ground state properties}
\label{sec:gsprop}


The graph $\G$ on a honeycomb lattice with periodic boundary conditions and $A$
units cells contains $|V|=(3 N + 2 N^2) A$ vertices, so that the Hilbert space
$\H_\G$ of our model is $2^{|V|}$-dimensional. The ground state manifold
$\H^0_\G$ of $H^0_\G$ for $U_0 > 4 \Delta$ is characterized by a state
$\ket{g_l}_l$ on each link, and a state $\ket{w_{s}^{g_i g_j}}$ on each site,
such that the group elements on the links $\vec{g}\equiv\{g_l\}$ satisfy the
constraints of the projectors $B_s$ for all sites (namely $g_1 g_2 g_3=\nel$
for the three links connected to a site).  
We denote states which satisfy this constraint on all sites as $\ket{\vec g}$;
they span the ground state manifold $\H^0_\G$. States in the orthogonal
complement $\H^\perp_\G \equiv (\H^0_\G)^\perp$ are denoted by $\ket{\vec n}$.
This demonstrates the one-to-one mapping between $\H^0_\G$ and the degenerate
ground states of the quantum double Hamiltonian $H_G$ for $J_p=0$. Note that
all states $\ket{\vec{n}}\in \H^\perp_\G$ exhibit an excitation gap of at least
$2 \Delta$.

Next, we show the existence of local graph automorphisms on each plaquette
(\cref{fig:fig2}). For each group element $h\in G$ we can define a permutation
$\Theta_p(h)$ of vertices that belong to the links and sites surrounding a
single plaquette $p$. To define $\Theta_p(h)$, we choose a labeling for each
affected site $s$ such that $l_3$ is the link that points outwards (is not part
of the plaquette). Then the permutation is defined via the permutations
\eqref{eq:varphi} and \eqref{eq:phi} as
\begin{equation}
    \Theta_p(h) 
    :=  
    \prod_{l\in \uparrow_{\!p}} \varphi_{l}(\nel,h) 
    \prod_{l'\in {}_p\!\uparrow} \varphi_{l'}(h,\nel)  
    \prod_{s \in p} \phi_s(\nel,h,\nel)\,,
    \label{eq:plaquette}
\end{equation}
where $\uparrow_{\!p}$ (${}_p\!\!\!\uparrow$) labels the links with the
plaquette to the right (left) of the arrow, and $s \in p$ denotes sites on the
boundary of plaquette $p$.
It is easy to convince oneself that this permutation is a local graph
automorphism of $\G$ for every $h \in G$. To see this, recall that every
\emph{single} site graph $\G_s$ has automorphisms $\Phi_s(\nel,h,\nel)$ which
leave one link ($l_3$) invariant. $\Theta_p(h)$ is the result of chaining six
of these permutations along common links bounding a plaquette. Hence we refer
to $\Theta_p(h)$ as \emph{plaquette automorphisms}.
Remarkably, plaquette automorphisms on different plaquettes commute with each
other -- in analogy to the operators $A_p(h)$ of the quantum double
model~\eqref{eq:QD}. 
Since $\Theta_{p}(h)$ are automorphisms of $\G$, the induced unitary
representations [which we also denote by $\U_{p}(h)$] give rise to local
symmetries of the blockade Hamiltonian $H_\G^0$, i.e.,
\begin{align}
    \U_{p}(h)H_\G^0=H_\G^0\U_{p}(h)
    \quad\text{for all plaquettes $p$ and $h\in G$.}
    \nonumber
\end{align}
As a consequence, all $\U_p(h)$ leave the ground state manifold $\H^0_\G$
invariant and act on ground states $\ket{\vec g}$ exactly like the operators
$A_p(h)$ act on the corresponding ground states of the quantum double
model~\eqref{eq:QD} for $J_p=0$. However, in contrast to the operators
$A_{p}(h)$, the operators $\U_p(h)$ affect not only states on links but also
states on sites, and furthermore act non-trivially on excited states $\ket{\vec
n}\in \H^\perp_\G$.  

We can now discuss the ground state of the full Hamiltonian $H_\G$ with finite
transverse field $\Omega\neq 0$. Note that for uniform $\Omega$ [recall
\cref{eq:H}] the operators $\U_p(h)$ remain symmetries of the full Hamiltonian
$H_\G$.
When a finite patch of the honeycomb lattice is embedded on a topologically
trivial surface with open boundaries (and ``dangling'' edges), the plaquette
automorphisms $\U_p(h)$ map all states $\ket{\vec g}$ in the ground state
manifold $\H^0_\G$ onto each other. This follows from the analogous property of
Kitaev's quantum double models~\cite{Kitaev2003}. Thus the graph automorphisms
$\Theta_p(h)$ generate a single orbit and the complete blockade structure
described by $\G$ is fully symmetric. This allows us to apply Theorem~1 from
Ref.~\cite{Maier2025} which states that in this case the ground state
$\ket{\Omega}$ of $H_\G$ is unique and has the form 
\begin{align}%
    \ket{\Omega}
    =\lambda(\Omega) \sum_{\ket{\vec g} \in \H^0_\G}  \ket{\vec g}
    +\sum_{\ket{\vec n} \in \H^\perp_\G} \eta_{\vec{n}}(\Omega)\ket{\vec n}\,.
    \label{eq:ground_state}
\end{align}%
The first term describes an equal-weight superposition of all product basis
states $\ket{\vec g}$ in $\H^0_\G$, while the second term describes admixtures
of additional states due to the coupling by the transverse field. Note that the
ground state~\eqref{eq:ground_state} is in the subspace $\H^S_\G$
(\emph{symmetric sector}) of all states with eigenvalue $+1$ for all symmetries
$\U_p(h)$, i.e.,
\begin{equation}
    \ket{\Omega}\in\H^S_\G
    = \left\{\,
        \ket{\psi}\,|\,\forall{p,h}:\U_p(h)\ket{\psi} = \ket{\psi} 
    \,\right\} 
    <\H_\G\,.
\end{equation}

Note that for \emph{periodic} boundary conditions, there is also a unique
ground state in the symmetric sector. However, equal-weight superpositions of
states in $\H^0_\G$ are only guaranteed within the orbits generated by
plaquette automorphisms \footnote{%
    On a torus, not all ground state configurations $\ket{\vec g}\in\H_{\G}^0$
    (satisfying the no-flux constraint on every site) can be generated by
    applying plaquette (or loop) automorphisms on the state
    $\ket{\text{\bfseries e}}$ with all $g_l=\nel$. This partitions $\H_{\G}^0$
    into \emph{topological sectors} -- and the amplitudes in the (unique)
    ground state between these sectors are not necessarily fixed by symmetries. 
}.

We stress that due to the admixtures in \cref{eq:ground_state}, one cannot
immediately conclude that $\ket{\Omega}$ is topologically ordered.  However,
for weak $\Omega \ll \Delta$, we can follow the arguments from
Ref.~\cite{Maier2025} to show that the state~\eqref{eq:ground_state} is
topologically ordered and characterized by the quantum double model
$\double{G}$. The rigorous proof for $G=\Z_2$ is given in Ref.~\cite{Maier2025}
and its extension to arbitrary groups $G$ is detailed in \cref{sec:prooftop}.
Here we only sketch the gist of the proof. It requires periodic boundary
conditions and makes use of the extended auxiliary Hamiltonian
\begin{equation}
    \tilde{H}_\G(\Omega, \omega) 
    := 
    \begin{aligned}[t]
        H^0_\G 
        \;+\;&\Omega \sum_{i\in V} \sigma_i^x 
        \\[-15pt]
        \;+\;&\omega\!\!\!\sum_{\text{Faces $p$}}\Bigl[
            \mathds{1}-\overbrace{\frac{1}{|G|}\sum_{h \in G}\U_p(h)}^{\mathclap{\text{Projector $\U_p$}}}
        \Bigr]\,.
        \end{aligned}
        \label{eq:exH}
\end{equation}
For $\omega=0$ and $\Omega=0$ we recover the extensive ground state degeneracy
of $\H^0_\G$ (now with topological degeneracies). In a first step, we turn on
$\omega$ with $0<\omega<\Delta$. This lifts the extensive degeneracy and the
ground states become equal-weight superpositions of states $\ket{\vec g}$ in
the orbits generated by plaquette symmetries $\U_p(h)$. The new ground state
manifold $\H^\omega_\G=\H^0_\G\cap\H^S_\G$ has only topological degeneracies
and is separated by a finite excitation gap of order $\omega$ from excited
states. The ground states in $\H^\omega_\G$ can be mapped by local unitaries to
the ground states of Kitaev's quantum double model~\eqref{eq:QD} and are
therefore topological ordered; note that they are also in the symmetric sector
$\H^S_\G$ since $\U_p(h)\U_p=\U_p$ for all $h\in G$.
The ground state manifold $\H^\omega_\G$, together with the Hamiltonian $\tilde
H_\G(0,\omega)$, shows that the latter is \emph{frustration-free} and satisfies
a condition called \emph{local topological quantum
order}~\cite{Bravyi2010,Michalakis2013,Bravyi2011}.
Given these features, it can be shown that its gap is stable in the
thermodynamic limit against arbitrary small, local
perturbations~\cite{Michalakis2013}, which implies that the ground state
manifold remains in the same gapped topological phase~\cite{Chen2010}. In
particular, we can turn on a small transverse field $\Omega \ll \omega,\Delta$
and the new ground state manifold $\H^{\Omega,\omega}_\G$ of the Hamiltonian
$\tilde H_\G(\Omega,\omega)$ remains topologically ordered. (Note that the
topological ground state degeneracy can be lifted by finite size effects.)
Using the arguments from above, we know that $H_\G$ has a unique ground state
$\ket{\Omega}$ in the symmetric sector $\H^S_\G$. Since $H_\G$ commutes with
$\tilde{H}_\G(\Omega,\omega)$ and $\ket{\Omega}$ is annihilated by the
positive-semidefinite auxiliary term in \cref{eq:exH}, $\ket{\Omega}$ must also
be the (unique) ground state of $\tilde{H}_\G(\Omega,\omega)$, i.e.,
$\ket{\Omega}\in \H^{\Omega,\omega}_\G$. This demonstrates that the ground
state $\ket{\Omega}$ of $H_\G$ for small enough $\Omega\ll \Delta$ is in the
same topological phase as Kitaev's quantum double model~\eqref{eq:QD}. 

The gap stability argument above also implies an excitation gap of order
$\Delta$ to all states in the symmetric sector $\H^S_\G$ (this includes ``flux
excitations'' of the quantum double, see \cref{sec:fluxlattice} below).
Consequently, the ground states are protected by a gap against any perturbation
which respects the plaquette symmetries $\U_{p}(h)$. 
  %
  %
An important question is whether the Hamiltonian $H_\G$ also exhibits a gap
between states in the symmetric sector $\H^S_\G$ and its complement
$(\H^S_\G)^\perp$ (this includes ``charge excitations''). Such a gap is at
least suggested by a (heuristic) Schrieffer-Wolff
transformation~\cite{Bravyi_2011} that predicts an effective Hamiltonian within
the manifold $\H^0_\G$ of the form
\begin{equation}
    H_\sub{eff} 
    \sim 
    - \Delta_c \sum_{p}\sum_{h\in G\backslash\{\nel\}} \U_{p}(h)
    \label{eq:SW}
\end{equation}
with coupling $\Delta_c\sim\Delta\left(\frac{\Omega}{\Delta}\right)^{K}$ and
$K$ the number of two-level systems that flip their state under the action of
$\U_p(h)$ (i.e., $K=24$ for the honeycomb lattice). However, a rigorous proof
of the existence of a finite charge gap is technically challenging and deferred
to an upcoming paper \cite{Gap2025}.

\section{Flux anyons and Wilson loops} 
\label{sec:fluxlattice}

The types of anyonic excitations of the quantum double $\double{G}$ are
classified by the irreducible representations of the Drinfeld
double~\cite{Drinfeld_1988,Kitaev2003}. A systematic characterization of the
anyons is given by the following construction~\cite{Dijkgraaf1991,Simon2023}.
First, one picks a conjugacy class $C$ of the group $G$ with an arbitrary
representative $r_C\in C$, i.e., $C =\{g r_C g^{-1}\,|\,g\in G\}$. Next, one
considers the centralizer $Z_G(r_C)=\{g\in G\,|\,g r_C=r_C g\}$ of this
representative, i.e., the subgroup of $G$ of all elements that commute with
$r_C$. The different anyon types of the quantum double $\double{G}$ can then be
labeled by pairs $[C,R]$ of a conjugacy class $C$ and an irreducible
representation $R$ of its centralizer $Z_G(r_C)$. (Note that the irreducible
representations are independent of the representative $r_C$ since centralizers
of different representatives are isomorphic via conjugation.) The quantum
dimension $d_{[C,R]}$ of an anyon turns out to be the product of the number of
elements in the conjugacy class $C$ and the dimension $d_R$ of the irreducible
representation $R$: $d_{[C,R]}=|C|d_R$.
Following the nomenclature of a lattice gauge theory, one distinguishes
\emph{flux anyons} $[C,E]$ given by a conjugacy class $C$ and the trivial
representation $E$, and \emph{charge anyons} $[C_\nel,R]$ given by an
irreducible representation $R$ of the group $G=Z_G(\nel)$ and the conjugacy
class $C_\nel\equiv\{\nel\}$ of the identity; anyons $[C,R]$ with $C\neq
C_\nel$ and $R\neq E$ carry flux and charge and are referred to as
\emph{dyons}.

The Hamiltonian $H_\G$ with its blockade graph $\G$ (defined via the site
graphs $\G_s$) enforces the zero-flux constraint $g_1 g_2 g_3 = \nel$ on every
site by construction. This means that flux excitations (which violate this
constraint) are energetically penalized. In the following, we present a
straightforward generalization of the site graph $\G_s$ that allows for the
preparation of states with an arbitrary flux anyon trapped on site $s$, while
maintaining a well-defined excitation gap to other flux sectors (and the
vacuum).
As explained above, flux anyons $[C,E]$ are characterized by a conjugacy class
$C$ of the group $G$. One can create such a localized flux on a given site $s$
by enforcing the modified constraint $g_1 g_2 g_3 \in C$ on the three links
connecting to this site. In the following, we describe a generalization
$\G_s[\bm\Delta]$ of the site graph $\G_s$ which allows for the preparation of
any flux anyon $[C,E]$ on site $s$ in the ground state. 
Instead of $N^2=|G|^2$ two-level systems on the site, we now need $N^3$
two-level systems, all of which are in blockade, so that the induced graph is
fully connected and only one of these two-level systems can be excited at any
time; we label these vertices by $w_s^{g_1 g_2 r}$ with $g_1,g_2,r\in G$. Next,
we choose detunings $\bm\Delta\equiv\{\Delta_C\}$ for all conjugacy classes $C$
of $G$; these determine by how much the energy of a flux anyon $[C,E]$ on this
site is lowered, and therefore play the role of site-local chemical potentials
for flux anyons. To this end, we set the detuning of vertex $w_s^{g_1 g_2 r}$
to $\Delta_C$ where $C$ is the conjugacy class of $r$ (i.e.~$r\in C$), such
that for every conjugacy class $C$ there are $|C|N^2$ vertices with the same
detuning $\Delta_C$.
We can now define the edges of the graph $\G_s[\bm\Delta]$ that connect the
vertices on the adjacent links to the vertices on the site. Each vertex
$v_{l_1}^{g_1}$ on link $l_1$ is connected to the $N^2$ vertices $w_s^{g_1 h
r}$ on the site, while each vertex $v_{l_2}^{g_2}$ on link $l_2$ is connected
to the $N^2$ vertices $w_s^{h g_2 r}$. Finally, we connect every vertex
$w_{s}^{g_1 g_2 r}$ on the site to the vertex $v_{l_3}^{g_3}$ on link $l_3$
with $g_3 = (g_1 g_2)^{-1}r$. As a consequence, each configuration which
satisfies the constraint $g_1 g_2 g_3 \in C$, with $C$ a conjugacy class of
$G$, has the same energy contribution $-\Delta_C$. To prepare a specific flux
anyon $[C_r,E]$ on such a site, one sets $\Delta_{C_r}=4\Delta$ and
$\Delta_C=-1$ for all $C\neq C_r$. This lowers the energy of the flux anyon
$[C_r,E]$ and separates it by a gap from all other flux excitations (and the
vacuum).
It is crucial that this site graph respects all local graph automorphisms of
$\G_s$ discussed in \cref{sec:model}, and a Hamiltonian $H_{\tilde\G}$ derived
from a graph $\tilde\G$ that uses the generalized site graph
$\G_s[\bm\Delta_s]$ on some sites (with potentially different detunings
$\bm\Delta_s$) is still symmetric under the local plaquette automorphisms
$\U_{p}(h)$ introduced in \cref{sec:gsprop}. Only global symmetries derived
from outer automorphisms of $G$ are modified by this construction; for details
see \cref{sec:aut_single}.

By construction, the ground states of the Hamiltonian $H^0_{\tilde\G}$ in the
symmetric sector $\H^S_{\tilde\G}$ are described by flux anyons trapped on the
sites with a generalized site graph $\G_s[\bm\Delta_s]$, and their energy can
be tuned by the detunings $\bm\Delta_s$.  It is important to point out that the
ground state manifold of a system with non-trivial flux anyons can have a
topological ground state degeneracy even on topologically trivial surfaces.
Furthermore, there is a finite excitation gap to states within the  symmetric
sector $\H^S_{\tilde\G}$. Therefore, this construction can be used to prepare
states with a preferred flux pattern on the lattice as \emph{ground states} of
the blockade Hamiltonian $H^{0}_{\tilde\G}$. We will use this feature in
\cref{sec:braiding} to implement the braiding of flux anyons. Note that the
above construction can be significantly simplified if the goal is to trap a
\emph{specific} flux anyon $[C_r,E]$. In this case, one can omit all site
vertices except the $|C_r|N^2$ vertices that belong to the relevant conjugacy
class (and, if required, the $N^2$ vertices for the vacuum $C_\nel$). 

Following the discussion in \cref{sec:gsprop} on the robustness of topological
order for a finite transverse field $\Omega \neq 0$ in a flux-free state on a
torus, we expect that the topological properties and ground state degeneracy of
the modified Hamiltonian $H_{\tilde\G}$ in the symmetric sector
$\H^S_{\tilde\G}$ are again robust in the presence of a small but finite
transverse field $\Omega \neq 0$, even in the presence of imprinted flux
anyons. In this case, the splitting of the ground state degeneracy is
exponentially suppressed in the distance between the flux anyons (instead of
the system size) and therefore requires that flux anyons are far apart.

\begin{figure}[tb]
    \centering
    \includegraphics[width=1.0\linewidth]{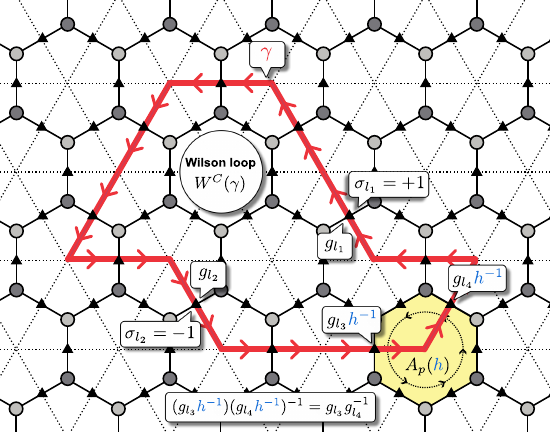}
    \caption{%
        \CaptionMark{Wilson loops.}
        Wilson loop operators $W^C(\gamma)$ are defined on closed, oriented
        loops $\gamma$ (red) on the dual lattice (dotted). For the evaluation
        of $W^C(\gamma)$ in the product basis $\ket{\vec g}\in \H_G$, the
        product of all group elements $g_l^{\sigma_l}$ on links $l$ crossed by
        $\gamma$ is needed (with multiplication order from left to right along
        the loop's orientation); here $\sigma_l=+1$ ($\sigma_l=-1$) if the
        crossed edge $l$ points to the left (right) of $\gamma$ when following
        its orientation. This construction ensures that $W^C(\gamma)$ commutes
        with the plaquette operators $A_p(h)$ (plaquette in the bottom right
        corner).
    }
    \label{fig:fig4}
\end{figure}


In analogy to lattice gauge theories, the flux anyons  are conveniently probed
and characterized by \emph{Wilson loops}. In the case of quantum doubles,
Wilson loops are closely related to closed charge-like \emph{Ribbon operators}
which probe the enclosed flux~\cite{Kitaev2003,Simon2023}. Wilson loop
operators are associated to a closed, oriented loop $\gamma$ on the \emph{dual}
lattice, see~\cref{fig:fig4}. In the product basis $\ket{\vec g}\in \H_G$ of
the quantum double model~\eqref{eq:QD}, they are defined by
\begin{align}
    W^{R}(\gamma) 
    := 
    \chi_R \biggl(\;\prod_{l \in \gamma} g_{l}^{\sigma_l}\biggr)
    \label{eq:Wilsonloop} 
\end{align}
%
%
with $R$ an irreducible representation of the group $G$ and $\chi_{R}$ its
character. The product runs over all links crossed by the closed loop $\gamma$
(with multiplication order from left to right along the loop's orientation).
The exponents $\sigma_l\in \{-1,1\}$ are defined such that $\sigma_{l}=+1$
($\sigma_{l}=-1$) if the arrow of the crossed link points left (right) when
following the loop along its orientation. This convention ensures that
$W^R(\gamma)$ commutes with all plaquette operators $A_p(h)$ (\cref{fig:fig4}).

For fixpoint ground states of a quantum double model -- which in our case
correspond to the ground states of $\tilde H_{\tilde \G}(0,\omega)$ or,
equivalently, the ground states of $H^0_{\tilde \G}$ in the symmetric sector --
the Wilson loops \eqref{eq:Wilsonloop} are independent of the shape of the
loop, and only depend on the enclosed flux.
For example, a loop $\gamma$ that encloses a flux $[C,E]$ yields $\langle
W^{R}(\gamma) \rangle=\chi_R(r_C)$ with $r_C\in C$, i.e., the measurement over
all irreducible representations $R$ uniquely determines the enclosed flux. 
This property can be made explicit by taking the discrete Fourier transform of
the Wilson loop $W^{R}(\gamma)$ over the group $G$,
\begin{subequations}
    \label{eq:Wilsonloop2}
    \begin{align}
        W^{C}(\gamma)
    &:=\frac{1}{|G|}\sum_{R} \sum_{r\in C}  \chi^{*}_{R}(r)W^{R}(\gamma) \label{eq:Wilsonloop2a}\\  
    &=
    \begin{cases}
        1  & \text{if $\prod_{l\in\gamma}g_{l}^{\sigma_l}\in C$}\,,\\
        0  & \text{otherwise}\,,
    \end{cases} 
    \label{eq:Wilsonloop2b} 
    \end{align}
\end{subequations}
%
%
and therefore $\langle W^{C'}(\gamma)\rangle_C = \delta_{C,C'}$ where
$\langle\bullet\rangle_C$ denotes a state with flux $C$ enclosed by $\gamma$;
see \cref{app:wilson} for details.


However, in general there are fluctuations of flux excitations. In this case it
is known from pure gauge theories in $2+1$ dimensions~\cite{Wilson1974} that
the Wilson loop in the trivial phase decays with an area law, whereas in the
topological phase it decays with a perimeter law
\begin{equation}
    \langle W^{R}(\gamma) \rangle  
    \sim e^{- |\gamma|/ \xi_R}\,,
\end{equation}
where $|\gamma|$ denotes the length of $\gamma$.
Hence the expectation values of Wilson loops can be used to probe two
properties: Varying the loop $\gamma$ and verifying the perimeter law
demonstrates the topological character of the phase, and the expectation value
for a fixed loop yields information about the enclosed flux.  

However, for our implementation of quantum doubles via a Hamiltonian $H_{\tilde\G}$
with finite transverse field $\Omega\neq 0$, we cannot necessarily associate a
group element $g_{l}$ to each link, since the full Hilbert space of a link is
much larger than $\H^G_l\equiv\spn{\ket{g}_l\,|\,g\in G}$.
To solve this problem, we modify the Wilson loop operator,
\begin{equation}
    \mathbb{W}^{R}(\gamma) 
    := W^{R}(\gamma)\prod_{l\in \gamma} \mathbb{P}_{l}^{G}
    \,,
    \label{eq:modwilson}
\end{equation}
where $\mathbb{P}_{l}^{G}$ is the projector on link $l$ onto the subspace
$\H^G_l$. This modification guarantees that the Wilson operator
$\mathbb{W}^{R}(\gamma)$ is well defined on the full Hilbert space
$\H_{\tilde\G}$ of the blockade structure $H_{\tilde\G}$.
To understand the effect of this projector, we consider a simplified model
where each link is with probability $p<1$ in a state in $\H^G_l$.
Then the projector $\prod_{l\in \gamma}\mathbb{P}_{l}^{G}$ leads to an
additional contribution $p^{|\gamma|}$ to the perimeter law. This suggests that
the modified Wilson loop \eqref{eq:modwilson} can still distinguish between the
trivial and the topological phase, as this additional factor is consistent with
a perimeter law.
Furthermore, we expect that one can factor out this additional contribution by
evaluating the ratio between two Wilson loops with irreducible representation
$R$ and trivial representation $E$, respectively:
\begin{equation}
    \overline{\mathbb{W}}^{R}(\gamma) 
    := 
    \frac{\langle \mathbb{W}^{R}(\gamma) \rangle}{\langle \mathbb{W}^{E}(\gamma) \rangle} 
    \approx
    \frac{1}{|\mathcal{I}_G|}\sum_{\vec g_\gamma\in\mathcal{I}_G}\chi_R 
    \biggl(\;\prod_{l\in\gamma} g_{l}^{\sigma_l}\biggr)
    \,.
    \label{eq:wilson_corrected}
\end{equation}
Here $\vec g_\gamma\equiv(g_l)_{l\in\gamma}$ denotes measurement outcomes of
one experimental sample for links along the loop $\gamma$. The last equality
shows that it is convenient for an experiment to evaluate the expectation value
by only taking into account the post-selected measurements $\mathcal{I}_G$
where all links along the loop $\gamma$ are in a state in $\H^G_l$.

Finally, we point out that if perturbations of the Hamiltonian $H_{\tilde\G}$
slightly break the local symmetries $\U_p(h)$, charges are also allowed to
fluctuate.  To identify the topological phase in this case, it is necessary to
use Fredenhagen-Marcu order parameters~\cite{Fredenhagen1983,Gregor2011}.

\begin{figure*}[tb]
    \centering
    \includegraphics[width=1.0\linewidth]{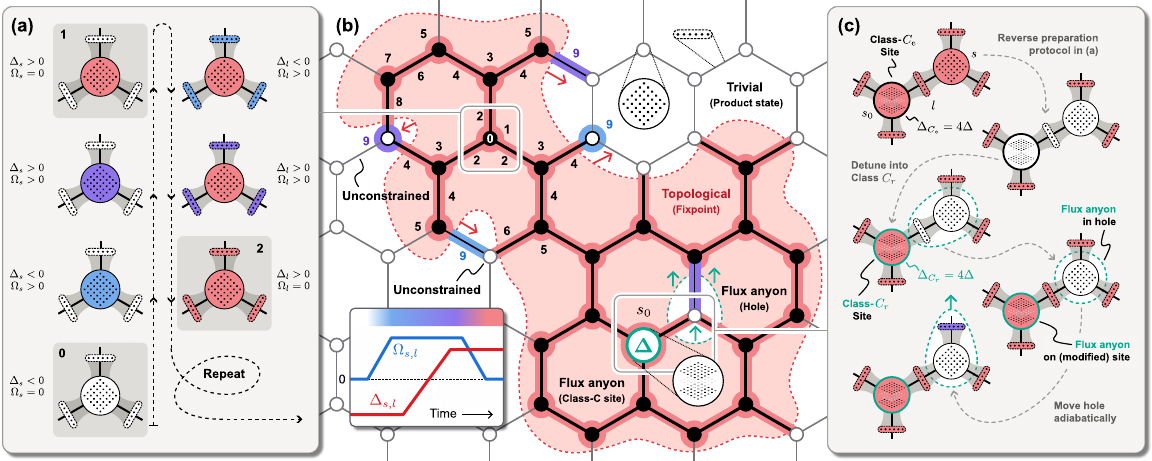}
    \caption{%
        \CaptionMark{Adiabatic preparation of ground states and flux anyons.}
        (a,b) The topological fixpoint ground state can be adiabatically
        ``grown'' starting from a completely de-excited array of two-level
        systems (i.e., $\Delta_s<0$ and $\Delta_l<0$ for all sites and links,
        white filling). One then iterates the ramping procedure (white
        \textrightarrow\;blue \textrightarrow\;purple \textrightarrow\;red)
        shown in the inset of panel (b) for sites and adjacent links
        repeatedly. Note that transverse fields $\Omega_{s,l}$ and detunings
        $\Delta_{s,l}$ are adiabatically modified uniformly for all two-level
        systems on a site or link, thereby respecting the local symmetries of
        the model at all times. A possible (non-optimal) initialization
        sequence is shown in (b) where numbers label time steps. Note that a
        site (link) can only be initialized in the ground state if there is at
        least one adjacent link (site) uninitialized (see four cases in time
        step 9).  Since $\Omega_s=0=\Omega_l$ after initialization (red area),
        the prepared ground state is the fixpoint wave function of the quantum
        double $\tilde H_\G(0,\omega)$, without admixtures from classical
        excited states in $\H^\perp_\G$.
        (c) Once a patch is prepared in the topological ground state, one can
        use the generalized site graph $\G_{s_0}[\bm\Delta]$ to adiabatically
        inject flux anyons into the system [labelled $\bm\Delta$ in panel (b),
        see \cref{sec:fluxlattice} for its construction]. The special site
        $s_0$ is initialized with $\Delta_{C_\nel}=-\Delta \nearrow 4\Delta$
        (and all other $\Delta_{C}=-\Delta=\const$) to prepare the no-flux
        constraint $g_1g_2g_3=\nel$. Then this site, together with an adjacent
        link $l$ and site $s$ are de-excited again [following the inverse
        protocol in panel (a)]. Subsequently, the special site is
        re-initialized with $\Delta_{C_r}=-\Delta\nearrow 4\Delta$ (green
        boundary, red filling) so that it enforces the constraint $g_1g_2g_3\in
        C_r$ for some non-trivial conjugacy class $C_r$. This protocol prepares
        two flux anyons in the vacuum fusion channel: $[C_r,E]$ localized on
        the ``flux factory'' site $s_0$, and the corresponding antiparticle
        $[\anti C_r,E]$ in the hole of de-excited sites and links. The hole
        (carrying its anyon) can then be adiabatically moved by sequences of
        (de-/re-)initializations of links and sites (purple link).
    }
    \label{fig:fig5}
\end{figure*}

\section{Adiabatic state preparation and anyon braiding}   
\label{sec:braiding}

A remarkable property of the Hamiltonian $H_\G$ (and $H_{\tilde\G}$)
[\cref{eq:H}] is that it can be generalized to \emph{local} transverse fields
$\Omega_s$ and $\Omega_l$ on sites $s$ and links $l$ without violating the
local plaquette symmetries $\U_p(h)$. 
The only constraint is that $\Omega_s$ ($\Omega_l$) and $\Delta_s$ ($\Delta_l$)
are equal for all two-level systems that are permuted by these local
automorphisms, i.e., transverse fields and detunings must be uniform on each
site $s$ and link $l$ respectively, but can vary between sites and links.
This allows us to locally control the system with time dependent parameters
$\Omega_s(t)$ and $\Delta_s(t)$ on sites [$\bm\Delta_s(t)$ for generalized
sites], and $\Omega_l(t)$ and $\Delta_l(t)$ on links. 
In a first step, we leverage this control to propose a protocol for the
adiabatic preparation of the flux-free topological quantum many-body ground
state of $H_\G$. Later, we extend this protocol to realize controlled braiding
of flux anyons using $H_{\tilde\G}$. 

For the adiabatic preparation of the ground state of $H_\G$ on a patch of the
honeycomb lattice (\cref{fig:fig5}), we start with $\Omega_i=0$ and $\Delta_i=
-\Delta <0$ on all vertices, and prepare these two-level systems in the unique
ground state $\ket{\psi_0} = \bigotimes_{i\in V} \ket{0}_i$. This state is
obviously in the symmetric sector, $\ket{\psi_0}\in\H^S_\G$, so that the
symmetry sector is completely fixed by this initial state. 
The main idea for an efficient adiabatic preparation is to \emph{grow} the
topological phase, starting from a single site $s$. On this site, we first
adiabatically turn on the transverse field $\Omega_s \sim \Delta$, then ramp
the detuning from $\Delta_s=-\Delta < 0$ to its final value $\Delta_s= 4\Delta
> 0$, and finally adiabatically turn of the transverse field again, see
\cref{fig:fig5}~(a).
Due to the strong blockade interaction on the site, only a single vertex can be
excited to $\ket{1}$, but since $\Omega_s$ acts uniformly on all vertices on
$s$, this protocol results in an equal-weight superposition of all possible
single-excitation states:  
\begin{equation}
    \ket{\psi_1}
    = \mathcal{N}\;
    \biggl[\bigotimes_{i \notin s}\ket{0}_i\biggr]
    \biggl[\;
    \sum_{g_1,g_2 \in G} \ket{w_s^{g_1g_2}}
    \biggr]\,,
\end{equation}
with normalization $\mathcal{N} = 1/|G|$. Recall that $\ket{w_s^{g_1g_2}}$
denotes the state with vertex $w_{s}^{g_1g_2}$ on site $s$ excited to $\ket{1}$
and all other vertices on the site in state $\ket{0}$. 
During this adiabatic ramping procedure, the system always exhibits a gap of
order $\max\{|\Delta_s|,\Omega_s|G|\}$. Note that the transverse field exhibits
a collective enhancement due to the blockade interaction, so that for optimal
ramping $\Omega_s \sim \Delta/|G|$ and the preparation can be achieved on the
time scale $\hbar /\Delta$. 

In the next step, we proceed to all links connected to site $s$ and repeat the
adiabatic ramping procedure with $\Omega_l \sim \Delta$ and final value
$\Delta_l=2\Delta$. Due to the blockade interactions between the vertices on
the links and the excited vertex on the site, one vertex on each link is in
blockade while the others can be efficiently adiabatically excited on the time
scale $\hbar/\Delta$. At the end of this procedure, the new ground state is
\begin{equation}
    \ket{\psi_2} 
    = \mathcal{N}\;
    \biggl[\bigotimes_{i\notin s,l}\ket{0}_i\biggr]
    \biggl[\;
    \sum_{g_1,g_2\in G} \ket{g_1,g_2,g_3}\ket{w_s^{g_1g_2}}
    \biggr]\,,
\end{equation}
where the states $\ket{g_1,g_2,g_3}$ on the three links $l\equiv l_{1,2,3}$
obey the no-flux condition $g_1g_2g_3=\nel$ [\cref{fig:fig5}~(a)].

The next step is to repeat the adiabatic ramping on all sites connected to the
links $l_1$, $l_2$, and $l_3$, and then proceed with all links connected to
these sites, etc. This iterative ramping protocol, alternating between links
and sites, grows the topological phase from the inside of the patch towards its
boundary, see~\cref{fig:fig5}~(b). 
In each step, the blockade interactions between links and sites constrains the
excitation patterns to the ground state manifold $\H^0_\G$. While there is much
freedom in the sequence of sites and links that are passed by the ramping
procedure [\cref{fig:fig5}~(b)], it is important that for each link only
\emph{one} connected site has already been excited, and for each site \emph{at
most two} connected links have already been excited. This guarantees that the
constraints imposed by the blockade interaction can always be fulfilled. 
Note that on sites with only \emph{one} connected link already excited, there
is still a collective enhanced coupling $\sqrt{|G|}\Omega_s$, whereas on sites
with two excited links there is no collective enhancement. But with proper
local addressing, this can be compensated by the strength of the transverse
field $\Omega_s$, such that each step can be implemented on a time scale
$\sim\hbar/\Delta$.

This protocol prepares a unique state on an open patch of the honeycomb lattice
and respects all local symmetries, i.e., the wave function $\ket{\psi_t}$
during the adiabatic ramping always satisfies $\U_{p}(h) \ket{\psi_t} =
\ket{\psi_t}$ and remains in the symmetric sector, $\ket{\psi_t}\in\H^S_\G$. It
is convenient to stop the preparation of a finite patch such that every
initialized site has all three emanating links initialized as well, i.e., the
patch has ``rough'' boundaries with dangling links. Then, the protocol prepares
the exact and unique ground state of $\tilde H_\G(0,\omega)$, i.e., the
equal-weight superposition of all configurations satisfying the zero-flux
constraint on every initialized site. This state corresponds to the unique
ground state of the quantum double model \eqref{eq:QD} for ``rough'' boundary
conditions. Note that after this initialization, it is still possible to ramp
up a homogeneous transverse field $\Omega$ to prepare the true ground state
$\ket{\Omega}$ of the Hamiltonian $H_\G$ on a time scale $\sim \hbar/\Delta$
due to the excitation gap in the symmetric sector.

In summary, the adiabatic preparation of the topological state $\ket{\Omega}$
can be achieved on a time scale $\tau\sim 2\sqrt{A}\,\hbar/\Delta$ with $A$ the
total number of sites in the patch of the honeycomb lattice. 
As required for the local unitary preparation of a state with topological order
from a trivial product state, the time for the preparation scheme scales with
the system size (but only with a $\sqrt{A}$ scaling)~\cite{Chen2010}.
Note that if the adiabatic ramping scheme \emph{violates} the local symmetries,
the ramping must be slower than the gap \emph{between} different symmetry
sectors to avoid the excitation of charge anyons. In particular, since the
charge gap is now essential for the adiabatic preparation, the transverse
fields $\Omega_s$ and $\Omega_l$ can no longer be switched off to prepare the
fixpoint ground state. Fortunately, this does not alter the overall scaling of
the preparation time with system size.

The protocol for adiabatic ground state preparation can be generalized to
states with well-defined flux anyons (ground states of $H_{\tilde\G}$), and the
subsequent adiabatic braiding of these anyons. 
Flux anyons can be either trapped inside holes, i.e., contiguous areas of sites
and links with all vertices in state $\ket{0}$, or pinned to generalized site
graphs $\G_s[\bm\Delta_s]$ (introduced in \cref{sec:fluxlattice}), where the
site-specific chemical potentials $\bm\Delta_s$ can be used to control the flux
anyon type pinned at this site.
To prepare well-defined flux anyons with the above method -- which is based on
the interplay of adiabatic ramping and blockade interactions -- one needs at
least one special site $s_0$ with the generalized blockade graph
$\G_{s_0}[\bm\Delta]$, detuned by $\bm\Delta=\{\Delta_C\}$. This site plays the
role of a ``flux factory'' to adiabatically inject fluxes into the system,
which subsequently can be trapped and moved inside holes (which does not
require modified sites), see \cref{fig:fig5}~(b).
The preparation starts with the initialization of the topological ground state
in the symmetric sector $\H^S_{\tilde\G}$, with all sites in the zero-flux
state.  For the generalized site $s_0$ this means
$\Delta_{C_\nel}=-\Delta\nearrow 4\Delta$ and $\Delta_C=-\Delta=\const$ for all
$C\neq C_\nel$, \cref{fig:fig5}~(c).
To inject flux anyons into the system, one selects the special site $s_0$, an
adjacent site $s$, and the link $l$ connecting them, and applies the inverse
adiabatic ramping procedure to bring all vertices on this link and the two
sites into state $\ket{0}$. This creates a hole encompassing the two sites,
with detunings $\Delta_C=-\Delta$ for all conjugacy classes $C$ on the modified
vertex $s_0$.
To create a flux anyon $[C_r,E]$ and its antiparticle $[\anti C_r,E]$ on the
two sites, one adiabatically ramps the transverse field $\Omega_{s_0}$ on all
site vertices, and subsequently the detunings
$\Delta_{C_r}=-\Delta\,\nearrow\,4\Delta$ of the site vertices that belong to
class $C_r$ (for $C\neq C_r$ the detunings $\Delta_C=-\Delta$ remain constant).
Finally, one also performs the ramping on the link to site $s$. This prepares a
state on the three adjacent links that satisfies the generalized condition $g_1
g_2 g_3 \in C_r$.
This step creates \emph{two} flux anyons: $[C_r,E]$ pinned at the special site
$s_0$ as the lowest energy state, and $[\anti C_r,E]$ trapped inside the hole
at the neighboring site $s$. Note that the hole necessarily carries the
anti-flux anyon $[\anti C_r,E]$ since the surrounding bulk state requires that
these two anyons fuse into the vacuum. 
Furthermore, the procedure prepares a well-defined state in the fusion space of
these two anyons, namely the unique topological state of two anyons $[C_r,E]$
and $[\anti C_r,E]$ in the vacuum fusion channel. 
After this initial creation of a flux pair, the hole can be adiabatically moved
around by ramping down a connecting link and an adjacent site, and subsequently
ramping up the original site and the connecting link again,
\cref{fig:fig5}~(b,c). In a similar fashion, anyons pinned at site $s_0$ can
first be immersed into a hole and subsequently moved away. Then the ``flux
factory'' $s_0$ can be reused for the creation of the next pair of flux anyons.
Finally, the Wilson loop operators $W^R(\gamma)$ [or $W^C(\gamma)$] for loops
$\gamma$ around flux-carrying holes can be used to measure the enclosed total
flux, which probes the fusion channel of the encircled holes. 

In summary, we have developed a complete toolbox to explore the non-abelian
character of flux anyons in quantum double models $\double{G}$: (i) adiabatic
ground state preparation, (ii) deterministic and adiabatic creation of flux
anyons in a well-defined fusion channel, (iii) adiabatic transport of these
anyons (necessary for braiding and fusion), and finally, (iv) probing of fusion
channels by measuring Wilson loop operators around flux anyons. In the
remainder of this paper, we apply this toolbox to the simplest non-abelian
quantum double $\double{S_3}$.

\begin{figure*}[tb]
    \centering
    \includegraphics[width=1.0\linewidth]{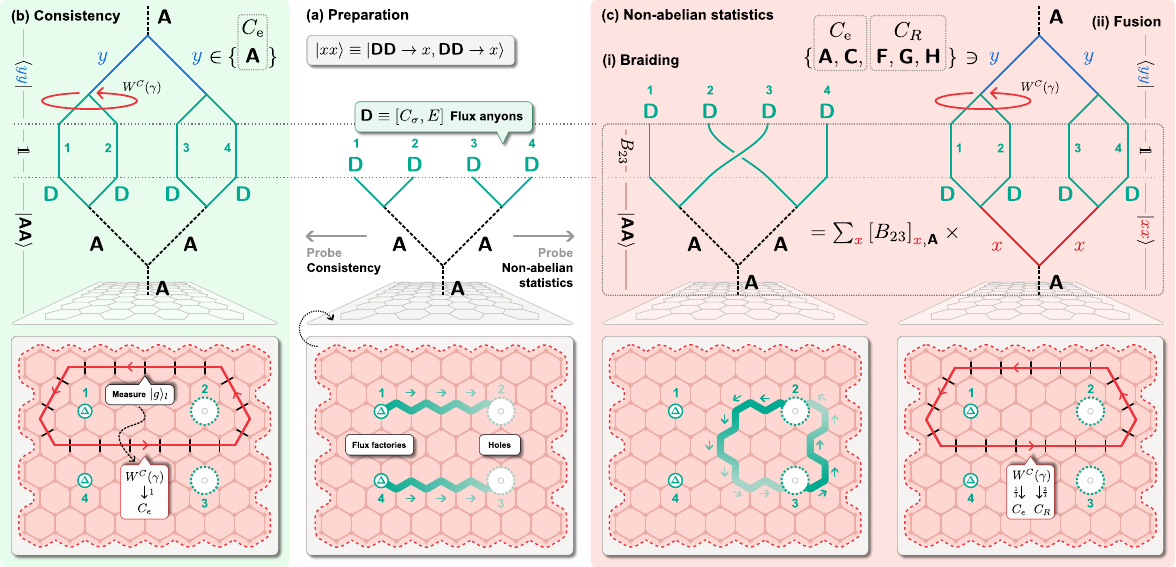}
    \caption{%
        \CaptionMark{Probing non-abelian statistics.}
        Protocol for the preparation, braiding and measurement (fusion) of
        anyons. The splitting/fusion diagrams (time evolution) are shown in the
        top row, the corresponding spatial configurations and manipulations in
        the bottom row.
        (a) As an initial step, two pairs $(1,2)$ and $(3,4)$ of $\anyon
        D=[C_\sigma,E]$ flux anyons (solid green lines) are created and
        separated within a topological domain (red area) following the protocol
        described in \cref{sec:braiding} and \cref{fig:fig5}~(c). We denote
        splitting/fusion states of this form as
        $\ket{xx}\equiv\ket{\anyon{DD}\rightarrow x,\anyon{DD}\rightarrow x}$.
        With this nomenclature, the system is initialized in the state
        $\ket{\anyon{AA}}$ where both pairs fuse into the vacuum $\anyon A$
        (black dashed lines).
        (b) As a consistency check, the Wilson loop $W^C(\gamma)$ can be
        computed from measuring the link states $\ket{g}_l$ along the depicted
        dual loop $\gamma$ (solid red line). Only for $C=C_\nel$ the
        expectation value should be non-zero since the enclosed anyons are in
        the vacuum channel.
        (c) Alternatively, the two anyons 2 and 3 from \emph{different} pairs
        can be exchanged with a half-braid (again using the protocol for
        adiabatically moving holes) to produce the new state
        $B_{23}\ket{\anyon{AA}}$ shown in (c-i). Here, $B_{23}$ denotes the
        unitary braiding matrix for exchanging flux 2 and 3. Using the $F$- and
        $R$- matrices of the unitary braided fusion category that describes the
        quantum double $\mathcal{D}(S_3)$ (see \cref{app:braiding}), this state
        can be expanded in the basis $\ket{xx}$, where the fusion rules allow
        $x\in\{\anyon{A},\anyon{C},\anyon{F},\anyon{G},\anyon{H}\}$. Here,
        $\{\anyon{A},\anyon{C}\}$ carry no flux ($C_\nel$) and
        $\{\anyon{F},\anyon{G},\anyon{H}\}$ carry non-zero flux of type $C_R$.
        These fluxes can be measured again by the same Wilson loop
        $W^C(\gamma)$, where now the expectation value for $C=C_R$ is finite.
        This demonstrates the non-abelian nature of $\mathcal{D}(S_3)$ in that
        the two states depicted in the lower left and right corners are locally
        indistinguishable while being linearly independent.
    }
    \label{fig:fig6}
\end{figure*}

\section{Examples for $\double{S_3}$}
\label{sec:anyons}

For the abelian group $G=\Z_2$, the construction of the Hamiltonian $H_\G$ an
its corresponding graph $\G$ reproduces the blockade structure studied in
Ref.~\cite{Maier2025}, which leads to an \emph{abelian} topological phase
known as the toric code~\cite{Kitaev2003}. 
Therefore we focus in the following on the simplest \emph{non-abelian} quantum
double derived from the permutation group $G=S_3\equiv C_{3v}$ with six
elements $\{\nel,R,R^2,\sigma,\sigma R,\sigma R^2\}$ with $R$ a three cycle and
$\sigma$ a two cycle ($\sigma^2=\nel$, $R^3=\nel$ and $\sigma R=R^2\sigma$).
Thus we have $|G|=N=6$, so that on each link there are six two-level systems
and on each site there are $36$ two-level systems; this setup is sketched in
\cref{fig:fig1}~(b).  
The quantum double $\double{S_3}$ features eight anyon types, given by the
irreducible representation of the Drinfeld double
$\double{S_3}$~\cite{Kitaev2003,Simon2023}. In particular, this implies an
eight-fold ground state degeneracy on a torus. As discussed in
\cref{sec:fluxlattice}, the anyons can be labeled by a conjugacy class $C$ of
the group $G$ and an irreducible representation of the centralizer of a
representative of the conjugacy class. The group $S_3$ has three conjugacy
classes $C_\nel$, $C_\sigma$, and $C_R$, and three irreducible representations:
the trivial representation $E$, a one dimensional representation $\Gamma_1$,
and a two-dimensional representation $\Gamma_2$. 
Taken together, these label the pure flux anyons and the pure charge anyons.
Following the standard notation, these are denoted as $\anyon
A\equiv[C_\nel,E]$ for the vacuum, $\anyon D\equiv [C_\sigma,E]$ and $\anyon
F\equiv [C_R,E]$ for the flux anyons, and $\anyon B\equiv [C_\nel,\Gamma_1]$,
and $\anyon C\equiv [C_\nel,\Gamma_2]$ for the charge anyons.  
In addition, there are three dyons: $\anyon E\equiv[C_\sigma,\Gamma_\sigma]$
for the non-trivial irreducible representation of the centralizer of
$C_\sigma$, and $\anyon G\equiv[C_R,\Gamma_{R_1}]$ and $\anyon
H\equiv[C_R,\Gamma_{R_2}]$ for the two non-trivial irreducible representations
of the centralizer of $C_R$. A full review of the anyon content of
$\double{S_3}$, their fusion channels, $F$-matrices and $R$-matrices can be
found in Ref.~\cite{Cui2015}.

Drawing from the toolbox developed above, we now describe a simple braiding
scheme for flux anyons to probe their non-abelian statistics, see
\cref{fig:fig6}. For this, we start by preparing a setup with four $\anyon
D=[C_\sigma,E]$ anyons using (ideally) two ``flux factories'' as explained in
\cref{sec:braiding} and illustrated in \cref{fig:fig6}~(a).
As discussed previously, the initial state is therefore in the fusion channel
where the first and second pair of $\anyon D$ anyons each fuse into the vacuum.
It is convenient to define a basis of the fusion space of four $\anyon D$
anyons $\H^\anyon{DDDD}_\anyon{A}$ that are in the global vacuum channel
$\anyon A$. We denote as $\ket{xx}\equiv\ket{\anyon{DD}\rightarrow
x,\anyon{DD}\rightarrow x}$ the fusion state where each pair fuses into anyon
$x$; note that both pairs must fuse into the same anyon since the fusion of all
four anyons yields the vacuum $\anyon A$ and for $\double{S_3}$ all anyons are
their own antiparticle.
The fusion rule~\cite{Cui2015}
\begin{align}
    \anyon D\otimes\anyon D= \anyon A\oplus \anyon C \oplus \anyon F\oplus \anyon G \oplus \anyon H
\end{align}
then determines a basis of the five-dimensional fusion space, namely
\begin{align}
    \H^\anyon{DDDD}_\anyon{A}
    =\operatorname{span}\left\{
        \ket{\anyon{AA}},\ket{\anyon{CC}},\ket{\anyon{FF}},\ket{\anyon{GG}},\ket{\anyon{HH}}
    \right\}\,,
\end{align}
where our system is initialized in the state $\ket{\anyon{AA}}$.  

Consequently, a measurement of the Wilson loop $W^C(\gamma)$ [via
post-selection, recall \cref{eq:wilson_corrected}] around the first (and
second) pair of anyons must yield the flux $C_\nel$ with probability 1, which
can be used to probe the consistency of the adiabatic preparation scheme,
\cref{fig:fig6}~(b).
Now we can perform braiding using the adiabatic ramping protocol from
\cref{sec:braiding}, see \cref{fig:fig6}~(c). In general, braiding anyons
induces unitary transformations on the fusion space, here
$\H^\anyon{DDDD}_\anyon{A}$. If we braid the first and second anyon around each
other, the fusion state $\ket{\anyon{AA}}$ remains invariant -- which can again
be tested by measuring Wilson loop operators. By contrast, if we exchange
(``half-braid'') the second and the third anyon, the initial state
$\ket{\anyon{AA}}$ transforms into (see \cref{app:braiding} and
Ref.~\cite{Cui2015})
\begin{align}
    \ket{\anyon{AA}}
    \;\mapsto\;
    \frac{1}{3}
    \Big[
        &\overbrace{%
            \ket{\anyon{AA}}+\sqrt{2}\ket{\anyon{CC}}
        }^{C_\nel} 
        \label{eq:braiding}\\  
        +\,&\sqrt{2}
        \big(
            \underbrace{%
                \ket{\anyon{FF}}+e^{i\frac{2\pi}{3}}\ket{\anyon{GG}}+e^{-i\frac{2\pi}{3}}\ket{\anyon{HH}}
            }_{C_{R}}
        \big)
    \Big]\,.
    \nonumber
\end{align}
As before, this state can be probed by measuring the Wilson loop around the
first two anyons. With probability $1/3$ one measures again the trivial flux
$C_\nel$, but with probability $2/3$ one now finds the non-trivial flux $C_R$. 
Hence this simple braiding protocol already reveals the non-abelian character
of the quantum double $\double{S_3}$: by adiabatically exchanging two identical
anyons trapped in two holes one can change the fusion channel of the system
without ever leaving the ground state manifold.

\section{Conclusion and Outlook}


We introduced and studied a family of two-dimensional models, constructed from
two-level systems subject to local transverse fields and detunings, where
excited states interact via a strong blockade interaction. These models are
motivated by the Rydberg platform with neutral atoms in optical tweezers. On an
abstract level, the Hamiltonians can be described by vertex-weighted blockade
graphs $\G$, with vertices representing two-level systems and edges blockade
interactions. 
We presented a construction for a family of blockade graphs, such that for
every finite group $G$, the ground state of the blockade Hamiltonian with weak
transverse fields is in the topologically ordered phase of the quantum double
model $\double{G}$. This family of topological phases is characterized by
anyonic flux and charge excitations which exhibit non-abelian statistics for
non-abelian groups $G$.
We proved the emergence of topological order in the many-body ground state
analytically. In an upcoming paper, we show the existence of a finite
excitation gap for the special case $G=\mathbb{Z}_2$ \cite{Gap2025}. The core
idea of our construction is to enforce the no-flux constraint in the ground
state by tailored blockade interactions such that the blockade graph exhibits
local graph automorphisms. These automorphisms translate to local symmetries of
the Hamiltonian, and the symmetric sector corresponds to the zero-charge sector
of the corresponding quantum double model. 
In this framework, we developed a complete toolbox to explore the non-abelian
character of flux anyons. This includes (i) efficient protocols for the
adiabatic preparation of ground states, (ii) deterministic and adiabatic
preparation schemes of flux anyons in a well-defined fusion channel, (iii) a
protocol for the adiabatic motion of these anyons (needed for braiding and
fusion), and finally, (iv) a procedure to probe the fusion channel of anyons by
measuring Wilson loops around them. Combined, these tools pave the way towards
probing non-abelian topological phases in artificial matter based on realistic
two-body interactions.


In this paper, both the construction of the blockade graph $\G$ and the
development of the toolbox were illustrated on the trivalent honeycomb lattice
as this is the simplest setting to discuss quantum doubles. Note that this is
not necessary, and the construction can be straightforwardly generalized to
arbitrary lattices and even irregular planar graphs, in accordance with
Kitaev's original formulation of quantum doubles~\cite{Kitaev2003}. 

Another generalization concerns the choice of detunings. In our construction of
$\G$, the detunings on sites and links where chosen such that $\Delta_s =
2\Delta_l$, combined with a sufficiently large blockade interaction $U_0 >
\Delta_s$. These choices are not unique. For example, it is easy to see that
the classical ground state manifold remains unchanged for $\Delta_s>2\Delta_l$
since larger $\Delta_s$ only stabilize the no-flux constraint. An interesting
open question is how the phase diagram is affected by variations of these
parameters.

While our framework is motivated by the Rydberg platform, we stress that our
abstract analysis omits the influence of microscopic van der Waals
interactions. To study their effects, a concrete \emph{embedding} of the
blockade graphs in two or three dimensions would be necessary. For the special
case $G=\mathbb{Z}_2$ (toric code topological order), an explicit embedding of
the corresponding blockade graph $\G$ was provided in Ref.~\cite{Maier2025}. It
is an interesting open question whether the proposed blockade graphs for
general groups $G$ can be embedded as well, and if so, how to achieve this most
efficiently.

Alternatively, the proposed models could be realized on other platforms.
Especially cold polar molecules in optical tweezers have recently seen
significant progress, with the potential advantage of much longer lifetimes of
excited states~\cite{Ruttley2025} and a high tunability of interaction
potentials~\cite{Micheli2006}. 
On the other hand, superconducting qubits that are connected by microwave
cavities -- which act as a bus to mediated blockade
interactions~\cite{Blais2021} -- have the potential to realize blockade graphs
with far less restrictions on geometry. Such ideas can naturally be extended to
Rydberg atoms in optical cavities that are connected by wave guides. On these
platforms, the realization of the blockade graph directly translates to a wave
guide structure, and therefore becomes a straightforward engineering task. 

We close with a comment on an intriguing though abstract problem. Our
construction leads to graphs with local automorphisms that, under certain
circumstances, generate a single orbit on the set of maximum-weight independent
sets. While there are trivial graphs which satisfy this condition, the
maximum-weight independent sets of our models have an intricate structure due
to the local constraints; in particular, they cannot be ``factorized'' (the
corresponding low-energy Hilbert space has no local tensor product structure).
This raises the question how graphs with these properties can be classified
and/or systematically constructed. This could be useful as each such graph
might give rise to an interesting quantum many-body phase. For example, it
would be interesting to explore whether there are graphs that stabilize the
topological order of Fibonacci anyons, the simplest anyon model universal for
topological quantum computation. If this turns out to be \emph{impossible}, it
would be helpful to understand \emph{why} to sharpen our
understanding of the limitations of blockade structures.



\begin{acknowledgments}
    We thank Jean-No\"el Fuchs for his introduction to quantum double models
    and Spyridon Michalakis for providing clarifications about the gap
    stability theorem.
\end{acknowledgments}


\bibliographystyle{./bib/bibstyle.bst}
\bibliography{bib/bibliography}

\clearpage

\appendix

\noindent\textbf{Important:} The technical nature of these appendices demands
for a streamlined notation which does not always match the notation of the main
text. We emphasize these differences when necessary.

\section{The graph for one site}%
\label{sec:aut_single}

In this appendix, we rigorously construct and discuss a generalized graph $\GC$
for a ``class-$C$ site'' with three emanating links. This site graph enforces
the flux condition $g_1 g_2 g_3 \in C$ for a conjugacy class $C$ of the group
$G$. $\GC$ corresponds to one part of the generalized site graph
$\G_s[\bm\Delta]$ introduced in \cref{sec:fluxlattice} and used in
\cref{fig:fig5}~(c), where $\Delta_C=4\Delta$ for one fixed class and all other
site vertices that do not belong to this class are omitted. At the end, we
consider the special case $C = \{1\}$ where we recover the graph $\GG$ (labeled
$\G_s$ in the main text) from \cref{fig:fig1}~(b) for a site with the zero-flux
condition $g_1 g_2 g_3 = 1$. Note that throughout this appendix, we refer to
the neutral element of a group as $1$ (in the main text we use $\nel$ instead).

In \cref{sec:GraphDefinition} we construct the graph $\GC$ and lay out the
notation. Subsequently, in \cref{sec:locaut}, we discuss its (local) graph
automorphisms and the structure of its automorphism group. Specifically for $C
= \{1\}$, we show that $\text{Aut}_{\text{loc}}(\GG) \cong G^2 \rtimes \aut{G}$
for local graph automorphisms, as claimed in \cref{sec:model}. Then we discuss
its maximum-weight independent sets in \cref{sec:MIS} and show that the graph
is fully-symmetric in \cref{sec:FullSymm}. Finally, in
\cref{sec:MultiClassSite}, we construct and discuss the ``multi-class graph''
$\GCl$ from \cref{sec:fluxlattice} for a site that incorporates all conjugacy
classes (labeled $\G_s[\bm\Delta]$ in the main text).

\subsection{Definition of the graph $\GC$}%
\label{sec:GraphDefinition}


We consider a generic group $G$ of order $N = |G|$ and a conjugacy class $C
\subseteq G$ of size $M = |C|$. In this subsection, we construct the
vertex-weighted graph $\GC = (\VC, \EC, \DC)$ for a class $C$-site $s$ with
three emanating links $l \in \ZD := \{1, 2, 3\}$. The labels $l$ fix an
ordering for the links. [In the tessellated structure in \cref{sec:torus}, 
this ordering alternates between adjacent sites of the two sublattices.
For just one site in this section, without any embedding in a larger graph
these are just labels used for the construction.]
Note the difference in notation compared to the main
text, where we used the variables $\{l_1, l_2, l_3\}$. An exemplary construction
with $N = 6$ and $M = 3$ is shown in \cref{fig:ClassSite}.

We start with a fully-connected (= complete) graph $K_{N^2M} = (\Vs^C, \Es^C,
\Ds^C)$ with vertices $\Vs^C = G^2 \times C$ on the site. In this notation,
each site-vertex is identified with a triple of group elements. As $K_{N^2M}$
is fully-connected, the edges $\Es^C = \{\, \{v_s, w_s\} \: | \;
v_s\neq w_s \in \Vs^C \,\}$ on the site include each unordered
pair of vertices. This is illustrated by the thick green circle in
\cref{fig:ClassSite}. We choose a uniform weight $\Delta_v = 4\Delta$ for
the vertices $v \in \Vs^C$.
	
Next, we construct the full graph $\GC$ as a graph extension of $K_{N^2M}$.
The full set of vertices $\VC = \Vs^C \cup \bigcup_{l \in \ZD} \Vi^C$ consists
of the vertices on the site and additionally includes the vertices $\Vi^C =
\{l\} \times G$ for each link. That is, the full graph consists of $N^2M + 3N$
vertices.  In this notation, each link-vertex is identified with the label of
its link and a group element from $G$. For the vertices $v \in \ZD \times G$ on
the links, we choose again uniform weights $\Delta_v = 1\Delta$. These vertices
are fully-disconnected, meaning no two vertices on the links are connected.
This is illustrated by the dashed black boxes in \cref{fig:ClassSite}.

Finally, we define the edges between the vertices on the links and the site.
To this end, we introduce the compact notation
\begin{align}
    \begin{split}
    &g_3 \equiv (g_1g_2)^{-1}c
    \quad\text{for}\quad 
    w \equiv (g_1, g_2, c) \in \Vs^C\,,
    \\
    &v_l \equiv (l,h_l) \in \Vi^C\,.
    \end{split}
    \label{eq:ShortNotation1}
\end{align}
Using this notation, we can write
\begin{align}
	\Ei^C &= \{\, \{w, v_l\} \,\vert\, w \in \Vs^C, \, v_l \in \Vi^C, \; g_l = h_l \,\}
\end{align} 
%
for the set of edges between site-vertices and vertices of link $l$. These
edges are drawn as the thin black lines in \cref{fig:ClassSite}.
The full set of edges of $\GC$ is then given by
\begin{align}
    \EC = \Es^C \cup \bigcup\nolimits_{l \in \ZD} \Ei^C.
\end{align}
This fully defines the graph $\GC$.

We conclude our construction with some remarks.
Firstly, note that each vertex $(g_1, g_2, c) \in \Vs^C$ on the site is connected to the
three vertices $(l, g_l) \in \Vi^C$ on the links, such that their group
elements fulfill the \emph{site constraint}
\begin{align}
    g_1 g_2 g_3 \in C\,.
    \label{eq:VertexCond}
\end{align}
Thus the site constraint writes the group structure of $G$ 
into the edges of the graph $\GC$.
This makes the constraint \cref{eq:VertexCond} central to the construction. 

Secondly, note that for $C = \{1\}$, the above construction corresponds exactly to
the graph defined in \cref{fig:fig1}~(b). In the main text, we denoted the
vertices on site $s$ as $w_s^{g_1\,g_2\,c}$ and the vertices on the link $l \in
\ZD$ as $v_l^{h_l}$ with group elements $g_1, g_2, h_l \in G$ and $c \in C$. 
This notation can be interpreted as bijective maps 
\begin{subequations}
    \begin{alignat}{5}
        w_s&: \; G^2 \times C \,&&\rightarrow\, \Vs^C&&: \; (g_1, g_2,c) &&\mapsto w_s^{g_1\,g_2\,c}\,,
        \\
        v_l&: \; \{l\} \times G \,&&\rightarrow\, \Vi^C&&: \; (l,h_l) &&\mapsto v_l^{h_l}\,.
    \end{alignat}
\end{subequations}
In the above construction, we used these bijections to identify the vertex sets
$\Vs^C \,\widehat{=}\, G^2 \times C$ and $\Vi^C \,\widehat{=}\, \{l\} \times
G$. This allows in the following for a more concise notation.

Finally, we note that this construction is not limited to the case of $|\ZD| =
3$ links on the honeycomb lattice. For $n > 3$ links this construction remains
well-defined for then $N^{n - 1}M$ vertices on the site.

\begin{figure*}
    \includegraphics[width=0.75\linewidth]{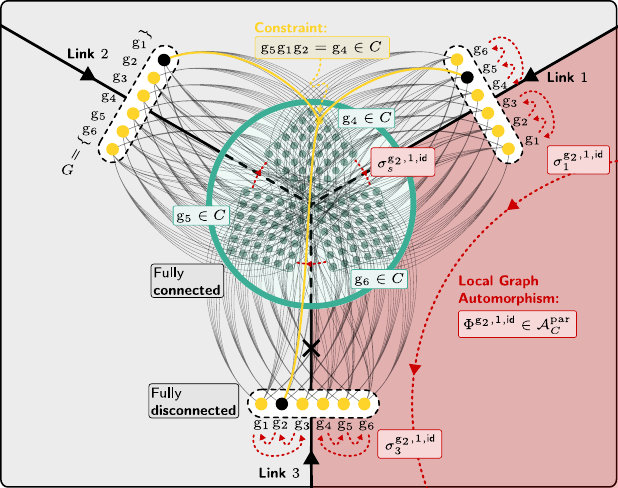}
    \caption{%
    \CaptionMark{Construction and automorphisms of $\GC$.}
    Construction of a generalized site graph $\GC$ with three emanating links
    $l \in \{1, 2, 3\}$ for an arbitrary conjugacy class $C$ of the group
    $G$. This construction extends the condition for the group multiplication
    on this site to $g_1g_2g_3 \in C$, where $g_l$ is the group element on link
    $l$. Italic symbols $g_l$ label dummy variables and roman symbols
    $\mathrm{g}_l$ label specific group elements. Illustrated is an example for
    the group $G = S_3$ with $N = |G| = 6$ and elements $(\mathrm{g}_1,
    \mathrm{g}_2, \mathrm{g}_3, \mathrm{g}_4, \mathrm{g}_5, \mathrm{g}_6) = (1,
    R, R^2, \sigma, R\sigma, R^2\sigma)$. As conjugacy class, we consider $C =
    \{\sigma, R\sigma, R^2\sigma\}$ with $M = |C| = 3$. The $N^2M$ vertices on
    the site (green points and one yellow) are fully connected (blockades
    not shown). On each link there are $N$ vertices (black and yellow points)
    that are fully disconnected. Each vertex on the site is connected via an
    edge (thin black lines) to one vertex on each link. Each link its assigned
    an inward-pointing orientation (black arrow). We mark the link $l = 2$
    (black \textsf{X}) to fix the group multiplication order $g_1g_2g_3$ of the
    group elements $g_l$ on the links, as well as the labeling $(g_1, g_2, c)
    \in G^2 \times C$ of vertices on the site. (It is important to stress that
    the constructed graph is still rotation symmetric.) The exemplary group
    multiplication $\mathrm{g}_5\mathrm{g}_1\mathrm{g}_2 = \mathrm{g}_4$, which
    satisfies the constraint for this site, is highlighted by yellow edges. The
    yellow vertices on the site and on the links are the corresponding
    maximum-weight independent set (MWIS) $M_w$ with $w =
    (\mathrm{g}_5,\mathrm{g}_1,\mathrm{g}_4)$. Note that this MWIS includes all
    but one vertex on each link, which associates the group element to the
    link. Also shown is a specific local graph automorphism
    $\Phi^{\mathrm{g_2},1,\Id} \in \AC$ that corresponds to a plaquette
    automorphism $\Theta_p(\mathrm{g}_2)$ on the lower-right plaquette (dashed
    red lines). It splits into the group permutations (short dashed red arrows)
    $\sigma^{\mathrm{g_2},1,\Id}_{0}$,
    $\sigma^{\mathrm{g_2},1,\Id}_{2}$ and
    $\sigma^{\mathrm{g_2},1,\Id}_{s}$ on link $0$, link $2$, and site
    $s$, respectively. This specific automorphisms leaves the vertices of link
    $1$ invariant.
    }
    \label{fig:ClassSite}
\end{figure*}

\subsection{Graph automorphisms of $\GC$}%
\label{sec:locaut}

A \emph{graph automorphism} of a vertex-weighted graph $\mathcal{G}=(V,E,W)$
is a bijection $\Phi\in\text{Sym}(V)$ on the vertices that preserves \ldots
\begin{enumerate}

    \item the connectivity,
        \begin{align}
            \{v, w\} \in E \;\Leftrightarrow\; \{\Phi(v), \Phi(w)\} \in E\,,
            \label{eq:Connectivity}
        \end{align}

    \item and the vertex weights,
        \begin{align}
            \forall v \in V: \; \Delta_{\Phi(v)} = \Delta_v\,.
            \label{eq:VertexWeights}
        \end{align}

\end{enumerate}
Here, $\Sym{X}$ denotes the group of permutations on the set $X$. The graph
automorphisms of $\mathcal{G}$ form a group $\aut{\G}$ via concatenation.

In the following, we consider the vertex-weighted graph $\GC$ constructed in
\cref{sec:GraphDefinition}. We use the shorthand notation
\eqref{eq:ShortNotation1} and denote as $\AC \equiv \aut{\GC}$ its group of
graph automorphisms. Note that the graph $\GC$ consists of vertices $\Vs^C$ and
$\Vi^C$ with different weights. Hence a graph automorphism $\Phi \in \AC$ only
satisfies the condition \eqref{eq:VertexWeights} if it has the form
\begin{align}
    \Phi = \phi_s \circ \varphi_l\,,
    \label{eq:GeneralAutomorphism}
\end{align}
with $\phi_s \in \text{Sym}(\Vs^C)$ and $\varphi_l \in
\text{Sym}(\,\bigcup_{l \in \ZD}\Vi^C\,)$.

As a starting point, notice that $\GC$ features the graph automorphism 
\begin{alignat}{5}
    \Phi_R :\; &\Vi^C &&\rightarrow V_{l+1}^C&&: \;\, (l, h_l) &&\mapsto (l+1, h_l)
    \nonumber\\
               &\Vs^C &&\rightarrow \Vs^C&&: \;\, (g_1, g_2, c) &&\mapsto (g_3, g_1, g_3cg_3^{-1})\,,
    \label{eq:RotSym}
\end{alignat}
with $g_3 \equiv (g_1g_2)^{-1}c$ as defined above; here addition of link
labels is meant modulo-3.
This automorphism reflects the rotation symmetry of the site since it permutes
the link vertices by increasing their indices $l \mapsto l + 1$ (mod $3$).
Note that in the current notation, it is not obvious that this is a graph
automorphism. To see this, we extend the shorthand notation
\eqref{eq:ShortNotation1} by
\begin{align}
    \begin{split}
    &g_3' \equiv (g_1'g_2')^{-1}c'
    \quad\text{for}\quad 
    w' \equiv (g_1', g_2', c') \equiv \Phi(w)\,, 
    \\
    &v_{l'}' \equiv (l', h_{l'}') \equiv \Phi(v_l)\,.
    \end{split}
    \label{eq:ShortNotation2}
\end{align}
First, note that $g_3cg_3^{-1} \in C$ so that the map \eqref{eq:RotSym} is
well-defined. With $\Phi_R^2(w) = (g_2, g_3, [g_2g_3]c[g_2g_3]^{-1})$, it is easy
to see that $\Phi_R^3 = \Id$; therefore the map is bijective. Second,
$\Phi_R$ acts bijectively on $\EC$ [recall~\cref{eq:Connectivity}] since
\begin{subequations}
    \begin{align}
        g_1' &= g_3 = h_3 = h_1'\,, \\
        g_2' &= g_1 = h_1 = h_2'\,, \\
        g_3' &= (g_3g_1)^{-1}g_3cg_3^{-1} = g_2 = h_2 = h_3'\,.
    \end{align}
\end{subequations}
Since $\Phi_R$ does not mix $\Vi^C$ and $\Vs^C$, \cref{eq:VertexWeights} is
trivially satisfied. This shows that $\Phi_R \in \AC $ is indeed a graph
automorphism of $\GC$.
This is also apparent in the rotational symmetry of \cref{fig:ClassSite}.

At this point the question arises whether $\GC$ also features a 
mirror symmetry of the form
\begin{alignat}{5}
    \Phi_M :\;  &V_1^C &&\rightarrow V_2^C&&: \;\, (1, h_1) &&\mapsto (2, h_1) \nonumber\\
                &V_2^C &&\rightarrow V_1^C&&: \;\, (2, h_2) &&\mapsto (1, h_2) \nonumber\\
                &\Vs^C &&\rightarrow \Vs^C&&: \;\, (g_1, g_2, c) &&\mapsto (g_2, g_1, c)\,,
    \label{eq:MirSym}
\end{alignat}
that exchanges only the vertices $V_1^C$ and $V_2^C$ leaving $V_3^C$ invariant.
This map is bijective as $\Phi_M^2 = \Id$,
and $\Phi_M$ satisfies \cref{eq:VertexWeights} as it does not mix $\Vi^C$ and $\Vs^C$.
Nevertheless, the map \eqref{eq:MirSym} is \emph{no} graph automorphism for a non-abelian group $G$ as 
$\Phi_M$ does not act bijectively on $\EC$ [recall~\cref{eq:Connectivity}]:
\begin{align}
    g_3' = (g_2g_1)^{-1}c \neq (g_1g_2)^{-1}c = g_3 = h_3 = h_3';
\end{align}
the site constraint \eqref{eq:VertexCond} that characterizes the edges $\EC$
inherently depends on the ordering of the links (up to cyclic permutations).
In fact, the (non-abelian) example constructed in \cref{fig:ClassSite} is actually not mirror symmetric.
Only for abelian groups the multiplication $g_1g_2 = g_2g_1$ is commutative for every group element $g_1,g_2 \in G$;
in this case \cref{eq:MirSym} defines a valid graph automorphism.

Although the rotation symmetry of $\GC$ is conceptually important, it is not
relevant for the following discussion. Henceforth we are interested in
\emph{local graph automorphisms} of the more restrictive form
[cf.~\cref{eq:GeneralAutomorphism}]
\begin{align}
    \Phi = \phi_s \circ \varphi_1 \circ \varphi_2 \circ \varphi_3\,,
    \label{eq:LocalAutomorphism}
\end{align}
where $\phi_s \in \text{Sym}(\Vs^C)$ and $\varphi_l \in
\text{Sym}(\Vi^C)$. This means that local graph automorphisms only permute
vertices \emph{within links}, but not between links. The local graph
automorphisms span a subgroup of the full automorphism group, we denote this
subgroup as $\AC^\text{loc}$.

Consider a local graph automorphism $\Phi \in \AC^\text{loc}$ of the form
\eqref{eq:LocalAutomorphism}. Using the shorthand notation of
\cref{eq:ShortNotation1,eq:ShortNotation2}, this implies $l' = l$.
Then we can define $\sigma_l \in \Sym{G}$ and $\varsigma_s: \Vs^C
\rightarrow C$ by $\sigma_l(h_l) = h_l'$ and $\varsigma_s(w) = c'$,
respectively. As a graph automorphism, $\Phi$ must fulfill the connectivity
condition \eqref{eq:Connectivity}, that means $g_l' = h_l'$ if and only if $g_l
= h_l$. With the newly defined $\sigma_l$ and $\varsigma_s$, this
implies 
\begin{subequations}
\begin{align}
    \varphi_l(v_l) &= (l, \sigma_l(h_l))\,,
    \label{eq:DefnitionOfPhi(l)}\\
    \phi_s(w) &= (\sigma_1(g_1), \sigma_2(g_2), \varsigma_s(w))\,,
    \label{eq:DefnitionOfPhi(s)}\\
    \varsigma_s(w) &= \sigma_1(g_1)\sigma_2(g_2)\sigma_3(g_3)\,,
    \label{eq:mastercondition}
\end{align}
\end{subequations}
for all $g_1, g_2 \in G$ and $c \in C$.  
\cref{eq:mastercondition} poses a constraint on the allowed choices for
$\sigma_l$ and $\varsigma_s$; if it is satisfied, then
\cref{eq:DefnitionOfPhi(l)} and \cref{eq:DefnitionOfPhi(s)} define
$\varphi_l$ and $\phi_s$ which, in turn, fully determine $\Phi \in
\AC^\text{loc}$ via \cref{eq:LocalAutomorphism}.

Note that $\varsigma_s(w)$ with $w \equiv (g_1, g_2, c)$ generally depends
on all three variables $g_1, g_2 \in G$ and $c \in C$. For a general group $G$,
we are not aware of a way to classify \emph{all possible} functions
$\sigma_l$ and $\varsigma_s$ that satisfy \cref{eq:mastercondition}.
However, this is not necessary for our purposes anyway. To solve
\cref{eq:mastercondition}, we make the ansatz $\varsigma_s(w) =
\sigma_s(c)$ which does not depend on the specific group elements $g_1, g_2
\in G$, such that $\sigma \in \text{Sym}(C)$ is well-defined. 
This leads to the simpler constraint
\begin{align}
    \sigma_s(c) = \sigma_1(g_1)\sigma_2(g_2)\sigma_3(g_3) \label{eq:AutomorphismCondition}
\end{align}
with $g_3 \equiv (g_1g_2)^{-1}c$ for every $(g_1, g_2, c) \in \Vs^C$.
Pictorially in \cref{fig:ClassSite}, this restricts the local graph automorphisms
to collective permutations of the patches of vertices with equal $c$
on the site.  Crucially, for the case $C = \{1\}$, there is only one patch and
the function $\sigma_s$ can only be constant and equal to $1$. Thus, for
this special case, the following characterization yields the \emph{full} local
automorphism group. 

In the next step, we derive a set of equivalent relations from
\cref{eq:AutomorphismCondition}.
For $c \in C$ and $g, h \in G$, consider the three vertices $(gh, 1, c)$, $(g,
h, c)$, and $(1, gh, c) \in \Vs^C$ for which $g_3 = (gh)^{-1}c$. Plugging this in \cref{eq:AutomorphismCondition}
and equating for $\sigma_s(c)[\sigma_3(g_3)]^{-1}$, we obtain
\begin{align}
    \begin{split}
        \sigma_1(gh)\sigma_2(1) 
        = \sigma_1(g) \sigma_2(h) 
        = \sigma_1(1) \sigma_2(gh)\,.
        \label{eq:s2}
    \end{split}
\end{align}
Specifically for $h = 1$, we can evaluate the right-hand side of \cref{eq:s2}
as
\begin{align}
    \sigma_2(g) 
    = [\sigma_1(1)]^{-1}\sigma_1(g) \sigma_2(1)\,,
    \label{eq:c1}
\end{align}
and plugging this in the left-hand side of \cref{eq:s2} yields
\begin{align}
    \sigma_1(gh) 
    = \sigma_1(g)[\sigma_1(1)]^{-1}\sigma_1(h)\,.
    \label{eq:c2}
\end{align}
Finally, for $c \in C$ and $g \in G$ consider $(cg^{-1}, 1, c) \in \Vs^C$ with
$g_3 = g$. Then \cref{eq:AutomorphismCondition} yields
\begin{align}
    \sigma_3(g) 
    = [\sigma_1(cg^{-1})\sigma_2(1)]^{-1}\sigma_s(c)\,.
    \label{eq:c3}
\end{align}
\cref{eq:c1,eq:c2,eq:c3} must be valid for all $g, h \in G$ and $c \in C$.

We now use \cref{eq:c2} to translate $\sigma_1$ into known algebraic
objects. To this end, we define the left ($\lambda)$ and right ($\rho$) translation
and the conjugation ($\chi$) on the group as follows:
\begin{subequations}
    \begin{alignat}{5}
        \lambda&: \; G &&\rightarrow \text{Sym}(G)&&: \; g &&\mapsto (\lambda_g: h \mapsto gh)\,,
        \label{eq:lambda}\\
        \rho&: \; G &&\rightarrow \text{Sym}(G)&&: \; g &&\mapsto (\rho_g: h \mapsto hg^{-1})\,,
        \label{eq:rho}\\
        \chi&: \; G &&\rightarrow \text{Sym}(G)&&: \; g &&\mapsto \chi_g = \lambda_g \circ \rho_g\,.
        \label{eq:conj}
    \end{alignat}
\end{subequations}
Setting $\tau := \lambda_{[\sigma_1(1)]^{-1}} \circ \sigma_1 \in
\text{Sym}(G)$, condition \eqref{eq:c2} becomes $\tau(gh) = \tau(g)\tau(h)$ for
all $g, h \in G$; this means that $\tau \in \text{Aut}(G)$ must be a
\emph{group automorphism}. 
If we select a group element $p_1 := \sigma_1(1) \in G$ and a group
automorphism $\tau \in \text{Aut}(G)$, this already fully defines $\sigma_1
= \lambda_{p_1} \circ \tau$. In addition, we can choose a group element $p_2 :=
[\sigma_2(1)]^{-1} \in G$. Then \cref{eq:c1} already fully determines
$\sigma_2 = \rho_{p_2} \circ \tau$.

Finally, we determine $\sigma_3$ and $\sigma_s$. To this end, we
define $p_3 := p_1^{-1}[\sigma_3(1)]^{-1}p_2$. Then for $g = 1$,
\cref{eq:c3} determines $\sigma_s = \chi_{p_1} \circ (\rho_{p_3} \circ
\tau)$.  But $\sigma_s$ can only take values in $C$, thus $\rho_{p_3} \circ
\tau$ must map to $C$.
We are not aware of a general criterion on $\tau$ and $p_3$ such that this is
satisfied; but again, this is not necessary for our purposes.
In the following, we restrict ourselves to the subgroup of $\AC^\text{loc}$
where $p_3 = 1$ and $\tau \in \text{Aut}_C(G)$. Here, $\text{Aut}_C(G)$
denotes the subgroup of $\text{Aut}(G)$ that preserves the conjugacy class $C$.
This includes at least the conjugations (= \emph{inner automorphisms})
$\chi_G$.
Then we obtain $\sigma_s = \chi_{p_1} \circ \tau$, and \cref{eq:c3} fully
defines $\sigma_3 = \lambda_{p_2} \circ \rho_{p_1} \circ \tau$. Crucially,
for $C = \{1\}$, the full automorphism group $\text{Aut}(G) =
\text{Aut}_{\{1\}}(G)$ is $C$-preserving, such that $p_3 = 1$ is the only
$C$-preserving choice.  Thus, for this case, this characterization yields the
full local automorphism group.

In summary, we parametrize the group permutations
\begin{subequations}
    \label{eq:AutomorphismGroupAction}
    \begin{align}
        \sigma_s^{p_1, p_2, \tau}(c) &= p_1\tau(c)p_1^{-1} 
        \label{eq:AutomorphismGroupAction1}\\
        \sigma_1^{p_1, p_2, \tau}(g) &= p_1\tau(g), 
        \label{eq:AutomorphismGroupAction2}\\
        \sigma_2^{p_1, p_2, \tau}(g) &= \tau(g)p_2^{-1}, 
        \label{eq:AutomorphismGroupAction3}\\
        \sigma_3^{p_1, p_2, \tau}(g) &= p_2\tau(g)p_1^{-1} 
        \label{eq:AutomorphismGroupAction4}
    \end{align}
\end{subequations}
on $g \in G$ and $c \in C$ by parameters $p_1, p_2 \in G$, and $\tau \in
\text{Aut}_C(G)$.
%
%
It is now easy to see that the group permutations
\eqref{eq:AutomorphismGroupAction} fulfill \cref{eq:AutomorphismCondition} for
any choice of parameters, and it is $\sigma_s(c) \in  C$ as required.

Now we derive the underlying structure of the parametrizing group.
Consider $p_1, p_2, p_1', p_2' \in G$ and $\tau, \tau' \in \text{Aut}_C(G)$,
then concatenation yields
\begin{align}
    \sigma_x^{p_1, p_2, \tau} \circ \sigma_x^{p_1', p_2', \tau'} 
    = \sigma_x^{p_1\tau(p_1'), p_2\tau(p_2'), \tau \circ \tau'} 
    \label{eq:GroupStructure}
\end{align}
for $x \in \ZD$ or $x = s$. 
Hence the parametrizing group is given by the semidirect product $G^2 \rtimes
\text{Aut}_C(G)$ with group product
\begin{align}
    (p_1, p_2, \tau) \cdot (p_1', p_2', \tau') 
    = (p_1\tau(p_1'), p_2\tau(p_2'), \tau \circ \tau')
    \label{eq:GroupProduct}
\end{align}
for $(p_1, p_2, \tau), (p_1', p_2', \tau') \in G^2 \rtimes \text{Aut}_C(G)$.

For each element $(p_1, p_2, \tau) \in G^2 \rtimes \text{Aut}_C(G)$, the
parametrization \eqref{eq:AutomorphismGroupAction}, together with
\cref{eq:DefnitionOfPhi(l),eq:DefnitionOfPhi(s)}, determines a distinct local
graph automorphism $\Phi^{p_1, p_2, \tau} \in \AC^\text{loc}$ of the form
\cref{eq:LocalAutomorphism}. These local graph automorphisms span a subgroup
$\AC^\text{par}$ of the full group of local graph automorphisms:
\begin{align}
    G^2\rtimes\text{Aut}_C(G) \cong \AC^\text{par} \le \AC^\text{loc}\,.
    \label{eq:GroupMonomorphism}
\end{align} 
Crucially, for the case $C = \{1\}$, all our choices were dictated by the
constraint \eqref{eq:mastercondition}. Hence, for this case, we have shown:
\begin{proposition}
  For $C = \{1\}$, the group of local graph automorphisms on $\GG$ ($\G_s$ in
  the main text) is isomorphic to $G^2 \rtimes \aut{G}$. 
  \label{eq:1Auts}
\end{proposition}

More generally, consider $C = \{c\}$ with only one element $c \in Z_G$ in the
center of $G$. Note that this is the general case if $G$ is Abelian. In this case,
$\varsigma_s(w) = c = \sigma_s(c)$ must be constant, hence condition
\eqref{eq:AutomorphismCondition} is equivalent to condition
\eqref{eq:mastercondition}.  Furthermore, the condition $\rho_{p_3} \circ \tau
\in C$ is uniquely solved by $p_3 \equiv c^{-1}\tau(c)$ for any $\tau \in
\text{Aut}(G)$.
This yields $\sigma_s(c) = c$ and $\sigma_3(g) =
p_2\tau(g)p_3^{-1}p_1^{-1}$ for $g \in G$. Together with
\cref{eq:AutomorphismGroupAction2,eq:AutomorphismGroupAction3}, it is easy to
see that \cref{eq:AutomorphismCondition} is fulfilled for any $p_1,p_2 \in G$
and $\tau \in \text{Aut}(G)$. Note that $p_3, p_3^{-1} \in Z_G$ are in the
center since $c^{-1}, \tau(c) \in Z_G$ for $c \in Z_G$.
%
%
Then it is easy to check that \cref{eq:GroupStructure} remains satisfied for
$\sigma_s$ and $\sigma_3$. Therefore, the parametrizing group is still
given by the semidirect product $G^2 \rtimes \text{Aut}(G)$ with group product
\eqref{eq:GroupProduct}. This construction therefore also yields the full group
of local graph automorphisms:
\begin{proposition}
  For an Abelian group $G$, the group of local graph automorphisms on $\GC$ is
  isomorphic to $G^2 \rtimes \aut{G}$. 
  \label{eq:AbGroupAuts}
\end{proposition}

By contrast, for a non-Abelian group $G$ with $C \nsubseteq Z_G$ not in the
center, the above parametrization only yields a subgroup
$\AC^\text{par}$ of $\AC^\text{loc}$.
For example, consider $G = S_3$ and $C = \{\sigma, \sigma R, \sigma R^2\}$ with
the notation used in \cref{sec:anyons}. Note that $\rho_{p_3} \circ \tau$ maps
to $C$ for $\tau = \Id$ and $p_3 = R$. This makes 
\begin{alignat}{5}
    \Phi:\;& V_3 &&\rightarrow V_3&&:\; (3, h_3) &&\mapsto (3, h_3R^{-1})\,,
    \nonumber\\
         &\Vs^C &&\rightarrow \Vs^C&&:\; (g_1, g_2, c) &&\mapsto (g_1, g_2, cR^{-1})
\end{alignat}
with $p_1, p_2 = 1$ a local graph automorphism on $\GC$ -- which we did not
capture with the above parametrization.

We conclude this section with some remarks on the rotation symmetry $\Phi_R$
given by \cref{eq:RotSym}.
The subgroup $\AC^\text{par}$ determined by \cref{eq:AutomorphismGroupAction}
is not manifestly rotationally symmetric. This is due to our asymmetric choice
of parametrization. For the concatenation with $\Phi_R$ we obtain
\begin{align}
    \Phi^{p_1, p_2, \tau} \circ \Phi_R 
    = \Phi_R \circ \Phi^{p_2^{-1},\, p_1p_2^{-1},\, \chi_{p_2} \circ \tau}\,.
\end{align}
That is, $\Phi_R$ permutes elements within $\AC^\text{par}$ such that $\tau
\mapsto \chi_{p_2} \circ \tau$ is modified by an additional conjugation.
However, the full group $\AC^\text{par}$ is rotationally symmetric:
\begin{align}
    \AC^\text{par} \circ \Phi_R 
    = \Phi_R \circ \AC^\text{par}\,.
\end{align}

\subsection{Maximum-weight independent sets of $\GC$}%
\label{sec:MIS}

Since the subgraph $K_{N^2M} = (\Vs^C, \Es^C, \Ds^C)$ of $\GC$ on the site is
fully connected, any independent set (IS) of $\GC$ can contain at most one
vertex from $\Vs^C$.
	
Consider a maximal IS (MIS) $M_0$ of $\GC$ which includes no vertex from
$\Vs^C$.  As the vertices $\Vi^C$ of link $l$ are only connected to vertices of
$\Vs^C$, for $M_0$ to be maximal, it must be $\Vi^C \subseteq M_0$. Therefore,
the MIS $M_0 = \bigcup_{l \in \ZD} \Vi^C$ consists of all $3N$ vertices on the
links.  Each such link-vertex has weight $1\Delta$, therefore $M_0$ possesses
the total weight $3N\Delta$.
	
Now we consider a MIS $M_w$ of $\GC$ which includes one vertex $w \equiv (g_1,
g_2, c) \in \Vs^C$ on the site. The vertex $w$ is connected to three
link-vertices $(1, g_1)$, $(2, g_2)$ and $(3, g_3)$, where we use the shorthand
notation $g_3 \equiv (g_1g_2)^{-1}c$ from \cref{eq:ShortNotation1}. Thus, for
$M_w$ to be maximal, it must include all but these three link-vertices:
\begin{align}
    M_w = \{w\} \cup \bigcup_{l \in \ZD} \,[\Vi^C \setminus (l, g_l)]\,.
\end{align}
Each such $M_w$ for $w \in \Vs^C$ includes $|M_w| = 3N - 2$ elements. The
vertex $w$ has weight $4\Delta$, hence $M_w$ possesses the total weight $(3N +
1)\Delta$.

This makes the MISs $M_w$ for each vertex $w \in \Vs^C$ the maximum-weight
independent sets (MWISs) of $\GC$. Thus, there are in total $N^2M$ MWISs. We
denote their set by $\mathcal{M}_C$.

Each MWIS $M_w$ is associated to a unique vertex $w$ on the site. The three
link-vertices $(1, g_1)$, $(2, g_2)$ and $(3, g_3)$ that are not part of $M_w$
satisfy the class-$C$ site constraint from \cref{eq:VertexCond}. Conversely, for each
such triple which satisfies the site constraint, there exists the vertex $w =
(g_1, g_2, g_1g_2g_3) \in \Vs^C$ with corresponding MWIS $M_w \in \mathcal{M}_C$.
That is, the MWISs of $\GC$ precisely encode all possible configurations that
satisfy the site constraint.

\subsection{Full symmetry of $\GC$}%
\label{sec:FullSymm}

A graph automorphism is a permutation on the vertex set, that conserves
connectivity [recall \cref{eq:Connectivity}] and the vertex weights [recall
\cref{eq:VertexWeights}]. That is, an MWIS must always be mapped to another
MWIS under element-wise application of a graph automorphism. For the blockade
graph $\GC$, from \cref{sec:MIS} we know that the MWISs are given by
$\mathcal{M}_C$. Therefore, for a given graph automorphism $\Phi \in \AC$, the MWIS
$M_w$ which includes the vertex $w \in \Vs^C$ must be mapped to the MWIS
$\Phi(M_w) = M_{\Phi(w)}$ which includes vertex $\Phi(w) \in \Vs^C$.  This
induced mapping $\Phi:\,\mathcal{M}_C \,\rightarrow\, \mathcal{M}_C$ on the MWISs is
bijective, as $\Phi$ acts bijectively on the vertex set.

For an MWIS $M_w \in \mathcal{M}_C$, we can define the set
\begin{align}
    \AC \cdot M_w := \{\Phi(M_w) \;\vert\; \Phi \in \AC\}
\end{align}
via element-wise application of all graph automorphisms. The graph $\GC$ is
\emph{fully symmetric} (in the sense of Ref.~\cite{Maier2025}) if $\AC \cdot
M_w = \mathcal{M}_C$ for some $M_w \in \mathcal{M}_C$, that is if $\mathcal{M}_C$ is an
orbit under the action of $\AC$.
	
In the following, we use the shorthand notation $w \equiv (g_1, g_2, c)$ with
$g_3 \equiv (g_1g_2)^{-1}c$ and $w' \equiv (g_1', g_2', c')$ with $g_3' \equiv
(g_1'g_2')^{-1}c'$ from \cref{eq:ShortNotation1} and \cref{eq:ShortNotation2},
respectively.  For two elements $c, c' \in C$ in the same conjugacy class, by
definition, there exists an element $r_{c \mapsto c'} \in G$ such that
$\chi_{r_{c \mapsto c'}}(c) = c'$.  If the centralizer $Z_G(c) \neq \{e\}$ is
non-trivial, this element is not unique, and we can choose some representative in
the set $r_{c \mapsto c'}Z_G(c)$ [defined by elementwise multiplication]. 
For any $w, w' \in \Vs^C$ and $\tau' \in
\text{Aut}_C(G)$, we can now define the graph automorphism $\Phi_{w \mapsto w'}
:= \Phi^{p_1, p_2, \tau}$ via the parameters
\begin{subequations}
    \label{eq:FSParametr}
    \begin{align}
        p_1 &:= g_1'[\chi_q \circ \tau'](g_1^{-1})\,,
        \label{eq:Defp_1}\\
        p_2 &:= g_2'^{-1}[\chi_q \circ \tau'](g_2)\,,
        \label{eq:Defp_2}\\
        \tau &:= [\chi_q \circ \tau']
    \end{align}
\end{subequations}
with $q := g_1'^{-1}r\tau'(g_1)$ and $r := r_{\tau'(c) \mapsto c'}$. Plugging
in $q$, \cref{eq:Defp_1} and \cref{eq:Defp_2} can be rewritten as
\begin{subequations}
    \begin{align}
        p_1 &= rq^{-1}\,,\\
        p_2 &= g_3'c'^{-1}r\tau'(cg_3^{-1})q^{-1} = g_3'r\tau'(g_3^{-1})q^{-1}\,,
    \end{align}
\end{subequations}
respectively. Then, by construction, the group permutations
\eqref{eq:AutomorphismGroupAction} become
\begin{subequations}
    \begin{alignat}{3}
        \sigma_1(g_1) &= p_1\tau(g_1) &&= g_1'\,,\\
        \sigma_2(g_2) &= \tau(g_2)p_2^{-1} &&= g_2'\,,\\
        \sigma_3(g_3) &= p_2\tau(g_3)p_1^{-1} &&= g_3'\,,\\
        \sigma_s(c) &= \chi_{p_1}(\tau(c)) = \chi_{r}(\tau'(c)) &&= c'\,.
    \end{alignat}
\end{subequations}
This is what is required such that
\cref{eq:DefnitionOfPhi(l),eq:DefnitionOfPhi(s)} yield $\Phi_{w \mapsto
w'}(w) = w'$ for $\Phi_{w \mapsto w'}$ of the form
\eqref{eq:LocalAutomorphism}.

It obviously is $\Phi_{w \mapsto w'} \in \AC$ because $(p_1, p_2, \tau) \in G^2
\rtimes \text{Aut}_C(G)$.  For a given $M_{w'} \in \mathcal{M}_C$, we can now
choose $\Phi = \Phi_{w \mapsto w'}$ for some $\tau' \in \text{Aut}_C(G)$ such
that $M_w \in \mathcal{M}_C$ is mapped to $\Phi(M_w) = M_{\Phi(w)} = M_{w'}$.
There always exists a $\tau' \in \text{Aut}_C(G)$ since conjugations (= inner
automorphisms) $\chi_G$ are always part of $\text{Aut}_C(G)$. In general, there
is some freedom in the choice of $\tau'$ so that the choice of $\Phi$ is not
unique. We conclude:
\begin{proposition}
    The generalized graph $\GC$ for a class-$C$ site is fully symmetric.
    \label{eq:GCFullySym}
\end{proposition}

As a concluding remark, we specialize to $C = \{1\}$. In this case,
$c, c' = 1$ is fixed, and $\tau' \in \text{Aut}(G)$ can be any group
automorphism.  Without loss of generality, we can choose $r_{1 \,\mapsto\, 1} =
g_1'\tau'(g_1)^{-1}$ as representative, such that $q = 1$ and $\chi_q =
\Id$ become trivial. This simplifies the parameters to $p_1 =
g_1'\tau'(g_1)^{-1}$, $p_2 = g_3'g_1'\tau'(g_3g_1)^{-1}$, and $\tau = \tau'$.

\subsection{The graph $\GCl$ for all conjugacy classes}%
\label{sec:MultiClassSite}

The construction from \cref{sec:GraphDefinition} of the graph $\GC$ for a class
$C$-site can readily be generalized to incorporate all conjugacy classes $C \in
\text{Cl}(G)$ of the group $G$. This results in the graph $\GCl$ of the
multi-class site discussed in \cref{sec:fluxlattice} (and denoted there by
$\G_s[\bm\Delta]$). In this section, we rigorously construct this graph,
discuss its automorphisms and MWISs, and show that it is fully symmetric.

In the following, we denote the graph of a class-$C$ site as $\GC = (\VC, \EC,
\DC)$, and the graph of the multi-class site as $\GCl = (\VCl, \ECl, \DCl)$.
The construction of $\GCl$ from the $\GC$ works as follows. For each conjugacy
class $C \in \text{Cl}(G)$, consider the graph $\GC$ and increase the weight of
the vertices $w \in \Vs^C$ on the site uniformly to $\Delta_w^C = \Delta_C$. 
We assume $\Delta_C > 3\Delta$ for each conjugacy class.  
The maximal weight is denoted $\Delta_\text{max}
:= \max_{C \in \text{Cl}(G)} \Delta_C$. Crucially, we assume that the maximal
weight is unique for a conjugacy class $C_\text{max}$, i.e., 
\begin{align}
    \Delta_C 
    = \Delta_{C'} 
    = \Delta_\text{max}
    \quad\Leftrightarrow\quad
    C = C' = C_\text{max} 
    \label{eq:UniquenessOfMaxDet}
\end{align}
for $C, C' \in \text{Cl}(G)$. As the weights are uniform on the site,
they change neither the graph automorphisms $\text{Aut}(\GC)$ nor the MWISs
$\mathcal{M}_C$ of the graphs $\GC$. 

For each conjugacy class $C \in \text{Cl}(G)$, we can now ``stack'' the graphs
$\GC$ on top of each other by identifying their vertices $\Vi^\text{Cl} = \Vi^C
= \{l\} \times G$ on the links. Then we fully connect any pair of vertices from
different conjugacy classes on the site via an additional edge in $\ECl$. Hence
we end up with a fully-connected (= complete) subgraph $K_{N^3} =
(\Vs^\text{Cl}, \Es^\text{Cl}, \Ds^\text{Cl})$ with vertices 
\begin{align}
    \Vs^\text{Cl} 
    = \bigcup_{C \in \text{Cl}(G)} \Vs^C 
    = G^3 
\end{align}
on the site. Each vertex $w \equiv (g_1, g_2, c) \in \Vs^\text{Cl}$ is still
connected to three link-vertices $(1, g_1)$, $(2, g_2)$ and $(3, g_3)$, where
we use the shorthand notation $g_3 \equiv (g_1g_2)^{-1}c$ from
\cref{eq:ShortNotation1}.
The total vertex set is then
$\VCl = \Vs^\text{Cl} \cup \bigcup_{l \in \ZD} \Vi^\text{Cl}$, 
and the total set of edges is $\ECl = \Es^\text{Cl} \cup \bigcup_{l \in \ZD}
\Ei^\text{Cl}$ with edges
\begin{align}
	\Ei^\text{Cl} 
  = \{\, \{w, v_l\} \,\vert\, w \in \Vs^\text{Cl}, \, v_l \in \Vi^\text{Cl}, \; g_l = h_l \,\}
\end{align}
between the site vertices and vertices of link $l$.

We now discuss the parametrization of local graph automorphisms.
$\varsigma_s:\, \Vs^\text{Cl} \rightarrow G$ must fulfill the condition
\begin{align}
    \varsigma_s(w) 
    = \sigma_1(g_1) \sigma_2(g_2) \sigma_3(g_3) 
    \label{eq:Nicolaissondervertex}
\end{align}
for all $g_1, g_2, c \in G$ with the shorthand notation
\cref{eq:ShortNotation1}, in analogy to \cref{eq:mastercondition}. 
In the generic case, the weights $\Delta_C$ are different, thus we must
require that $\varsigma_s(w) \in C_c \equiv \text{Cl}(c)$
preserves the conjugacy class. This makes the classification problem of local
graph automorphisms similar to the one solved in \cref{sec:locaut}. The
condition $\varsigma_s(w) \in C_c$ is now slightly stricter, as it
requires that $\tau \in \text{Aut}_{\text{Cl}(G)}(G)$ is $C$-preserving for
\emph{every} conjugacy class $C \in \text{Cl}(G)$. 
Hence we obtain the parametrized group $\AC^\text{par} \cong G^2 \rtimes
\text{Aut}_{\text{Cl}(G)}(G)$ as subgroup of $\AC^\text{loc}$.

The subgraph $K_{N^3}$ on the site is again fully-connected. This excludes ISs
with multiple vertices on the site. By the same reasoning as in \cref{sec:MIS},
we can argue that the MISs are given by $M_0$ and $M_w$
for each $w \in \Vs^\text{Cl}$. Again, $M_0$ possesses the weight $3N\Delta$,
and the $\mathcal{M}_C$ possess the weights $(3N - 3 + \Delta_C)\Delta$. Assuming
condition \eqref{eq:UniquenessOfMaxDet} for our weights, the MWISs are
given by $\mathcal{M}_{C_\text{max}}$ with weight $(3N - 3 +
\Delta_\text{max})\Delta$.

Finally, we show full symmetry of $\GCl$ in the sense of Ref.~\cite{Maier2025},
that is, $\AC \cdot M_w = \mathcal{M}_{C_\text{max}}$ for some $M_w \in
\mathcal{M}_{C_\text{max}}$. Fortunately, we can simply construct $\Phi_{w \mapsto
w'} = \Phi^{p_1,p_2,\tau}$ via the parametrization \eqref{eq:FSParametr} for
some $\tau' \in \text{Aut}_{\text{Cl}(G)}(G)$, similar to the previous
\cref{sec:FullSymm}. This is possible since we guaranteed that $w,w' \in \Vs^{C_\text{max}}$
correspond to the same conjugacy class $C_\text{max}$ by assumption \eqref{eq:UniquenessOfMaxDet}. 
We can conclude that: 
\begin{proposition}
    The generalized graph $\GCl$ for a multi-class site is fully symmetric.
    \label{eq:GClFullySym}
\end{proposition}

\section{The tessellated blockade structure $\G$ on the torus}%
\label{sec:torus}

In this appendix we discuss the automorphism group of the tessellated graph
$\G$ on the honeycomb lattice with periodic boundaries. Note that
without modifications, the honeycomb lattice can only be embedded on a torus
or, as a finite patch with appropriate boundary conditions, in the plane. Here
we discuss the situation on the torus. In \cref{sec:boundary} we discuss how
boundary conditions affect our model.


\subsection{Construction and maximum-weight independent sets}%
\label{sec:MWIS_tess}

\begin{figure}[tbp]
    \begin{center}
        \includegraphics*[width = 0.95\linewidth]{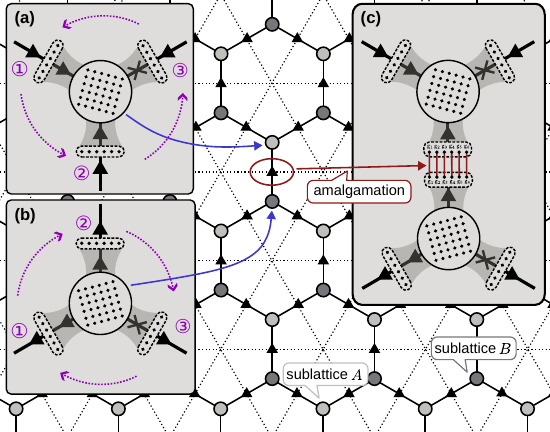}
    \end{center}    
    \caption{\textbf{Construction of the tessellated graph $\G$.}
        The honeycomb lattice is bipartite and can be partitioned into two
        sublattices $A$ and $B$, indicated by light and dark gray circles.
        \textbf{(a)} On every site in sublattice $A$, we place the graph
        constructed in \cref{sec:aut_single}, such that the links numbered $1$,
        $2$ and $3$ are ordered \emph{counterclockwise} around the site. 
        \textbf{(b)} On every site in sublattice $B$, we place the graph
        constructed in \cref{sec:aut_single}, such that the links numbered $1$,
        $2$ and $3$ are ordered \emph{clockwise} around the site. 
        \textbf{(c)} We amalgamate these graphs by identifying (on every link
        of the honeycomb lattice) the vertices that correspond to the same
        group element.
    }
    \label{fig:amalg}
\end{figure}

To characterize the maximum-weight independent sets (MWISs) of the
tessellation, it is convenient to phrase the construction differently than in
the main text. We can view the tessellated complex as an \emph{amalgamation}
(see Ref.~\cite{Stastny2023a}) of the building blocks introduced in
\cref{sec:aut_single}, see \cref{fig:amalg}.
Let $\Lambda_\hex = (S_\hex, L_\hex, P_\hex)$ denote the honeycomb lattice,
where $S_\hex$ denotes the sites, $L_\hex$ the set of links and $P_\hex$ the
set of plaquettes. As the honeycomb lattice is bipartite, we can partition it
into two subsets $S_\hex = A \cup B$ such that sites from $A$ are only
connected to sites from $B$ and vice versa. We place the graph $\G_C$, as
constructed in \cref{sec:aut_single}, on the sites in $A$, and a mirrored
version of this graph on every site in $B$. We denote these graphs associated
to specific sites as $\G_{C_s}^s$ for $s \in S_\hex$. 

On each link $l \in L_\hex$, there are now two sets of $|G|$ vertices (from the
two graphs placed on the endpoints of $l$). The vertices in both sets are in
one-to-one correspondence with group elements from $G$. To amalgamate the
graphs $\G_{C_s}^s$, we first identify the vertices on the same link which are
associated to the same group element, and then add up their detunings.
Hence we obtain the detuning $2\Delta$ for the vertices on the links.  In the
following, we label the vertices on links by $(l,g)$ for $l \in L_\hex$, $g \in
G$, and the vertices on the sites as $(s, g_1, g_2, r)$ for $s \in S_\hex$
and $g_1, g_2, r \in G$. If
$\G_{C_s}^s$ is the graph for the conjugacy class $C_s = \{1\}$, then we omit
the last argument and write $(s, g_1, g_2)$. This construction is shown in
\cref{fig:amalg}.

Next, we characterize the MWISs of the tessellated graph $\G$. It is
straightforward to see that if there exists an \emph{independent set} $M$ of
$\G$, such that its restriction on $\G_{C_s}^s$ is a MWIS of $\G_{C_s}^s$ for
all $s \in S_\hex$, then $M$ is a MWIS of $\G$. We refer to an independent set
with this property as a \emph{globally consistent independent set}. It
is then easy to check that \emph{if} a globally consistent IS exists, then
\emph{all} MWIS of $\G$ have the property that their restriction on
$\G_{C_s}^s$ is a MWIS of $\G_{C_s}^s$ for every $s \in S_\hex$. 
Note that these observations are true for arbitrary weighted graphs. 
The existence of a globally consistent independent set is equivalent
to the condition given in Ref.~\cite[Section V.B]{Stastny2023a} (namely that
the ``$\gamma$-intersection'' of the languages associated to the structures
$\G_{C_s}^s$ is nonempty).

Weather a globally consistent IS exists must be checked for each choice of
conjugacy classes $\{C_s\}$. For the important case $C_s \equiv \{1\}$ on all
sites $s \in S_\hex$ (flux-free vacuum), the set
\begin{align}
    M := &\{(l, g)\, |\, l \in L_\hex, g\in G\setminus\{1\}\}
    \nonumber\\
    &\cup \{(s, 1,1)\, |\, s \in S_\hex\}
    \label{eq:globcons}
\end{align}
is independent since $(s, 1, 1)$ is adjacent to $(l, 1)$ on the links $l$
emanating from $s$, and to $(s, g_1, g_2)$ for all $(g_1, g_2) \in
G^2\setminus\{(1,1)\}$. 
Furthermore, the restriction of $M$ on each $\G_s \equiv \G_{\{1\}}^s$
is a MWIS (see \cref{sec:MIS}). This makes $M$ globally consistent. So in this
case, the restriction of every MWIS $M$ of $\G$ on $\G_s$ is a MWIS of $\G_s$
for all $s \in S_\hex$. Therefore all MWIS of $\G$ are uniquely characterized
by group elements $g_l \in G$ for each link $l \in L_\hex$ which satisfy the
zero-flux condition. 
This means that for the emanating links $l_1, l_2, l_3$ of $s$, listed in
counterclockwise order, the group elements satisfy
\begin{subequations}
    \begin{align}
        \label{eq:zf1}
        g_{l_1} g_{l_2} g_{l_3} = 1
    \end{align}
    if the links at $s$ are pointing inwards, and
    \begin{align}
        \label{eq:zf2}
        g_{l_3} g_{l_2} g_{l_1} = 1
    \end{align}
    if the links at $s$ are pointing outwards. 
    The form of condition~\eqref{eq:zf2} stems from the fact that the graph on
    these sites is mirrored.
\end{subequations}

In conclusion, the MWIS of $\G$ are in one-to-one correspondence with the
(product) ground states of Kitaev's quantum double model $H_G$ for
$J_p = 0$ and group $G$. 
In the following, we define $\mathcal{L}_G$ as the \emph{subset} of
$G^{|L_\hex|}$ that contains configurations $\vec{g} = (g_l)_{l \in L_\hex} \in
G^{|L_\hex|}$ which satisfy \cref{eq:zf1,eq:zf2}. Note that for non-abelian
$G$, $\mathcal{L}_G$ is not a group as it is not closed under (link-wise) multiplication.

\subsection{Definitions and preliminaries}
\label{sec:prelim}

The goal of the following sections is to characterize the automorphism group of
$\G$. In particular, we are interested in the question if (and if so, which)
automorphisms exist besides products of plaquette automorphisms. Since
we lack an exhaustive characterization of the automorphism group of the
single-vertex graph $\G_C$ for arbitrary conjugacy classes $C$, we focus here
on the flux-free case $C_s \equiv \{1\}$. We comment on the general case in
\cref{sec:gen}. We start in this section with definitions used throughout this
appendix.

We denote the set of vertices associated to link $l \in L_\hex$ as $V_l$, and
the set of vertices associated to site $s \in S_\hex$ as $V_s$. We are
primarily interested in a subgroup of the full automorphism group $\Aut{\G}
\equiv \A_\G$, consisting of automorphisms that map the sets $V_l$ to
themselves. This immediately implies that the same must be true for the vertex
sets $V_s$. We denote this subgroup of the automorphism group as
$\A_\G^\text{loc}$. It is easy to see that every automorphism in $\Aut{\G}$ can
be expressed as an automorphism in $\A_\G^\text{loc}$ followed by a symmetry of the 
underlying decorated Honeycomb lattice.

Every automorphism in $\A_\G^\text{loc}$ can be written in the form
\begin{align}
	\Phi = \prod_{l\in L_\hex} \varphi_l \circ \prod_{s \in S_\hex}\phi_s\,,
\end{align}
where $\varphi_l \in \Sym{V_l}$ and $\phi_s \in \Sym{V_s}$. So in particular
the restriction $\Phi_s: \G_s \rightarrow \G_s$ of $\Phi$ onto $\G_s$ is
well-defined and an automorphism of $\G_s$. Conversely, if $\Phi$ is a
permutation of $V_\G$ such that its restriction $\Phi_s$ onto $\G_s$ is an
automorphism of $\G_s$ for every $s \in S_\hex$, then $\Phi$ is an automorphism
of $\G$. This is so, because the amalgamation does not introduce additional edges.
Thus, to characterize the automorphisms $\Phi \in \A_\G^\text{loc}$, we can
apply our results from \cref{sec:locaut} to the restrictions $\Phi_s$ for every
$s \in S_\hex$.

This implies that the permutations $\varphi_s$ are uniquely
determined by the permutations $\varphi_l$. By identifying the vertices with
group elements (see \cref{sec:MWIS_tess}), the permutations $\varphi_l$ are
uniquely represented by functions $\sigma_l \in \Sym{G}$. 
For every site $s$ with emanating links $l_1, l_2, l_3$, listed in
counterclockwise order, condition~\eqref{eq:AutomorphismCondition} for $C =
\{1\}$ leads to the constraints
\begin{subequations}
    \begin{align}
	    \sigma_{l_3}((gh)^{-1}) 
      = (\sigma_{l_1}(g)\sigma_{l_2}(h))^{-1}\quad\forall g,h \in G
      \label{eq:cons1}
    \end{align}
    if the edges are directed inwards at $s$, and 
    \begin{align}
	    \sigma_{l_3}((gh)^{-1}) 
      = (\sigma_{l_2}(g)\sigma_{l_1}(h))^{-1}\quad\forall g,h \in G
      \label{eq:cons2} 
    \end{align}
    if the edges are directed outwards at $s$ (recall \cref{fig:amalg}).
\end{subequations}

We denote the subgroup of $\Sym{G}^{|L_\hex|}$, consisting of elements that
satisfy \cref{eq:cons1} and \cref{eq:cons2} at every site $s$, as
$\mathcal{L}_{\Sym{G}}$. Moreover, every $\vec{\sigma} \in
\mathcal{L}_{\Sym{G}}$ can be translated back to a set of permutations
$\varphi_l$. For these, there exist unique permutations $\phi_s$, such that
their composition is an automorphism in $\A_\G^\text{loc}$.  
Overall, this shows that there is a bijection 
\begin{align}
    \label{eq:bij}
    \Lambda: \mathcal{L}_{\Sym{G}} \rightarrow \A_\G^\text{loc}
    \,.
\end{align}
For convenience, we still refer to the elements of $\mathcal{L}_{\Sym{G}}$ as
\emph{automorphisms}. In the following, we characterize the group
$\mathcal{L}_{\Sym{G}}$.  

The classification in \cref{sec:locaut} shows that for every element
$(\sigma_{l})_{l \in L_\hex} \in \mathcal{L}_{\Sym{G}}$, the functions
$\sigma_{l}$ have the form 
\begin{align} 
    \label{eq:sigmae}
    \sigma_{l} = \lambda_{g_l} \circ \rho_{h_l} \circ \tau_l
\end{align}
on every link $l \in L_\hex$, for some $g_l, h_l \in G$ and $\tau_l \in
\aut{G}$.  
This expression does not use the full strength of our result in
\cref{sec:locaut}. However, the
form~\eqref{eq:sigmae} has the advantage that there is no distinguished link.
Later, we factor out global group automorphisms such that only left- and
right multiplications remain. 
This step anyway breaks the form of our classification in \cref{sec:locaut}, as
the group automorphism $\tau_l$ can ``leave behind'' a conjugation, so
that keeping track of the explicit forms from \cref{sec:locaut} would only
complicate the proof without much benefit.

Note that in the case $\tau_l = \Id$, the permutation $\sigma_l =
\lambda_{g_l}\circ \rho_{h_l}$ exactly corresponds to the permutation
$\varphi_l(g_l, h_l)$ defined in the main text.

In \cref{sec:model,sec:gsprop} we constructed the plaquette automorphisms
$\Theta_p(g) \in \A^\text{loc}_\G$. Via the bijection $\Lambda$, these can be
translated to elements of $\mathcal{L}_{\Sym{G}}$. We refer to these maps as
$\boldsymbol{\Theta}_p(g) := \Lambda^{-1}(\Theta_p(g))$. 
Pictorially they are represented as follows:
\begin{center}
    \includegraphics[width = 0.5\linewidth]{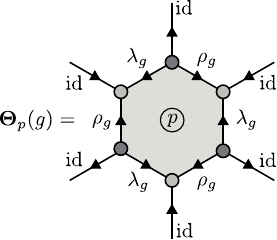}
\end{center}
It is straightforward to verify that the maps $\boldsymbol{\Theta}_p(g)$
satisfy \cref{eq:cons1} and \cref{eq:cons2}. We denote the subgroup of
$\mathcal{L}_{\Sym{G}}$ that is generated by plaquette automorphisms as
$\mathcal{P}_G$.

Lastly, we define the natural group action of $\mathcal{L}_{\Sym{G}}$ on
$\mathcal{L}_G$ by
\begin{align}
    \mathcal{L}_{\Sym{G}} \times \mathcal{L}_G &\rightarrow \mathcal{L}_G
    \,,\nonumber\\
    \boldsymbol{\sigma} \cdot (g_l)_{l\in L_\hex} &\mapsto (\sigma_l(g_l))_{l \in L_\hex}
    \,.
    \label{eq:grpact}
\end{align}
In the following, we omit the dot and write the group action just as
juxtaposition. Every subgroup of $\mathcal{L}_{\Sym{G}}$ then also induces a
natural group action on $\mathcal{L}_G$ by \cref{eq:grpact}.

\subsection{Orbits of $\mathcal{L}_G$ under plaquette automorphisms}%

In this section, we review a result by Cui~\etal~\cite[Theorem 2.4]{Cui_2020}
concerning the number of orbits in $\mathcal{L}_G$ under the group action of
plaquette automorphisms $\mathcal{P}_G$.

To characterize the orbits, it is useful to introduce another group action
\begin{align}
    & G \times \Hom{\pi_1(\TT,p_0), G} \rightarrow \Hom{\pi_1(\TT,p_0), G}
    \,,\nonumber\\
	& g \cdot\ \psi \mapsto (\gamma \mapsto g \psi(\gamma) g^{-1})\,.\label{eq:grpactconj}
\end{align}
Here, $\pi_1(\TT,p_0)$ denotes the first homotopy group of the torus with base
point $p_0$, i.e., $\pi_1(\TT,p_0) \simeq \Z^2$. We use a plaquette
$p_0$ as base point, since we consider paths on the dual lattice
$\tilde{\Lambda}_\hex = (\tilde{S}_\hex, \tilde{L}_\hex, \tilde{P}_\hex)$ with
$\tilde{S}_\hex = P_\hex$, $\tilde{P}_\hex =V_\hex$ and $\tilde{L}_\hex =
\{\{p,p'\}\, |\, p,p' \in \tilde{S}_\hex\} \simeq L_\hex$. That is, the base
point $p_0$ is a plaquette of the original lattice $\Lambda_\hex$. For two
groups $G, H$, $\Hom{G, H}$ denotes the set of group homomorphisms from $G$
to $H$. 
We denote the set of orbits under this group action by $\Hom{\pi_1(\TT,p_0),
G}/G$. In the following, $X = \{\psi_1,\ldots,\psi_m\}$ denotes a set of
representatives of these orbits.

As introduced in \cref{sec:fluxlattice}, we endow the dual lattice with an
orientation in the following way. Let $\gamma = (\ldots, p_1, p_2, \ldots)$ be
an oriented path on the dual lattice. This path contains the dual
edge $\{p_1,p_2\}$, which corresponds uniquely to an edge $l$ of the original
lattice. We say that $\gamma$ crosses the edge $l$. We set
$\text{sign}(\{p_1,p_2\}, \gamma) = 1$ if the crossed edge $l$ points to the
left of $\gamma$, and we set $\text{sign}(\{p_1,p_2\}, \gamma) = -1$ if the
crossed edge $l$ points to the right of $\gamma$. [So $\text{sign}(\{p_1,p_2\},
\gamma) = \sigma_{\{p_1,p_2\}}$ in the notation of the main text, recall
\cref{fig:fig4}.]

For every path $\gamma$ on the dual lattice, we define the \emph{directed
product} as $g_\gamma := \prod_{l \in \gamma} g_l^{\text{sign}(l, \gamma)}$.
Let $s \in \tilde{P}_\hex = S_\hex$ be a dual plaquette (= a site of
the honeycomb lattice). Let $\gamma_s$ be a path that traverses the boundary
of this dual plaquette in counterclockwise direction; we denote the crossed
links as $l_1, l_2, l_3 \in L_\hex$. 
If the links are directed inwards at $s$, then $g_{\gamma_s} =
g_{l_1}g_{l_2}g_{l_3}$.  If the links are directed outwards at $s$, then
$g_{\gamma_s} = g_{l_1}^{-1}g_{l_2}^{-1}g_{l_3}^{-1}$. In both cases, and for
configurations in $\mathcal{L}_G$, \cref{eq:zf1,eq:zf2} imply that
$g_{\gamma_s} = 1$. It is easy to see that this implies $g_\gamma = 1$ for
every contractible loop $\gamma$. This observation is crucial for the proof by
Cui~\etal~\cite{Cui_2020}.

Let $T$ be a spanning tree of the dual lattice (i.e., a subgraph
without cycles that connects to every site), and $p_1, p_2 \in P_\hex$
some plaquettes such that $l := \{p_1,p_2\} \notin T$. Then there exist unique
paths $\gamma_1, \gamma_2 \subseteq T$ that connect $p_0$ to $p_1$ and $p_2$
and do not contain duplicate edges. 
For each $\psi \in X$, we define the group elements
\begin{align}
    g_{\{p_1,p_2\}}^\psi 
    := 
    \left[ 
        \psi(\gamma_1 \circ (p_1,p_2) \circ \gamma_2^{-1})
    \right]^{-\text{sign}(\{p_1,p_2\},(p_1,p_2))}
    \,.
\end{align}
In this equation, $(p_1, p_2)$ is interpreted as the path on the dual lattice
that starts at $p_1$ and ends at $p_2$; ``$\circ$'' denotes the
concatenation of paths on the dual lattice.
With these group elements, we define a configuration $(g_l^\psi)_{l \in L_\hex}
\in \mathcal{L}_G$ by
\begin{align}
    g_l^\psi := 
    \begin{cases}
        1, & l \in T\,,\\
        g_{\{p_1, p_2\}}^\psi, & l = \{p_1,p_2\} \notin T \,.
    \end{cases}
\end{align}
This definition ensures that for every closed path $\gamma$ on the dual lattice
which contains $p_0$, the directed product satisfies
\begin{align}
    \psi(\gamma) 
    = \prod_{l \in \gamma} (g_l^\psi)^{\textup{sign}(l,\gamma)}
    \,.
\end{align}
Finally, we can formulate the theorem from Ref.~\cite{Cui_2020}:
\begin{theorem}[Cui~\etal~\cite{Cui_2020}]
    \label{thm:Cui}
    The elements $(g_l^\psi)_{l \in L_\hex} \in \mathcal{L}_G$ for $\psi \in X$
    form a set of representatives of the orbits $\mathcal{L}_G/\mathcal{P}_G$.
    A given element $\boldsymbol{h} := (h_l)_{l \in L_\hex} \in \mathcal{L}_G$
    is in the same orbit as $(g^\psi_l)_{l \in L_\hex}$ if and only if the map
    \begin{align}
        \label{eq:defpsi}
        \Psi_{\boldsymbol{h}}: 
        \pi_1(\TT,p_0) \rightarrow G\,,\quad
        \gamma \mapsto \prod_{l \in \gamma} h_l^{\textup{sign}(l,\gamma)} 
    \end{align}
    lies in the same orbit as $\psi$ in $\Hom{\pi_1(\TT,p_0),G}$. 
\end{theorem}
We later use this theorem to obtain a characterization of the automorphisms 
that have the form $\sigma_l = \lambda_{z_l}$ for $z_l \in Z(G)$ (see \cref{sec:class_center}).

\subsection{Global automorphisms and loop automorphisms}%
\label{sec:class_center}

In this section, we introduce global automorphisms and general loop
automorphisms. Together with plaquette automorphisms, these form the building
blocks of the automorphism group $\mathcal{L}_{\Sym{G}}$ (recall \cref{eq:bij}).

Let $\tau \in \Aut{G}$ be a group automorphism, then we define the \emph{global
automorphism} $\boldsymbol{\tau} \in \mathcal{L}_{\Sym{G}}$ by
\begin{align}
	\label{eq:tauglob}
   	\boldsymbol{\tau} := (\tau)_{l \in L_\hex}
    \,,
\end{align}
i.e., the map $\tau$ is associated to every link. Conditions \eqref{eq:cons1}
and \eqref{eq:cons2} follow directly from $\tau$ being a group automorphism.

To define general \emph{loop automorphisms}, consider $(z_l)_{l \in L_\hex} \in
\mathcal{L}_{Z(G)}$, i.e., a configuration of elements from the center $Z(G)$
of $G$, such that for every site with links $l_1,l_2,l_3$ (listed in
counterclockwise order) the constraint
\begin{align}
    \label{eq:zcons1}
    z_{l_1}z_{l_2}z_{l_3} = 1
\end{align}
is satisfied. Note that as $Z(G)$ is abelian no case distinction is needed.
Also note that in contrast to $\mathcal{L}_G$, $\mathcal{L}_{Z(G)}$ forms a group
under component-wise multiplication because $Z(G)$ is abelian. Such a
configuration defines an element in $\mathcal{L}_{\Sym{G}}$ by left
multiplication with $z_l$ on every link $l \in L_\hex$. 
Hence, we have an injective group homomorphism
\begin{align}
    \label{eq:center_aut}
    \Gamma: \mathcal{L}_{Z(G)} \rightarrow \mathcal{L}_{\Sym{G}}
    \,,\quad
    (z_l)_{l\in L_\hex} \mapsto (\lambda_{z_l})_{l\in L_\hex}\,.
\end{align}
The map $\lambda_g$ for $g \in G$ is defined in \cref{eq:lambda}.
The fact that $\Gamma(\vec{z}) \in \mathcal{L}_{\Sym{G}}$ for $\vec{z} \in
\mathcal{L}_{Z(G)}$ follows from \cref{eq:zcons1} and because
elements of the center commute with all group elements. 
$\Gamma$ is a group homomorphism since for all $x \in G$
\begin{align}
	\lambda_{z z'}(x) 
  = z z' x 
  = z \lambda_{z'}(x) 
  = (\lambda_z \circ \lambda_{z'})(x)
  \,.
\end{align}
Finally, if $\Gamma(\vec{z}) = (\Id)_{l \in L_\hex}$, then it follows that for
all $x \in G$ it is $x = \lambda_{z_l}(x) = z_l x$ and thus $z_l = 1$. Hence
$\text{ker}(\Gamma) = \vec{1}$, which shows injectivity.

The definition of $\Gamma$ immediately shows that for $\vec{y}, \vec{z}\in
\mathcal{L}_{Z(G)}$ it holds that $\Gamma(\vec{y})\vec{z} = \vec{y}\vec{z}$.
Here, the right-hand side is to be understood as component-wise multiplication
in $\mathcal{L}_{Z(G)}$. 
This allows us to define a subgroup of $\mathcal{L}_{\Sym{G}}$ as
$\mathcal{Z}_{\Sym{G}} := \Gamma(\mathcal{L}_{Z(G)})$. The previous comment and
the homomorphism property show that for $\vec{z} \in \mathcal{L}_{Z(G)}$ and
$\boldsymbol{\sigma} \in \mathcal{Z}_{\Sym{G}}$, it holds that
$\Gamma(\boldsymbol{\sigma} \vec{z}) = \boldsymbol{\sigma} \Gamma(\vec{z})$.
Hence, we find that $\Gamma(\boldsymbol{\sigma}\vec{1}) = \boldsymbol{\sigma}$.

A special class of automorphisms in
$\mathcal{Z}_{\Sym{G}}$ are \emph{loop automorphisms}. 
Let $\ell = (l_1,l_2,\ldots)$ be an arbitrary directed, closed loop on the lattice $\Lambda_\hex$.
For $l \in \ell$ we write $\ell \uparrow\uparrow l$ if the direction of $\ell$ 
matches the direction of the link $l$ and otherwise we write $\ell \uparrow\downarrow l$.
Then we define $\boldsymbol{\Theta}_\ell(z) \in \mathcal{Z}_{\Sym{G}}$ by
\begin{align}
    \label{eq:loop}
    [\boldsymbol{\Theta}_\ell(z)]_l := 
    \begin{cases} 
	\lambda_{z} & \text{for $l \uparrow\uparrow \ell$}\,,\\
	\lambda_{z^{-1}} & \text{for $l \uparrow\downarrow \ell$}\,,\\
        \operatorname{id} & \text{for $l \notin \ell$}\,,
    \end{cases}
\end{align}
for every $z \in Z(G)$. It is easy to check, that $\boldsymbol{\Theta}_\ell(z)$ 
satisfies \cref{eq:cons1,eq:cons2} and thus defines an automorphism in 
$\mathcal{L}_{\Sym{G}}$. 
Note that if $\ell$ is the boundary of a plaquette, we recover the plaquette
automorphisms, in this case the definition~\eqref{eq:loop} defines an 
automorphism for every $z \in G$.
However, in the general case, \cref{eq:loop} only defines an automorphism if $z
\in Z(G)$. This is consistent with Ref.~\cite{Maier2025} where some of us found
loop automorphisms for every closed loop when working with the abelian group
$\Z_2$.

We now show that $\mathcal{Z}_{\Sym{G}}$ is generated by plaquette
automorphisms and loop automorphisms along non-contractible loops.

\begin{figure}[tbp]
    \centering
    \includegraphics[width = 1.0\linewidth]{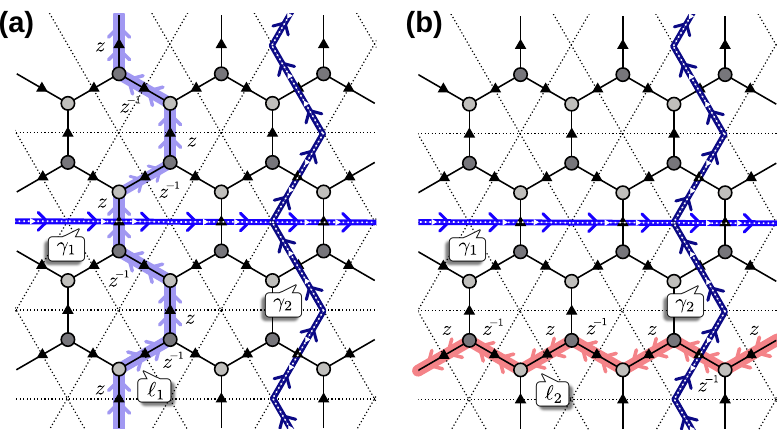}
    \caption{%
        \textbf{Generators of $\pi_1(\TT,p_0)$}
        \textbf{(a)} The dual paths $\gamma_1$ and $\gamma_2$ generate the
        group $\pi_1(\TT,p_0)$. Applying the loop automorphism
        $\boldsymbol{\Theta}_{\ell_1}(z)$ on an arbitrary configuration in
        $\mathcal{L}_{Z(G)}$ multiplies every edge along $\ell_1$ with $z$
        ($z^{-1}$) if the edge points in the same (opposite) direction as
        $\ell_1$. The directed product of group elements along $\gamma_2$
        remains unchanged, whereas the directed product along $\gamma_1$ is
        multiplied by $z$.
        \textbf{(b)} Applying $\boldsymbol{\Theta}_{\ell_2}(z)$ on any
        configuration in $\mathcal{L}_{Z(G)}$ leaves the directed product along
        $\gamma_1$ unchanged, the directed product along $\gamma_2$ is
        multiplied by $z$.
    }
    \label{fig:loops}
\end{figure}

To achieve this we use the one-to-one correspondence of configurations $\mathcal{L}_{Z(G)}$ and
automorphisms in $\mathcal{Z}_{\Sym{G}}$, i.e., we must classify configurations
in $\mathcal{L}_{Z(G)}$; to this end, we invoke \cref{thm:Cui}.
As $Z(G)$ is abelian, the group action of $Z(G)$ by conjugation [see
\cref{eq:grpactconj}] is trivial and therefore $X =
\Hom{\pi_1(\TT,p_0),Z(G)}$. Now choose two loops
$\gamma_1$ and $\gamma_2$ on the dual lattice that generate the group
$\pi_1(\TT, p_0)$, and two loops $\ell_1$ and $\ell_2$ on the lattice that are
non-contractible and not homotopic, as shown in \cref{fig:loops}. 
Consider the configurations
\begin{align}
    \label{eq:ref_conf}
    \boldsymbol{g}^{z_1,z_2} 
    := \boldsymbol{\Theta}_{\ell_1}(z_1) \boldsymbol{\Theta}_{\ell_2}(z_2) \boldsymbol{1}
    \,,
\end{align}
for arbitrary $z_1, z_2 \in Z(G)$. 
These configurations satisfy 
\begin{subequations}
    \begin{alignat}{3}
        \Psi_{\boldsymbol{g}^{z_1,z_2}}(\gamma_1) 
        &= \prod_{l \in \ell_1} (g^{z_1,z_2}_l)^{\text{sign}(l,\gamma_1)} 
        &&= z_1
        \,,\\
        \Psi_{\boldsymbol{g}^{z_1,z_2}}(\gamma_2) 
        &=\prod_{l \in \ell_2} (g^{z_1,z_2}_l)^{\text{sign}(l,\gamma_2)} 
        &&= z_2
        \,.
    \end{alignat}
\end{subequations}
This shows that all configurations $\boldsymbol{g}^{z_1,z_2}$ belong to
mutually distinct orbits. 
Furthermore, every orbit is reached by this construction because every
homomorphism in $\Hom{\pi_1(\TT,p_0),Z(G)}$ is uniquely defined by the image of
the generators $\gamma_1$ and $\gamma_2$. 
Thus we have shown that every element $\boldsymbol{h} \in \mathcal{L}_{Z(G)}$
can be written as
\begin{align}
    \label{eq:char_center_conf}
    \boldsymbol{h} 
    = \left(
        \boldsymbol{\Theta}_{\ell_1}(z_1) 
        \circ \boldsymbol{\Theta}_{\ell_2}(z_2) 
        \circ \prod_{p \in P_\hex} \boldsymbol{\Theta}_p(z_p) 
    \right)\boldsymbol{1}
    \,,
\end{align}
for suitable $z_1, z_2 \in Z(G)$ and $z_p \in Z(G)$ for $p \in P_\hex$.
The characterization of $\mathcal{Z}_{\Sym{G}}$ follows almost trivially:
\begin{proposition}
    \label{prop:center}
    Every element $\boldsymbol{\sigma} \in \mathcal{Z}_{\Sym{G}}$ has the form
    \begin{align}
        \boldsymbol{\sigma} 
        = \boldsymbol{\Theta}_{\ell_1}(z_1)
        \circ \boldsymbol{\Theta}_{\ell_2}(z_2) 
        \circ \prod_{p\in P_\hex} \boldsymbol{\Theta}_p(z_p) 
    \end{align}
    for some $z_1,z_2 \in Z(G)$ and $z_p \in Z(G)$ for $p \in P_\hex$. 
\end{proposition}
\begin{proof}
    By definition, every element $\boldsymbol{\sigma} \in
    \mathcal{Z}_{\Sym{G}}$ is generated by a configuration $\vec{h} =
    (h_l)_{l\in L_\hex} \in \mathcal{L}_{Z(G)}$, i.e., $\boldsymbol{\sigma} =
    \Gamma(\vec{h})$. As shown, $\vec{h}$ has the form
    \cref{eq:char_center_conf}, thus we obtain
    \begin{subequations}
        \begin{align}
            \Gamma(\vec{h})
            &= \boldsymbol{\Theta}_{\ell_1}(z_1)
            \circ \boldsymbol{\Theta}_{\ell_2}(z_2)
            \circ \prod_{p \in P_\hex} \boldsymbol{\Theta}_p(z_p) 
            \circ \Gamma(\vec{1})
            \\
            &= \boldsymbol{\Theta}_{\ell_1}(z_1)
            \circ \boldsymbol{\Theta}_{\ell_2}(z_2)
            \circ \prod_{p \in P_\hex} \boldsymbol{\Theta}_p(z_p)
        \end{align}
    \end{subequations}
    which concludes the proof.
\end{proof}

\subsection{Characterization of $\mathcal{L}_{\Sym{G}}$}
\label{sec:caracterization}

We are now equipped to characterize the automorphism group
$\mathcal{L}_{\Sym{G}}$. As a first step, we show that the group automorphism
in \cref{eq:sigmae} can by removed by pre-composing a suitable group
automorphism on each link.

\begin{lemma}
    \label{lm:reduc}
    Let $\vec{\sigma} \in \mathcal{L}_{\Sym{G}}$, then there exists $\tau \in
    \Aut{G}$ such that for each link $l \in L_\hex$, $\sigma_l \circ \tau =
    \lambda_{g_l} \circ \rho_{h_l}$ for some $g_l,h_l \in G$.
\end{lemma}
\begin{proof}
    Let $l \in L_\hex$ be an arbitrary link; by \cref{eq:sigmae} we have
    $\sigma_l = \lambda_{\tilde{g}_l}\circ \rho_{\tilde{h}_l}\circ \tau$ with
    $\tau \in \Aut{G}$ and some $\tilde g_l,\tilde h_l\in G$. 
    Then
    \begin{align}
	    \sigma_l \circ \tau^{-1}
        = \lambda_{\tilde{g}_l}
        \circ \rho_{\tilde{h}_l}
        \,.
    \end{align}
    Suppose that $l, l_2, l_3$ are the emanating links of site $s$ and all
    links are directed inwards. Then \cref{eq:cons1} implies that for all $x
    \in G$
    \begin{subequations}
        \begin{align}
            \sigma_{l_3}(x) 
            &= \sigma_{l_2}(1)^{-1} \sigma_{l}(x^{-1})^{-1}\\
            &= \sigma_{l_2}(1)^{-1} [\tilde{g}_l \tau(x^{-1}) \tilde{h}_l^{-1}]^{-1}\\
            &= [\sigma_{l_2}(1)^{-1}\tilde{h}_l] \tau(x) \tilde{g}_l^{-1}\,,
        \end{align}
    \end{subequations}  
    and therefore
    \begin{align}
	    \sigma_{l_3} \circ \tau
      = 
      \lambda_{\sigma_{l_2}(1)^{-1}\tilde{h}_l} 
      \circ \rho_{\tilde{g}_l}
      \,.
    \end{align}
    This argument also holds for $l_2$ and for sites with outward directed
    links.
    As the honeycomb lattice is a connected graph, every link $l' \in L_\hex$
    can be connected to $l$ by a sequence of links and sites. For every site in
    this sequence, the argument above applies. This concludes the proof.
\end{proof}

\cref{lm:reduc} shows that we only have to characterize the graph automorphisms
that arise from pure multiplications. This motivates the definition of the
subgroup
\begin{align}
    \mathcal{L}_{\Sym{G}}^\mathrm{red} :=\,
       &\{\boldsymbol{\sigma} \in \mathcal{L}_{\Sym{G}}\, |\, \forall l \in L_\hex:
        \nonumber\\
       &\qquad\sigma_l = \lambda_{g_l} \circ \rho_{h_l}\;\text{for $g_l, h_l \in G$}\}\,.
    \label{eq:Lsymred}
\end{align}
Before we proceed to the main part of the characterization of 
$\mathcal{L}_{\Sym{G}}^\text{red}$, we prove a technical lemma:
\begin{lemma}
    \label{lm:center}
    Let $g, h \in G$. If for all $x \in G$ it holds $x = gxh$, then $g =
    h^{-1}$ and $g, h \in Z(G)$.
\end{lemma}
\begin{proof}
    For $x = 1$ we obtain $1 = gh$, which implies $g = h^{-1}$. Then it follows
    that for all $x \in G$:
    \begin{align}
        x = gxh = gxg^{-1} 
        \quad\Rightarrow\quad
        g \in Z(G)\,.
    \end{align}
\end{proof}

\cref{lm:center} allows us to prove the first part of the characterization of $\mathcal{L}_{\Sym{G}}^\text{red}$:

\begin{lemma}
	\label{lm:equality}
	Let $\boldsymbol{\sigma} = (\sigma_l)_{l \in L_\hex} \in \mathcal{L}_{\Sym{G}}^\text{red}$.
	Then for every site $s \in S_\hex$ with inwards pointig links 
	and $l_1,l_2,l_3 \in L_\hex$ its emanating links 
	listed in counterclockwise order, the permutations $\sigma_{l_i}$ 
	for $i \in \{1,2,3\}$ have the form
	\begin{align}
		(\sigma_{l_1}, \sigma_{l_2},\sigma_{l_3}) = (\lambda_{g_1}\circ\rho_{g_2}, \lambda_{g_2}\circ\rho_{g_3}, \lambda_{g_3}\circ\rho_{g_1} )\,.
	\end{align}
	for some $g_1,g_2,g_3 \in G$.
\end{lemma}
\begin{proof}
	Let $s \in S_\hex$ be a site with inwards pointing links and $l_1,l_2,l_3$ its emanating 
	links listed in counterclockwise order.
	By definition, $\boldsymbol{\sigma}$ satisfies \cref{eq:cons1}, hence the group elements 
	$g_l, h_l \in G$ [see \cref{eq:Lsymred}] satisfy
	\begin{align}
		h_{l_3} gh  g_{l_3}^{-1} = g_{l_1} g h_{l_1}^{-1} g_{l_2} h h_{l_2}^{-1}
	\end{align}
	for all $g, h \in G$. Setting $h = 1$ and invoking \cref{lm:center}, we obtain 
	\begin{align}
		\label{eq:mult1}
		h_{l_3}^{-1}g_{l_1} = (h_{l_1}^{-1}g_{l_2} h_{l_2}^{-1}g_{l_3})^{-1} \in Z(G).
	\end{align}
	Analogously, for $g = 1$ we obtain 
	\begin{align}
		h_{l_3}^{-1}g_{l_1} h_{l_1}^{-1}g_{l_2} = (h_{l_2}^{-1}g_{l_3})^{-1} \in Z(G).
	\end{align}
	This shows that $z_{12} := h_{l_1}^{-1}g_{l_2} \in Z(G)$, $z_{13} := h_{l_3}^{-1}g_{l_1} \in Z(G)$
	and $z_{23} := h_{l_2}^{-1}g_{l_3} \in Z(G)$. Furthermore, from \cref{eq:mult1} it follows that these 
	group elements satisfy $z_{13} = (z_{12}z_{23})^{-1}$.
	Thus, we find that 
	\begin{align}
		\sigma_{l_2}& = \lambda_{g_{l_2}}\circ \rho_{h_{l_2}} = \lambda_{g_{l_2} z_{12}^{-1}}\circ \rho_{h_{l_2} z_{12}^{-1}}\nonumber\\
		&= \lambda_{h_{l_1}} \circ \rho_{h_{l_2} z_{12}^{-1}}
	\end{align}
	and
	\begin{align}
		\sigma_{l_3} &= \lambda_{g_{l_3}} \circ \rho_{h_{l_3}} = \lambda_{g_{l_3}z_{23}^{-1}z_{12}^{-1}} \circ \rho_{h_{l_3}z_{23}^{-1}z_{12}^{-1}}\nonumber\\
		&= \lambda_{h_{l_2}z_{12}} \circ \rho_{h_{l_3} z_{13}} = \lambda_{h_{l_2}z_{12}} \circ \rho_{g_{l_1}}.
	\end{align}
	In conclusion we have shown that 
	\begin{align}
		&(\sigma_{l_1}, \sigma_{l_2},\sigma_{l_3})\nonumber \\
		& \quad = (\lambda_{g_{l_1}} \circ \rho_{h_{l_1}},\, \lambda_{h_{l_1}} \circ \rho_{h_{l_2} z_{12}^{-1}},\, \lambda_{h_{l_2}z_{12}} \circ \rho_{g_{l_1}})\, ,
	\end{align}
	as desired. 
\end{proof}
Now we can prove the main result:
\begin{proposition}
    Every automorphism $\boldsymbol{\sigma} \in
    \mathcal{L}^\mathrm{red}_{\Sym{G}}$ has the form 
    \begin{align}
        \boldsymbol{\sigma} 
        = \boldsymbol{\Theta}_{\ell_1}(z_1) 
        \circ \boldsymbol{\Theta}_{\ell_2}(z_2) 
        \circ \prod_{p\in P_\hex} \boldsymbol{\Theta}_p(g_p)
        \,,
    \end{align}
    for some group elements $g_p \in G$ and elements from the center $z_1, z_2
    \in Z(G)$.
\end{proposition}
\begin{proof}
Let $S^\text{i}_\hex$ be the set of all sites with links directed
inwards. The links of the honeycomb lattice (disregarding 
their direction) are parallel to one of
three possible directions, we denote them $\hat{e}_1, \hat{e}_2, \hat{e}_3$.
The three links emanating from one site are in one-to-one correspondence with
the three directions. We now imagine deleting all links parallel to 
$\hat{e}_1$. We denote the resulting set of links as $L'_\hex$. Then,
for each site $s \in S^\text{i}_\hex$, we define $p_s$ to be the unique
plaquette that contains the two links in $L'_\hex$ that emanate from $s$. 

Now consider an arbitrary element $\boldsymbol{\sigma} \in
\mathcal{L}^\mathrm{red}_{\Sym{G}}$ and some site $s \in S^\text{i}_\hex$. Let
$l_1,l_2,l_3$ denote the links emanating from $s$ (listed in counterclockwise
order) such that $l_1 \notin L'_\hex$.
By \cref{lm:equality} we know that  
\begin{align}
    (\sigma_{l_1}, \sigma_{l_2}, \sigma_{l_3}) 
    = (\lambda_{g_1}\circ\rho_{g_2}, \lambda_{g_2}\circ\rho_{g_3}, \lambda_{g_3}\circ\rho_{g_1} )\,.
\end{align}
Now we post-compose $\boldsymbol{\sigma}$ with the plaquette automorphism
$\boldsymbol{\Theta}_{p_s}(g_3^{-1})$, i.e., we define the new automorphism
$\boldsymbol{\sigma}' := \boldsymbol{\Theta}_{p_s}(g_3^{-1}) \circ
\boldsymbol{\sigma}$. 
Then the maps associated to $l_1$, $l_2$ and $l_3$ are given by 
\begin{align}
    \label{eq:red_triple}
    (\sigma'_{l_1}, \sigma'_{l_2}, \sigma'_{l_3}) 
    = (\lambda_{g_1}\circ\rho_{g_2}, \lambda_{g_2}, \rho_{g_1} )
    \,.
\end{align}
This construction can be illustrated as follows:
\begin{center}
    \includegraphics[width = 0.8\linewidth]{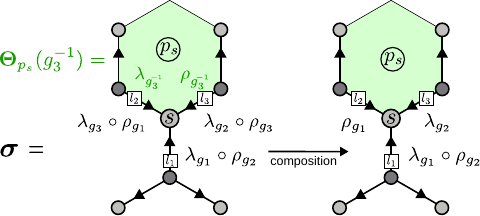}
\end{center}
Note that pre-composing another plaquette automorphism on some of the
neighboring plaquettes changes only the group elements $g_1$ and $g_2$ in
\cref{eq:red_triple} (but not the overall structure) so that this procedure can
be repeated for all sites $s \in S^\text{i}_\hex$. 
Thus we have shown that every $\boldsymbol{\sigma} \in
\mathcal{L}^\text{red}_{\Sym{G}}$ has the form 
\begin{align}
    \label{eq:sigma_tilde}
    \boldsymbol{\sigma} 
    = \left(
        \prod_{p\in P_\hex} \boldsymbol{\Theta}_{p}(g_p) 
    \right)
    \circ \boldsymbol{\sigma}'
\end{align}
for suitable $g_p \in G$ and $\boldsymbol{\sigma}' \in
\mathcal{L}^\text{red}_{\Sym{G}}$ which satisfies \cref{eq:red_triple} for
every site with inward pointing links. It remains to characterize the latter.

Consider a site $s' \notin S^\text{i}_\hex$; it has three neighboring
sites in $S^\text{i}_\hex$. On the links connected to these neighboring sites,
the associated maps satisfy \cref{eq:red_triple}. 
This can be illustrated as follows:
\begin{align}
    \label{eq:neigh}
    \begin{aligned}
   	 \includegraphics[width = 0.4\linewidth]{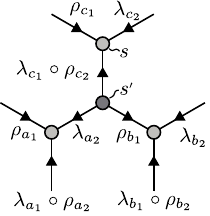}
    \end{aligned}
\end{align}
Note that here the vertical links are the ones that are parallel to
$\hat{e}_1$. For the three links emanating from
$s'$, \cref{eq:cons2} must be satisfied and we find for all $x, y \in G$
\begin{subequations}
    \begin{align}
        \lambda_{a_2}((xy)^{-1}) 
        & = (\lambda_{c_1}\circ\rho_{c_2}(x) \rho_{b_1}(y))^{-1}\\
        \Rightarrow\quad a_2 (xy)^{-1} 
        & = (c_1 x c_2^{-1} y b_1^{-1})^{-1}\\
        &= b_1 y^{-1} c_2 x^{-1} c_1^{-1}
        \,.
    \end{align}
\end{subequations}
For $x = 1$ and $y = 1$, this implies
\begin{subequations}
\begin{align}
    a_2 y^{-1}&= b_1 y^{-1} c_2 c_1^{-1}\\
    \text{and}\quad
    a_2 x^{-1}&= b_1c_2 x^{-1} c_1^{-1}\,.
\end{align}
\end{subequations}
Hence, by \cref{lm:center} it follows 
\begin{align}
    c_2 c_1^{-1} \in Z(G)
    \quad\text{and}\quad
    c_1^{-1} \in Z(G)
\end{align}
and therefore $c_1, c_2 \in Z(G)$. 
This argument applies to every site $s \in S^\text{i}_\hex$, as we can
always consider the site $s'\notin S^\text{i}_\hex$ that is adjacent to
$s$ such that the link connecting $s$ and $s'$ points is parallel to $\hat{e}_1$.
Then, for site $s'$ we recover the situation in Fig.~\eqref{eq:neigh}, with $s$
the site on the top. 
This shows that $\boldsymbol{\sigma}'$ only contains multiplications by
elements in the center. Therefore all these multiplications can be reordered to be
left-multiplications, i.e., $\boldsymbol{\sigma}' = (\lambda_{z_l})_{l \in L_\hex}$ 
for some $z_l \in Z(G)$. As $\boldsymbol{\sigma}'$ satisfies 
\cref{eq:cons1,eq:cons2}, it directly follows that $(z_l)_{l \in L_\hex}$ must satisfy \cref{eq:zcons1},
which shows that $\boldsymbol{\sigma}' \in \mathcal{Z}_{\Sym{G}}$.
Now we invoke \cref{prop:center} and conclude that
$\boldsymbol{\sigma}'$ has the form
\begin{align}
    \boldsymbol{\sigma}' 
    = \boldsymbol{\Theta}_{\ell_1}(z_1) 
    \circ \boldsymbol{\Theta}_{\ell_2}(z_2) 
    \circ \prod_{p\in P_\hex} \boldsymbol{\Theta}_p(z_p)
\end{align}
for suitable $z_1,z_2,z_p \in Z(G)$. 
Together with \cref{eq:sigma_tilde}, we conclude that 
\begin{align}
    \label{eq:emb}
    \boldsymbol{\sigma} 
    = \boldsymbol{\Theta}_{\ell_1}(z_1) 
    \circ \boldsymbol{\Theta}_{\ell_2}(z_2) 
    \circ \prod_{p\in P_\hex} \boldsymbol{\Theta}_p(g_p)
\end{align}
for some $z_1,z_2 \in Z(G)$ and $g_p \in G$, as desired.
\end{proof}

Combining \cref{eq:emb} with \cref{lm:reduc} we obtain the surjective group homomorphism
\begin{align}
    \label{eq:everything}
    &\iota: (G^{|P_\hex|} \times Z(G)^2)\rtimes \aut{G} \rightarrow \mathcal{L}_{\Sym{G}}
    \,,\nonumber\\
    &\begin{aligned}
        &((g_p)_{p\in P_\hex}, z_1, z_2, \tau) \mapsto 
        \\
        &\qquad\boldsymbol{\Theta}_{\ell_1}(z_1) 
        \circ \boldsymbol{\Theta}_{\ell_2}(z_2)
        \circ \prod_{p\in P_\hex} \boldsymbol{\Theta}_p(g_p) 
        \circ \boldsymbol{\tau}\,.
    \end{aligned}
\end{align}
The group structure on the domain of $\iota$ is defined as
\begin{align}
    &((g_p)_{p\in P_\hex}, z_1, z_2, \tau) \cdot ((g'_p)_{p\in P_\hex}, z'_1, z'_2, \tau')
    \nonumber\\
    =\,&((g_p \tau(g'_p))_{p\in P_\hex}, z_1\tau(z'_1), z_2\tau(z'_2), \tau\tau')
    \,.
\end{align}
Unfortunately, this group homomorphism is not injective (i.e., not an isomorphism),
e.g., for every $z \in Z(G)$
\begin{align}
    \iota((z)_{p\in P_\hex}, 1, 1, \Id) &= (\Id)_{l\in L_\hex}
    \nonumber\\
		& = \iota((1)_{p\in P_\hex}, 1, 1, \Id)\,,
\end{align}
as each link is multiplied from the left with $z$ and from the right with
$z^{-1}$. 
However, this characterization is sufficient, e.g., to determine which
topological sectors are connected by graph automorphisms.

\subsection{Automorphisms of $\tilde\G$ with generalized site graphs}
\label{sec:gen}

In this section we comment on the automorphism group of $\tilde\G$ if the
latter contains site graphs $\G_C^s$ for arbitrary conjugacy classes $C$ or
universal site graphs $\G_\text{Cl}^s$. In particular, we show 
that the automorphisms $\boldsymbol{\Theta}_p(h)$ and
$\boldsymbol{\Theta}_\ell(z)$ defined in
\cref{sec:caracterization} carry over to $\tilde\G$.

Recall [\cref{eq:AutomorphismCondition}] that a set of maps $\sigma_1,
\sigma_2, \sigma_3$ and $\sigma_c$ that satisfy the condition
\begin{align}
	\label{eq:cpy_condition}
	\sigma_s(c) = \sigma_1(g_1)\sigma_2(g_2)\sigma_3((g_1g_2)^{-1}c)\,,
\end{align}
for all $g_1, g_2 \in G$ and $c \in C$, describes an automorphism of the site
graph $\G_C^s$. 
Thus to define an automorphism 
of $\tilde{G}$ we must not only specify the group permutations on the links ($\sigma_1,\sigma_2,\sigma_3$) 
but also permutations of conjugacy classes on the sites ($\sigma_s$). We say that an automorphism 
$\tilde{\boldsymbol{\sigma}} = ((\tilde{\sigma}_l)_{l \in L_\hex}, 
(\tilde{\sigma}_s)_{s \in S_\hex})$ of $\tilde{G}$
\textit{extends} the automorphism $\boldsymbol{\sigma} = (\sigma_l)_{l \in L_\hex}\in \mathcal{L}_{\Sym{G}}$, 
if $\tilde{\sigma}_l = \sigma_l$ for all $l \in L_\hex$.

We first construct an extension of $\boldsymbol{\Theta}_p(h)$ to $\tilde{G}$. 
Let $s \in S_\hex$ be an arbitrary site such that $\boldsymbol{\Theta}_p(h)$ acts nontrivially 
on $\G_C^s$. Then there are three possibilities:
\begin{subequations}
    \begin{alignat}{6}
        \sigma_1 &= \rho_g\,,&\quad\sigma_2 &= \lambda_g\,,&\quad\sigma_3 &= \Id\,, 
        \label{eq:p1}\\
        \text{or}\quad
        \sigma_1 &=\Id\,,&\quad\sigma_2 &= \rho_g\,,&\quad\sigma_3 &= \lambda_g\,, 
        \label{eq:p2}\\
        \text{or}\quad
        \sigma_1 &= \lambda_g\,,&\quad\sigma_2 &= \Id\,,&\quad\sigma_3 &= \rho_g\,. 
        \label{eq:p3}
    \end{alignat}
\end{subequations}
For all of these there exists a suitable permutation $\sigma_s \in \Sym{C}$, 
such that \cref{eq:cpy_condition} is satisfied.
For \cref{eq:p1,eq:p2} we can choose $\sigma_s = \Id$, for \cref{eq:p3} we can
choose $\sigma_s = \chi_{g^{-1}}$ (which obviously preserves the
conjugacy class $C$). This shows that the plaquette automorphisms 
$\boldsymbol{\Theta}_p(h)$ can be extended to an automorphism of $\tilde{G}$.
The argument is analogous for $\tilde\G$ that contain site graphs
$\G_\text{Cl}^s$.

The situation for loop automorphisms $\boldsymbol{\Theta}_\ell(z)$
for $z \in Z(G)$ is even simpler. In this case, one of the maps on the right-hand side of
\cref{eq:cpy_condition} is $\lambda_z$, one is $\rho_z$, and the remaining one
is $\Id$. As $z$ commutes with all group elements, the left-multiplication by
$z$ cancels with the right-multiplication by $z^{-1}$, so that we can satisfy
\cref{eq:cpy_condition} with $\sigma_s = \Id$. This shows that
$\boldsymbol{\Theta}_\ell(z)$ can be extended to an automorphism of $\tilde\G$.

For global automorphisms that arise from group automorphisms [recall
\cref{eq:tauglob}] the situation is more complicated. 
Consider a graph $\tilde\G$ that contains exactly one site graph $\G^s_C$ with
$C \neq \{1\}$. \cref{eq:cpy_condition} then shows that the global permutation
$\boldsymbol{\tau}$ for $\tau \in \aut{G}$ can be extended to an automorphism
of $\tilde{G}$ if any only if
$\tau$ preserves the conjugacy class $C$. For graphs $\tilde\G$ with more then
one class-$C$ site, $\tau$ must preserve all present conjugacy classes.

\subsection{Notes on loop automorphisms}
\label{sec:notes}

In \cref{sec:gsprop} of the main text, we mentioned the loop permutations
$\Theta_{\ell}(h)$ for $h \in G$ and a generic loop $\ell$ on the honeycomb
lattice. In this appendix, we define these permutations and show that they are
graph automorphisms if $h \in Z(G)$.

We first define $\Theta_\ell(h)$ on the links of the lattice, in analogy with
plaquette automorphisms. For the directed links of the honeycomb lattice
$\Lambda_\hex$, we write $l\!\uparrow\uparrow\!\ell$ if the direction of $\ell$
coincides with the direction of $l$ and $l\!\uparrow\downarrow\!\ell$ otherwise. 
On the former links, $\Theta_\ell(h)$ acts by left-multiplication with $h$,
which corresponds to the permutation $\varphi_l(h, 1)$ as defined in
\cref{eq:varphi}. On the latter links, $\Theta_\ell(h)$ acts by
right-multiplication with $h^{-1}$, which corresponds to the permutation
$\varphi_l(1, h)$. On all links that are not part of the path $\ell$, the
permutation $\Theta_\ell(h)$ acts as the identity. 

Next, we define $\Theta_\ell(h)$ on the sites of the lattice. To this end, we
partition the sites into two subsets $S_1(\ell)$ and $S_2(\ell)$: Let $s\in
S_\hex$ be a site and $l_1, l_2, l_3\in L_\hex$ its emanating links such that
two of these links are part of $\ell$. Without loss of generality, we define
$l_3$ as the link that is \emph{not} part of $\ell$. The links are ordered
clockwise if the links are outward directed at $s$ and anticlockwise otherwise.
Then we define $s\in S_1(\ell)$ if $l_1\!\uparrow\uparrow\!\ell$ and $s \in
S_2(\ell)$ if $l_1\!\uparrow\downarrow\!\ell$. 
This allows us to define 
\begin{align}
    \Theta_\ell(h) =\,&
    \prod_{l\uparrow\uparrow \ell} \varphi_l(h, 1) 
    \prod_{l\uparrow\downarrow \ell} \varphi_l(1, h)
    \nonumber\\
    &\quad \times\prod_{s\in S_1(\ell)} \phi_s(1, h, 1) \prod_{s\in S_2(\ell)} \phi_s(h, 1, h)
    \,.
\end{align}
Note that if $\ell$ is the counterclockwise oriented boundary of a plaquette,
then $S_2(\ell)$ is empty and we recover \cref{eq:plaquette}. 

If $z \in Z(G)$, then for every link $l \in L_\hex$ it holds that $\varphi_l(z,
z) = \varphi_l(1,1) = \Id$. In this case, $\Theta_\ell(z)$ has the form
\eqref{eq:Phi} when restricted to a site $s \in S_\hex$ and its emanating links
$l_1, l_2, l_3$. This shows that $\Theta_\ell(z)$ it is an automorphism of
$\G$. 
As $z$ is in the center of $G$, $\Theta_\ell(z)$ can be viewed as acting by
left-multiplication on every link. Hence we find
\begin{align}
	\Theta_\ell(z) 
  = \Lambda(\boldsymbol{\Theta}_\ell(z))
  \,,
\end{align}
where $\boldsymbol{\Theta}_\ell(z)$ is the loop automorphism defined in
\cref{sec:class_center}.

Conversely, one can show that an automorphism that acts on the links like
$\Theta_\ell(h)$ does not exist if $h \notin Z(G)$. 

\section{The tessellated blockade structure $\G$ with open boundaries}%
\label{sec:boundary}

In \cref{sec:torus} we discussed the graph $\G$ with periodic boundary
conditions (= on the torus). Another important -- and experimentally more
realistic -- setup are open boundary conditions. In this section, we discuss
modifications of our results for periodic boundaries when $\G$ is defined on a
finite, open patch with \emph{rough} boundary conditions. 

``Rough'' means that we consider a finite patch of the honeycomb lattice
$\Lambda_\hex$ such that each vertex remains trivalent. That is, the graph has
``dangling'' edges on the boundaries. This implies that there are
``incomplete'' plaquettes on the boundary.
For these incomplete plaquettes, we can still define \emph{reduced plaquette
automorphisms} via \cref{eq:plaquette}, by restricting the products to sites
and links that are part of the lattice. 

It is easy to see that these are still graph automorphism in
$\A_\G^\text{loc}$: In \cref{sec:prelim} we established that permutations
$\Phi$ of $V_\G$ (which map the vertex sets $V_l$ and $V_s$ to themselves) are
automorphisms of $\G$ if and only if their restriction $\Phi_s$ to $\G_s$ is an
automorphism of $\G_s$ for all sites $s \in S_\hex$. Since this statement is
independent of boundary conditions, it still applies here. 
The action of a reduced plaquette automorphism on $\G_s$ for some $s \in
S_\hex$ is identical to the action of a full plaquette automorphism. This shows
that the reduced plaquette automorphisms are part of $\A^\text{loc}_\G$. Loop
automorphisms and global automorphisms (derived from group automorphisms) can
be adapted analogously.

The first homotopy group of a finite patch of the plane without holes is
trivial. Thus, in view of our
classification in \cref{sec:torus}, we expect that for rough boundary
conditions, all automorphisms in $\A_\G^\text{loc}$ can be written as a product
of plaquette automorphisms, followed by the global application of a group
automorphism. However, we did not rigorously prove this. (Note that 
\cref{thm:Cui} does not consider manifolds with boundary, however for 
rough boundary conditions the proof of Cui \etal goes through unchanged.)

Lastly, we consider the maximum-weight independent sets of $\G$. The set
defined by restricting \cref{eq:globcons} to the finite lattice is a globally
consistent independent set of $\G$. Hence the MWISs of $\G$ are described by
configurations of group elements $(g_l)_{l \in L_\hex} \in \mathcal{L}_G$ which
satisfy \cref{eq:zf1,eq:zf2} for each site that is part of the lattice. As $\G$
can be embedded on a surface with trivial first homotopy group, invoking
Theorem 2.4 from Ref.~\cite{Cui_2020} (adapted for rough boundary 
conditions) shows that $\mathcal{L}_G$ is a single
orbit under the action of plaquette automorphisms. Thus the blockade structure
$\G$ is fully-symmetric as defined in Ref.~\cite{Maier2025}.

\section{Proof of topological order}
\label{sec:prooftop}

In this appendix, we give the detailed proof that the ground state of $H_\G$ is
topologically ordered for finite $\Omega \neq 0$ (see \cref{sec:gsprop} of the
main text).  

For technical reasons (\cref{sec:technical}) we work with periodic boundary
conditions. Throughout this appendix, we use the following notation. Excitation
patterns of two-level-systems are described by $\vec{n} \in \Z_2^n$ with $n$
the number of two-level-systems in $\G$. An excitation pattern corresponds to a
state $\ket{\vec{n}} \in \H_\G$, these states form a basis of $\H_\G$. The
Hamiltonian $H_\G^0$ is diagonal in this basis. We denote the set of excitation
patterns that correspond to ground states of $H_\G^0$ as $\L_\G$. The ground
state manifold is then given by $\H_\G^0 =
\spn{\ket{\vec{n}}\,|\,\vec{n}\in \L_\G}$.

\subsection{Overview}
\label{sec:overview}

\begin{figure}[tb]
    \centering
    \includegraphics[width=1.0\linewidth]{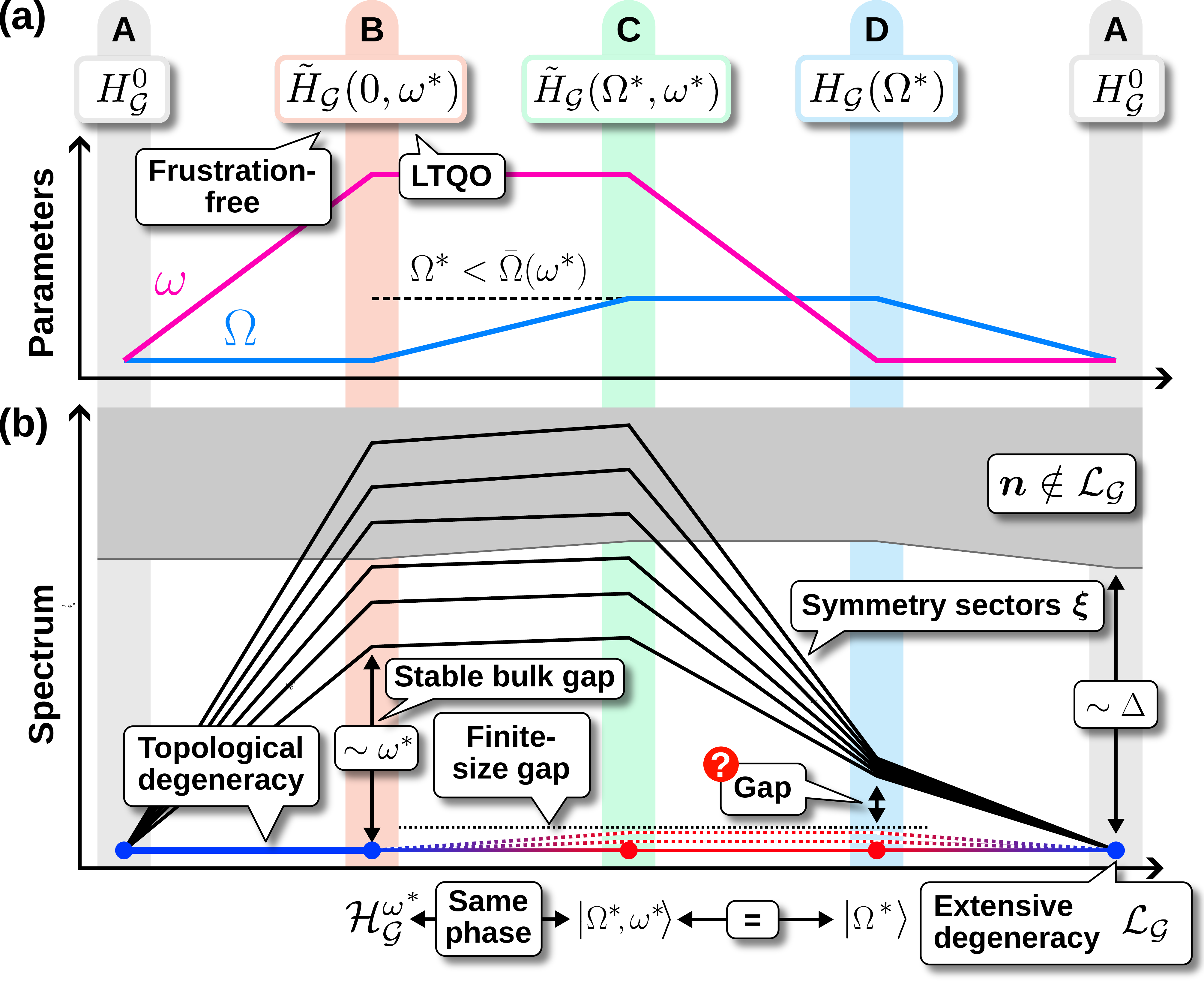}
    \caption{%
        \CaptionMark{Proof of topological order.}
        (a) Path that parameterizes the family of Hamiltonians
        $\tilde{H}_\G(\Omega,\omega)$. Ultimately, we are interested in the
        quantum phase of the ground state in \textbf{D}. Following the path
        from \textbf{A} to \textbf{D} allows us to rigorously characterize the
        quantum phase of the ground state of $H_{\G}(\Omega)$.
        (b) Schematic spectrum of $\tilde{H}_\G(\Omega,\omega)$ along the
        parametric path shown in (a). For $\Omega = \omega = 0$ (\textbf{A}),
        the Hamiltonian $H_{\G}(0) = \tilde H_\G(0,0)$ is classical with an
        exponentially large, degenerate ground state manifold $\H_\G^0$ spanned
        by configurations in $\L_{\G}$. (\textbf{B}) For $\Omega = 0 <
        \omega^{*}$, the Hamiltonian $\tilde H_{\G}(0,\omega^{*})$ is
        frustration-free and satisfies a condition called \emph{local
        topological quantum order} (local-TQO). 
        Moreover, its ground state manifold is separated by a gap of order
        $\omega^*$ from the rest of the spectrum. This ground state
        manifold $\H_\G^{\omega*}$ can be mapped to the ground state manifold
	of Kitaev's quantum double model by a generalized local unitary transformation;
	in particular these states are topologically ordered. 
        Because of frustration-freeness and local-TQO, the bulk gap of the
        Hamiltonian $\tilde H_{\G}(0,\omega^{*})$ is stable against weak, local
        perturbations. Thus, for $\Omega^{*} < \bar{\Omega}(\omega^*)$, the
        Hamiltonian $\tilde{H}_\G(\Omega^{*},\omega^{*})$ (\textbf{C}) remains
        gapped. Here $\bar{\Omega}(\omega^*)$ denotes an upper bound on the
        perturbation strength that guarantees gap stability. 
        As the gap remains open when ramping up $\Omega$, the new (unique)
        ground state $\ket{\Omega^{*},\omega^{*}}$ remains topologically 
	ordered. 
        Lastly, switching off $\omega \searrow 0$ leads to the target
        Hamiltonian $H_{\G}(\Omega^{*})=\tilde H_\G(\Omega^*,0)$, the ground
        state $\ket{\Omega^{*}}$ of which we want to characterize. Due to the
        local symmetries it is $\ket{\Omega^{*}} = \ket{\Omega^{*},
        \omega^{*}}$. 
        Note that this construction does not prove the existence of a gap for
        $H_{\G}(\Omega^{*})$.
    }
    \label{fig:path}
\end{figure}

As discussed in \cref{sec:gsprop}, the graph $\G$ on a torus is not fully
symmetric, i.e., the set of ground state configurations $\L_{\G}$ splits into
multiple orbits $Q_1,\ldots, Q_O$. Thus, as shown in Ref.~\cite{Maier2025}, the
unique ground state for $\Omega \neq 0$ has the form 
\begin{align}
    \ket{\Omega} 
    &= \sum_{k = 1}^O \lambda_k(\Omega) \sum_{\vec{n} \in Q_k} \ket{\vec{n}} 
    + \sum_{\vec{n} \notin \L_\G} \eta_{\vec{n}}(\Omega) \ket{\vec{n}}
    \,.
\end{align}
Note that we have no control over the coefficients $\lambda_k(\Omega)$.

Despite the lack of full symmetry (and the resulting uncontrolled superposition
of topological sectors) we can nevertheless establish that the ground state of
$H_\G$ is topologically ordered. To this end, we generalize the technique
introduced in Ref.~\cite{Maier2025}. In this section, we summarize the main
argument; technical details for some of the steps are provided in subsequent
sections.

As already stated in \cref{sec:gsprop}, we introduce the auxiliary
Hamiltonian~\eqref{eq:exH}
\begin{align}
    \label{eq:H_tilde}
    \begin{aligned}
    \tilde{H}_\G(\Omega,\omega) 
    :=\,&H_{\G}(\Omega) \\
     &+ \omega\sum_{p} 
     \Big(\id-\underbrace{\frac{1}{|G|}\sum_{h} \Theta_p(h)}_{=: \Theta_p}\Big)
     \,.
    \end{aligned}
\end{align}
We say that the term proportional to $\omega$ introduces \emph{artificial
plaquette fluctuations}. This term is chosen to resemble the plaquette term of
the quantum double Hamiltonian~\eqref{eq:QD}.

We establish topological order in three steps (\cref{fig:path}):
\begin{itemize}

    \item \textbf{Step 1: A~\textrightarrow~B}
    
        At \textbf{A} we have $\Omega = \omega = 0$. Thus the ground state
        manifold is extensively degenerate and spanned by the classical ground
        state configurations $\L_{\G}$. The ground state manifold $\H_\G^0$ is
        separated by a gap of size $\min_i(\Delta_i)$ from the rest of the
        spectrum.

        To reach \textbf{B}, we ramp up the artificial plaquette fluctuations
        to some value $0 < \omega^{*} < \Delta$, keeping $\Omega = 0$. As 
        \begin{subequations}
            \begin{align}
                [\tilde{H}_\G(0,\omega^{*}),\Theta_p] &= 0
                \intertext{and}
                [\Theta_p, \Theta_{p'}] &= 0
            \end{align}
        \end{subequations}
        for all plaquettes $p,p'$, the spectrum of $\tilde{H}_\G(0,\omega^{*})$
        decomposes into sectors labeled by eigenvalues of $\Theta_p$ (=
        symmetry sectors). 
        Since $\Theta_p$ is a projector, it has eigenvalues $0$ and $1$. States
        with eigenvalue $0$ are energetically punished by the
        Hamiltonian~\eqref{eq:H_tilde} with energy $\omega^*$. Thus the ground
        states of $\tilde{H}_\G(0,\omega^{*})$ are in the sector characterized
        by $\Theta_p = 1$ for all plaquettes $p$. 
        The ground state manifold $\H_\G^{\omega^*}$ of
        $\tilde{H}_\G(0,\omega^*)$ consists of topologically degenerate ground
        states that are separated by a gap $\omega^*$ from the rest of the
        spectrum. Moreover, it can be shown (see \cref{sec:gs}) that the ground
        state manifold can be mapped to the degenerate ground state manifold
        $\H_G^{J_p}$ of the quantum double model~\eqref{eq:QD} for $J_p>0$ by a
        generalized local unitary transformation. 
	Thus the Hamiltonians $H_G$ and $\tilde{H}_\G(0,\omega^*)$ describe 
	the same quantum phase, in particular the states in $\H_\G^{\omega^*}$ are
        topologically ordered.
    
    \item \textbf{Step 2: B~\textrightarrow~C}
        
        To reach \textbf{C}, we switch on quantum fluctuations with some finite
        value $\Omega^{*} \neq 0$. Note that the term $\Omega^{*} \sum_i
        \sigma_i^x$ in $H_\G(\Omega^*)$ couples sectors with different
        excitation numbers. As a consequence, the Hamiltonian
        $\tilde{H}_\G(\Omega^{*},\omega^{*})$ can no longer be diagonalized
        exactly; in particular the ground state changes in an unknown way. 
        
        However, as long as the excitation gap above the ground state manifold
	does not close when $\Omega$ is ramped up, the Hamiltonian 
	$\tilde{H}_\G(\Omega^*,\omega^*)$ describes the same quantum phase as 
	$\tilde{H}_\G(0,\omega^*)$ and consequantly also the same quantum phase 
	as $H_G$. In particular it follows that the new (now unique)
        ground state $\ket{\Omega^{*},\omega^{*}} \in \H_\G^{\omega^*,
        \Omega^*}$ is topologically ordered.

        Hence the remaining problem is to establish gap stability. We remark
        that this is a priori not obvious, as in general, arbitrarily weak
        perturbations can close spectral gaps in the thermodynamic limit
        \cite{Bravyi2010}. Fortunately, the Hamiltonian
        $\tilde{H}_\G(0,\omega^*)$ is \emph{frustration-free}, \emph{locally
        gapped}, and satisfies a condition called \emph{local topological
        quantum order} (local-TQO). Under these assumptions, it can be shown
        rigorously that the gap is stable under sufficiently weak, local
        perturbations \cite{Michalakis2013}.
        
        Thus we can conclude that $\ket{\Omega^{*},\omega^{*}}$ is 
	topologically ordered, as long as $\Omega^{*} < \bar{\Omega}(\omega^{*})$,
        where $\bar{\Omega}(\omega^{*})$ denotes an upper bound on the
	perturbation strength that guarantees gap stability. 
        For a detailed proof of the gap stability, see \cref{sec:gap}. Note
        that the uniqueness of the ground state $\ket{\Omega^*, \omega^*}$ does
        not contradict the stability of the phase, since ground state
        degeneracies can be lifted by finite-size effects. In particular,
        Ref.~\cite{Michalakis2013} shows that such splittings of the ground
        state degeneracy are exponentially suppressed with system size.
    
    \item \textbf{Step 3: C~\textrightarrow~D}
    
        In the last step, we switch off the artificial plaquette fluctuations
        to reach \textbf{D} with the desired Hamiltonian
        $\tilde{H}_\G(\Omega^{*},0) = H_{\G}(\Omega^{*})$. 
        In contrast to step~2 above, this change of the Hamiltonian cannot be
        treated as a small perturbation since $\omega^{*} \gg \Omega^{*}$. In
        addition, $\tilde{H}_\G(\Omega^{*},\omega^{*})$ is \emph{not}
        frustration-free, so the result of Ref.~\cite{Michalakis2013} is not
        applicable. (We are not aware of any gap stability results for
        frustrated Hamiltonians.)

        Now we utilize the local symmetry projectors $\Theta_p$. By
        construction, these commute with the full Hamiltonian, i.e.,
        $[\tilde{H}(\Omega,\omega),\Theta_p] = 0$. Thus there exists a basis of
        eigenstates of both $\tilde{H}_\G(\Omega,\omega)$ and $\Theta_p$ for
        all plaquettes $p$ and arbitrary $\Omega$ and $\omega$.  By
        \cref{eq:H_tilde}, these states are also eigenstates of
        $H_{\G}(\Omega)$. 
        Thus we can label these eigenstates as $\ket{E^{\vec{\xi}}_\Omega,\vec\xi}$
        with 
        \begin{subequations}
            \begin{align}
                H_{\G}(\Omega) \ket{E^{\vec{\xi}}_\Omega,\vec{\xi}} 
                &= E^{\vec{\xi}}_\Omega \ket{E^{\vec{\xi}}_\Omega,\vec{\xi}}
                \label{eq:eb1}\\
                \Theta_p \ket{E^{\vec{\xi}}_\Omega,\vec{\xi}} 
                &= \xi_p \ket{E^{\vec{\xi}}_\Omega,\vec{\xi}}
                \,,\label{eq:eb2}
            \end{align}
        \end{subequations}
        and $\xi_p \in \{0,1\}$. 
        \cref{eq:eb1} and \cref{eq:eb2} yield an expression for the energy of
        the states $\ket{E^{\vec{\xi}}_\Omega,\vec{\xi}}$:
        \begin{align}
            \label{eq:eigval}
            \begin{aligned}
               &\tilde{H}_\G(\Omega, \omega) \ket{E^{\vec{\xi}}_\Omega,\vec{\xi}} 
               \\
               & = \left[
               E^{\vec{\xi}}_\Omega + \omega\sum_{p}(1-\xi_p)
                \right]\ket{E^{\vec{\xi}}_\Omega,\vec{\xi}}
               \,.
            \end{aligned}
        \end{align}
        
        Consider the ground state $\ket{\Omega^{*}}$ of
        $\tilde{H}_\G(\Omega^{*},0) = H_{\G}(\Omega^{*})$. We cannot directly
        apply Proposition~1 from Ref.~\cite{Maier2025}
        to this Hamiltonian since $\G$ is not necessarily
        fully-symmetric.  However, the proof of this proposition shows that
        $\ket{\Omega^{*}}$ nevertheless satisfies $\Theta_p(h) \ket{\Omega^{*}}
        = \ket{\Omega^{*}}$ for all $h \in G$ and all plaquettes $p$. This
        shows that $\Theta_p \ket{\Omega^*} = \ket{\Omega^*}$ for all
        plaquettes $p$, i.e., $\ket{\Omega^{*}}$ is labeled by $\vec{\xi} =
        \vec{1}$. Moreover, $\ket{\Omega^{*}}$ is defined as eigenvector with
        smallest eigenvalue of $H_{\G}(\Omega^{*})$. Hence \cref{eq:eigval}
        shows that $\ket{\Omega^{*}}$ is \emph{a} ground state of
        $\tilde{H}_\G(\Omega^{*}, \omega)$ for all $\omega$.

        Let $\ket{E^{\vec{\xi}}_\Omega,\vec{\xi}}$ be a ground state of
        $\tilde{H}_\G(\Omega^*, \omega)$. As $\ket{\Omega^*}$ minimizes both
        terms in \cref{eq:eigval} simultaneously, the same must be true for
        $\ket{E^{\vec{\xi}}_\Omega,\vec{\xi}}$. This implies in particular that
        $\ket{E^{\vec{\xi}}_\Omega,\vec{\xi}}$ is a ground state of
        $H_\G(\Omega^*)$. Since the ground state of $H_\G(\Omega^*)$ is unique,
        it follows that $\ket{E^{\vec{\xi}}_\Omega,\vec{\xi}} \propto
        \ket{\Omega^*}$. Hence $\ket{\Omega^*}$ is the unique ground state of
        $\tilde{H}_\G(\Omega^*, \omega)$ for all $0 \leq \omega \leq \omega^*$.
        This implies in particular that $\ket{\Omega^*,\omega^*} =
        \ket{\Omega^*}$, hence $\ket{\Omega^*}$ is topologically ordered.

\end{itemize}

\subsection{Technical details}
\label{sec:technical}

\subsubsection{Ground state manifold of $\tilde{H}_\G(0,\omega)$}
\label{sec:gs}

In this section, we discuss the degenerate ground state manifold of 
\begin{align}
    \label{eq:Homega}
    \tilde{H}_\G(0, \omega) 
    = H_\G^0 + \omega \sum_{p} (\mathds{1} - \Theta_p)
    \,,
\end{align}
defined on a torus. 
In particular, we show that the ground state manifold $\H_\G^{\omega}$ of
the Hamiltonian~\eqref{eq:Homega} can be mapped to the ground state
manifold $\H_G^{J_p}$ of the quantum double model for $J_p > 0$ by a generalized local
unitary (LU) transformation. This shows that both Hamiltonians
represent the same quantum phase~\cite{Chen2010}. Since the
models are defined on different Hilbert spaces, we first must embed them in a
common Hilbert space.

Let $\Lambda_\hex = (S_\hex, L_\hex, P_\hex)$ denote a honeycomb lattice with
sites $S_\hex$, links $L_\hex$ and plaquettes $P_\hex$. Both the quantum double
model and our blockade structure realization are defined on this lattice.

The Hilbert space of the quantum double model has the natural basis
$\ket{\vec{g}} = \ket{(g_l)_{l \in L_\hex}} \in \H_G$ for $g_l\in G$. The
quantum double Hamiltonian~\eqref{eq:QD} for $J_p = 0$ is diagonal in this
basis. Thus we can define the set of ground state configurations as 
\begin{align}
	\L_G := \{(g_l)_{l \in L_\hex}\in G^{L_\hex}\, |\, \ket{\vec{g}} \in \H_G^0 \}\,,
\end{align}
such that the ground state manifold is given by $\H_G^0 = \spn{\ket{\vec
g}\,|\,\vec{g}\in \L_G}$. In this basis, the plaquette operators $A_p(h)$ act as
permutations, thus they define a group action of $G$ on $\L_G$ in the natural
way. By abuse of notation, we denote this group action as $A_p(h)\cdot
\vec{g}$.

By Ref.~\cite[Theorem 2.4]{Cui_2020}, the ground state manifold $\H^{J_p}_G$ of
$H_G$ for $J_p > 0$ is degenerate and its dimension is given by
$|\Hom{\pi_1(\TT, p_0), G}/G| = |\Hom{\Z^2, G}/G|$. Here $/$ denotes the set of
orbits under the group action of $G$ on $\Hom{\pi_1(\TT, p_0)}$ by $g\cdot\pi
\mapsto g^{-1} \pi(\cdot) g$ and $\TT$ refers to the torus on which
the lattice is embedded. $p_0$ denotes an arbitrary plaquette that is used as
the base point for the homotopy group. 
The ground states $\ket{\psi} \in \H_G^{J_p}$ are characterized by
\begin{align}
    A_p \ket{\psi} = B_s \ket{\psi} = \ket{\psi}
\end{align}
for all sites $s \in S_\hex$ and plaquettes $p \in P_\hex$. 
To give an explicit form of the ground states, we define an equivalence
relation $\sim$ on $\L_G$ by $\vec{g} \sim \vec{g}'$ if and only if $\vec{g}$
can be transformed by plaquette operators $A_p(h)$ into $\vec{g}'$. Then a
basis of $\H^{J_p}_G$ is given by 
\begin{align}
    \label{gs:QD}
    \ket{[\vec{g}]} 
    := \frac{1}{|[\vec{g}]|}\sum_{\vec{g}' \sim \vec{g}} \ket{\vec{g}'}
    \,,
\end{align} 
where $[\vec{g}]$ denotes the equivalence class of $\vec{g}$ and $|[\vec{g}]|$
its cardinality.

Next, we construct the common Hilbert space to connect both models. The Hilbert
space of our blockade Hamiltonian $H_\G$ is given by
\begin{align}
	\H_\G 
  := \bigotimes_{l \in L_\hex} \H_l^\G \otimes 
	    \bigotimes_{s \in S_\hex} \H_s^\G
      \,.
\end{align}
where $\H_l^\G \simeq \mathbb{C}^{2^{|G|}}$ and $\H_s^\G \simeq
\mathbb{C}^{2^{|G|^2}}$ denote the Hilbert spaces associated to a link and a
site of the blockade model. By contrast, the Hilbert space associated to one
link of the quantum double model is given by $\H_l^G = \{\ket{g}\,|\, g \in
G\}$. 
We extend the Hilbert space $\H_\G$ by adding $\H_l^G$ to every link. Thus we
obtain the enlarged Hilbert space $\H_\GTG := \H_\G \otimes \bigotimes_{l \in
L_\hex} \H_l^G$. 
We can embed the states $\ket{\vec{g}} \in \H_G$ of the quantum double model
into this larger Hilbert space as
\begin{align}
    \ket{\vec{0}} \ket{\vec{g}} 
    := \bigotimes_{l\in L_\hex} \ket{\vec{0}}_l \ket{g_l}_l 
    \bigotimes_{s\in S_\hex} \ket{\vec{0}}_s
    \,.
\end{align}
In particular, $\ket{\vec{0}} \ket{[\vec{g}]}$ is in the same quantum phase as
$\ket{[\vec{g}]}$ \cite{Chen2010}.  
Let us define the following subspaces of embedded states:
\begin{subequations}
    \begin{align}
    \H_\GTG^{\vec{0}}
      &:= \spn{\ket{\vec{0}} \ket{\vec{g}}\, |\, \vec{g}\in \L_G}
      \,,\\
      \H_\GTG^{\vec{1}}
      &:= \spn{\ket{\vec{n}}\ket{\vec{1}} \, | \, \vec{n} \in \L_{\G}}
      \,.
    \end{align}
\end{subequations}
We now construct an LU quantum circuit that maps $\H_\GTG^\vec{0}$ to
$\H_\GTG^\vec{1}$. 
Let $\L_l^\G$ ($\L_s^\G$) denote the set of ground state excitation patterns of
$H^0_\G$ restricted to the link $l \in L_\hex$ (the site $s \in S_\hex$). By
construction of $\G$, there exists a bijective map $\eta_l: G \rightarrow
\L_l^\G$ which maps each group element to (the restriction of) a ground state
excitation pattern on the Link $l$. 
Moreover, the concatenation of all these excitation patterns $\bigoplus_{l \in
L_\hex} \eta_l(g_l)$ is part of a ground state pattern of $H^0_\G$ if and only
if $\vec{g} = (g_l)_{l \in L_\hex} \in \L_G$. In particular, for every site $s
\in S_\hex$ with emanating links $l_1, l_2, l_3 \in L_\hex$, there exists a map
$\eta_s: (\L_l^\G)^3 \rightarrow \L_s^\G$ that maps $\vec{n}_{l_i} :=
\eta_{l_i}(g_{l_i})$ (for $i = 1,2,3$) to the unique pattern
$\eta_s(\vec{n}_{l_1},\vec{n}_{l_2},\vec{n}_{l_3}) \in \L_s^\G$.

Thus, for some $l \in L_\hex$, we can define the unitary $U^{(1)}_l$ that acts
on the link $l$ as
\begin{align}
	U^{(1)}_l \ket{\vec{0}}_l \ket{g}_l 
  = \ket{\eta_l(g)}_l \ket{\vec{1}}_l\,,
\end{align}
and trivially on all other links. 
Note that this does not define $U^{(1)}_l$ on all of $\H_\GTG$; it can be
extended in an arbitrary way as we are only interested in its application to
the subspace $\H_\GTG^\vec{0}$.

Similarly, for a site $s \in S_\hex$, we define the unitary $U^{(2)}_s$ that
acts on $s$ and its emanating links $l_1, l_2, l_3$ as
\begin{align}
    U^{(2)}_s&\ket{\vec{0}}_s \ket{\vec{n}_{1}}_{l_1}\ket{\vec{n}_{2}}_{l_2}\ket{\vec{n}_{3}}_{l_3}
    \nonumber\\
	  &= \ket{\eta_s(\vec{n}_{1}, \vec{n}_{2}, \vec{n}_{3})}_s 
    \ket{\vec{n}_{1}}_{l_1}\ket{\vec{n}_{2}}_{l_2}\ket{\vec{n}_{3}}_{l_3} 
  \,,
\end{align}
and trivially on all other parts of the tensor product. As before, $U^{(2)}_s$
can be extended arbitrarily to a unitary on $\H_\GTG$. 

As the honeycomb lattice $\Lambda_\hex$ is bipartite, we can partition its
sites $S_\hex$ into two sublattices $A$ and $B$, such that no two sites from
$A$ and $B$ are connected by a link. Then, $\{U^{(1)}_l\}_{l \in L_\hex}$,
$\{U^{(2)}_s\}_{s\in A}$ and $\{U^{(2)}_s\}_{s\in B}$ are sets of unitary
operators that act on non-overlapping regions of finite size. 
Thus $U := \prod_{s \in B} U^{(2)}_s\prod_{s \in A} U^{(2)}_s \prod_l
U^{(1)}_l$ defines a local unitary quantum circuit of constant depth (it has
three layers). By construction, it constitutes the desired map from
$\H_\GTG^\vec{0}$ to $\H_\GTG^\vec{1}$. We define $\eta: \L_G \rightarrow \L_\G$
as the bijective map defined by $U \ket{\vec{0}}\ket{\vec{g}} =
\ket{\eta(\vec{g})}\ket{\vec{1}}$.

Finally, we show that $U$ maps the embedded ground states of the full quantum
double Hamiltonian $H_G$ to the embedded ground states of $\tilde{H}_\G(0,
\omega)$. 
To this end, we define the subspaces
\begin{subequations}
    \begin{align}
    	\H_\GTG^{\vec{0},J_p} 
      &:= \spn{\ket{\vec{0}}\ket{\psi}\, |\, \ket{\psi} \in \H_G^{J_p}}
      \,,\\
    	\H_\GTG^{\vec{1},\omega} 
      &:= \spn{\ket{\omega}\ket{\vec{1}}\, |\, \ket{\omega} \in \H_\G^\omega}
      \,.
    \end{align}
\end{subequations}
Note that by construction the plaquette operators satisfy
\begin{align}
    [\Theta_p(h)\ket{\eta(\vec{g})}]\otimes\ket{\vec{1}} 
    = U\left\{\ket{\vec{0}}\otimes [A_p(h) \ket{\vec{g}}]\right\}
\end{align}
for all plaquettes $p$ and group elements $h \in G$.
Let $\ket{\psi} \in \H_G^{J_p}$ be a quantum double ground state. By
construction of $U$, there exists a state $\ket{\omega(\psi)} \in \H_\G^0$ such
that $U \ket{\vec{0}}\ket{\psi} = \ket{\omega(\psi)} \ket{\vec{1}}$.  
By linearity, it follows that 
\begin{subequations}
    \begin{align}
        [\Theta_p(h)\ket{\omega(\psi)}] \ket{\vec{1}} 
        &= U \ket{\vec{0}}  [A_p(h) \ket{\psi}]\\
        &= U \ket{\vec{0}} \ket{\psi}\\
        &= \ket{\omega(\psi)} \ket{\vec{1}}
    \end{align}
\end{subequations}
which shows that $\ket{\omega(\psi)} \in \H_\G^\omega$. The proof for the
converse direction is analogous.

In summary, this shows that $U$ maps $\H_\GTG^{\vec{0},J_p} =
\spn{\ket{\vec{0}}} \otimes \H_G^{J_p}$ unitarily to $\H_\GTG^{\vec{1},\omega}
= \H_\G^\omega \otimes \spn{\ket{\vec{1}}}$. As $U$ is a LU quantum
circuit with finite depth, this transformation does not change the quantum
phase represented by the states \cite{Chen2010}. Moreover, adding and
removing local degrees of freedom in form of tensor products also does not
alter the topological order \cite{Chen2010}. Thus we have constructed the desired 
generalized local unitary transformation. This shows that the Hamiltonians $H_G$ 
and $\tilde{H}_\G(0,\omega)$ describe the same quantum phase and thus that
the states in $\H_G^{J_p}$ are topologically ordered. 
 
Finally, we give a concrete basis of $\H^{\omega}_{\G}$. 
Let $\sim_\Theta$ denote the equivalence relation on $\L_{\G}$ defined by
$\vec{n} \sim_\Theta \vec{n}'$ if and only if $\vec{n}$ can be transformed into
$\vec{n}'$ by plaquette operators $\Theta_p(h)$. Let $[\vec{n}]_\Theta$ denote
the equivalence class of $\vec{n}$ under this equivalence relation and consider
a set $\{\vec{n}_k\}$ of representatives of all classes. 
Then from \cref{gs:QD} we obtain that 
\begin{align}
    \ket{\omega_k} 
    := \frac{1}{|[\vec{n}_k]_\Theta|} \sum_{\vec{n}'\sim_\Theta\vec{n}_k} \ket{\vec{n}'}
\end{align}
is a basis of $\H^{\omega}_{\G}$.

\subsubsection{Gap stability of $\tilde H_\G(0,\omega)$}
\label{sec:gap}

\paragraph{Conditions.}

To establish the gap stability necessary for step 2
(\textbf{B~\textrightarrow~C}), we utilize a result by Michalakis and Zwolak
\cite[Theorem~1]{Michalakis2013}. A summary of these results and a detailed
explanation of their application to the case $G = \Z_2$ can be found in
Ref.~\cite{Maier2025}. We start with a brief overview of the conditions that
must be verified.

The systems considered by Michalakis and Zwolak are defined on a square lattice
$\Lambda_\sq = (S_\sq, L_\sq, P_\sq)$ with sites $S_\sq = [0,L]^2 \subseteq
\Z^2$, where $L$ denotes the system size. This lattice is endowed with an
arbitrary norm $\|\cdot\|$ (here we choose the $\ell^\infty$-norm). This norm
defines balls centered at $I \in \Lambda_\sq$ of radius $r$ by
\begin{align}
    B_r(I) := \{J \in S_\sq\, |\, \|I - J\| \leq r\}\,.
\end{align}
Note that for the $\ell^{\infty}$-Norm, $B_r(I)$ is a rectangular region.
For each site $I \in S_\sq$, there is an associated Hilbert space $\H_I$. The
complete system Hilbert space is given by the tensor product $\H_{\Lambda_\sq}
= \bigotimes_{I \in S_\sq} \H_I$. 

The Hamiltonian of interest is of the form $H = H_0 + V$, where $H_0$ denotes
the unperturbed Hamiltonian and $V$ the perturbation. The Hamiltonian $H_0$ has
to satisfy the following properties:
\begin{enumerate}

    \item $H_0$ has the form 
        \begin{align}
            \label{eq:decomp}
            H_0 = \sum_{I \in S_\sq} Q_I,
        \end{align}
        such that $Q_I$ has a constant range of support.

    \item The Hamiltonian $H_0$ satisfies periodic boundary conditions.
    
    \item The Hamiltonian $H_0$ is \emph{frustration-free}, i.e., if $P_0$ 
        denotes the projector onto the ground state subspace of $H_0$ and
        $q_{I,0}$ denotes the minimal eigenvalue of $Q_I$, then 
        \begin{align}
            P_0 Q_I = q_{0,I} Q_I\,.
        \end{align}

    \item For $L \geq 2$, the Hamiltonian $H_0$ has a spectral gap that is
        independent of the system size.

\end{enumerate}
In addition, $H_0$ must satisfy the conditions \emph{local-gap} and
\emph{local-TQO}. For a set $A \subseteq S_\sq$, define the localized
Hamiltonian by
\begin{align}
    H_0^A := \sum_{\text{supp}(Q_I) \subseteq A} Q_I
    \,,
\end{align}
where $\text{supp}(Q_I)$ denotes the support of the operator $Q_I$. 
Let $E_0^A$ denote the ground state energy of $H_0^A$. For $\epsilon \geq 0$,
let $P_A(\epsilon)$ denote the projector onto the eigenstates of $H_0^A$
with energy less or equal to $E_0^A + \epsilon$. 
\begin{enumerate}

    \setcounter{enumi}{4}

    \item The \emph{local-gap condition} then states that there exists a
        function $\gamma(r) > 0$, which decays at most polynomially, such that
        for all $I_0 \in S_\sq$, it is $P_{B_r(I_0)}(\gamma(r)) =
        P_{B_r(I_0)}(0)$. 

\end{enumerate}

To define \emph{local topological quantum order (local-TQO)}, let $I_0 \in
S_\sq$ and define the two regions $A = B_r(I_0)$ and $A(l) = B_{r+l}(I_0)$ for
some $r \leq L^{*} < L$ and $l \leq L - r$. The parameter $L^{*}$ is a cutoff
of order $L$. For any two ground states $\ket{\psi_1}$ and $\ket{\psi_2}$ of
$H_0^{A(l)}$, define $\rho_i(A) := \Tr{\ket{\psi_i}\bra{\psi_i}}{\Ac}$
(for $i = 1,2$) as their reduced density matrices when the
complement of $A$ (denoted as $\Ac$) is traced out.
\begin{enumerate}

    \setcounter{enumi}{5}

    \item Then $H_0$ satisfies local-TQO if and only if
        \begin{align}
            \label{eq:rdm}
            \|\rho_1(A) - \rho_2(A)\|_1 \leq 2F(l)\,,
        \end{align}
        where $F$ is a decaying function and $\|\cdot\|_1$ denotes the
        Schatten-1 norm. Intuitively, local-TQO formalizes the notion that
        different ground states cannot be distinguished by local observables.

\end{enumerate}

Lastly, the perturbation $V$ is assumed to have the form 
\begin{align}
    V = \sum_{I \in S_\sq} \sum_{r = 0}^L V_I(r)\,,
\end{align}
such that $\text{supp}(V_I(r)) \subseteq B_r(I)$ and $\|V_I(r)\| \leq J f(r)$
for some constant $J > 0$ and a rapidly decaying function $f(r)$. Here, the norm 
$\|\cdot\|$ denotes the operator norm that is induced by the scalar product
of $\H_{\Lambda_\sq}$. A specification of the necessary decay rate of $f$ is gven 
in Ref.~\cite{Michalakis2013}. For our purposes this is irrelevant as we only consider 
perturbations of finite range (that is, $f(r)$ can be chosen as $0$ for $r$ larger than 
some fixed threshold). We refer to a perturbation satisfying the aforementioned 
conditions as a $(J,f)$-perturbation.

This preparation allows us to formulate the gap stability result of Ref.~\cite{Michalakis2013}.

\begin{theorem}[Michalakis and Zwolak \cite{Michalakis2013}]
	\label{thm:Michalakis}
	Let $H_0$ be a Hamiltonian that satisfies conditions (1)-(6) and $V$ be a $(J,f)$-
	perturbation. 
	Then, there exist finite thresholds $J_0 > 0$ and $L_0 \geq
	2$ such that the gap of $H$ remains uniformly bounded from below for $L \geq
	L_0$ and $J \leq J_0$.
\end{theorem}

\paragraph{Locality and frustration-freeness of $\tilde H_\G(0,\omega)$.}
\label{par:loc}

\begin{figure}[tb]
    \centering
    \includegraphics[width=0.9\linewidth]{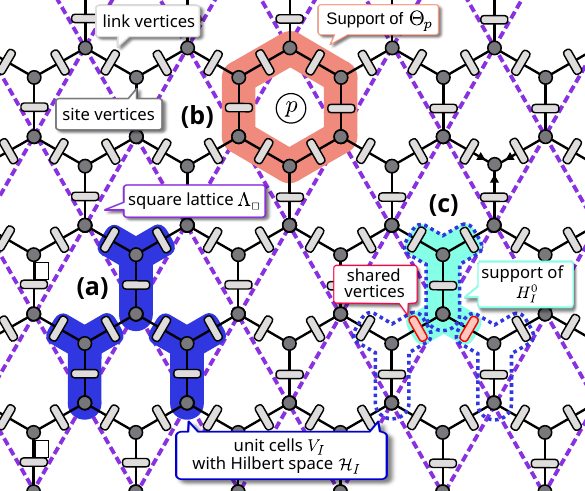}
    \caption{%
        \CaptionMark{Coarse graining of the Hamiltonian $\tilde H_0$.}
        Summary of the construction of a local, frustration-free decomposition
        of the Hamiltonian $\tilde H_0$. The blockade graph is
        represented by gray circles and light gray rectangles, connected by
        black lines. The gray circles represent the site vertices, the light
        gray rectangles represent the link vertices. 
        \textbf{(a)} The vertex set of $\G$ is partitioned into unit cells
        $V_I$, highlighted in dark blue. These unit cells form a square
        lattice, indicated by purple dashed lines. neighboring sites in the
        square lattice contain two-level systems that are in blockade. We
        associate a Hilbert space $\H_I$ to each unit cell (the tensor product
        of the Hilbert spaces of the constituents of $V_I$).
        \textbf{(b)} The support of the plaquette operator $\Theta_p$ is
        highlighted in red. It overlaps with four unit cells.
        \textbf{(c)} The Hamiltonian $H_\G^0$ is partitioned into local terms
        $H^0_{I}$. The support of one such operator is highlighted in cyan.
        Note that the support of this operator contains vertices that belong to
        the neighboring unit cells of $V_I$. These shared vertices are
        highlighted in red. The blue dotted lines mark the unit cells
        affected by $H^0_I$.
    }
    \label{fig:coarsegr}
\end{figure}

To apply \cref{thm:Michalakis}, we first have to define
suitable local Hilbert spaces and a decomposition of the form~\eqref{eq:decomp}
of the unperturbed Hamiltonian 
\begin{align}
    \tilde{H}_0 
    \equiv \tilde{H}_\G(0,\omega) 
    = H^0_\G + \omega \sum_{p} (\mathds{1} - \Theta_p)
    \,,
\end{align}
such that it is frustration-free. Note that due to the blockade interactions
frustration-freeness is a nontrivial property.

To this end, we follow the same procedure as in Ref.~\cite{Maier2025}. 
We partition the vertex set $V_\G$ of the blockade graph $\G$ into unit
cells $V_I$ as shown in \cref{fig:coarsegr}. These unit cells consist of the
vertex sets associated to two sites from $S_\hex$ and three links from
$L_\hex$; as a consequence, these unit cells form a square-lattice $\Lambda_\sq
= (S_\sq, L_\sq, P_\sq)$. We can view $L_\sq$ as a subset of $L_\hex$, where
the (vertical) links connecting the sites that belong to the same unit cell are
removed. Thus our orientation convention for the direction of links
in $L_\hex$ (see \cref{fig:fig1})
induces an orientation convention for the direction of links in $L_\sq$.
To the unit cells we associate
the local Hilbert spaces $\H_I := \bigotimes_{i \in V_I} \H_i$, where $\H_i
\simeq \mathbb{C}^2$ denotes the Hilbert space of one two-level-system. 
Next, we extend the square lattice $\Lambda_\sq$ to a finer lattice
$\tilde{\Lambda}_\sq = (\tilde{S}_\sq, \tilde{L}_\sq, \tilde{P}_\sq)$ that has
one site for every edge, site and plaquette of $\Lambda_\sq$. We refer to the
set of sites that arise from plaquettes of $\Lambda_\sq$ as $\tilde{S}_\sq^P$,
and the to the set of sites that arise from links as $\tilde{S}_\sq^L$. This
lattice is again endowed with the $\ell^\infty$-norm. For $I \in
\tilde{S}_\sq^L \cup \tilde{S}_\sq^P$, we associate the local Hilbert spaces
$\H_I \simeq \mathbb{C}$ (i.e., no degree of freedom). This construction
is a formality needed to define the local decomposition for $\tilde H_0$.
In summary, the Hilbert space of our system is decomposed as
\begin{align}
	\H_\G = \bigotimes_{I\in \tilde{S}_\sq} \H_I.
\end{align}

To obtain a local decomposition of the Hamiltonian $H^0_\G$, we define the
vertex set $\text{Int}(V_I)$ as the subset of $V_I$ consisting of the vertices
on the sites $s \in S_\hex$ that are part of $V_I$, and the vertices on the
link $l \in L_\hex$ that connects the two sites within $V_I$. Moreover, we
define $\partial V_I$ to consist of the vertices on the four links $l \in
L_\hex$ that connect a site in $V_I$ to a site not in $V_I$. Be aware that
$\partial V_I \nsubseteq V_I$. In addition, we define $\bar{V}_I :=
\text{Int}(V_I) \cup \partial V_I$; see \cref{fig:coarsegr}.
This allows us to define the local operators
\begin{align}
	\label{eq:HI}
    H^0_{I} 
    := - \sum_{\mathclap{i \in \text{Int}(V_I)}} \Delta_i n_i 
    - \sum_{\mathclap{i \in \del V_I}} \frac{\Delta_i}{2} n_i
    + U_0 \sum_{\mathclap{\substack{i,j \in \bar{V}_I\\i \sim j}}} n_i n_j
    \,,
\end{align}
where the notation $i \sim j$ denotes vertices in $\G$ that are in blockade.

The induced subgraph of $\bar{V}_I$ and the weights from \cref{eq:HI} define a
blockade graph $\G_I$. The ground state excitation patterns of \cref{eq:HI},
restricted to $\bar{V}_I$, are exactly the maximum-weight independent sets
(MWIS) of $\G_I$. The Hamiltonians $H^0_{I}$ are constructed such that the sum
$H_I^0 + H_J^0$ for two adjacent sites $I, J \in \S_\sq$ is equivalent to the
amalgamation of the blockade graphs $\G_I$ and $\G_J$. 
Hence
\begin{align}
    \label{eq:H0decomp}
    H^0_\G = \sum_{I \in S_\sq} H^0_I
\end{align}
is equivalent to the amalgamation of the blockade graphs $\G_I$ for all $I \in
S_\sq$. 
Thus the decomposition~\eqref{eq:H0decomp} is frustration-free if and only if
there exists a globally consistent independent set on $\G$ (see
\cref{sec:MWIS_tess}). Since such sets exist (by construction),
\cref{eq:H0decomp} is a local, frustration-free decomposition of $H^0_\G$.

To obtain a decomposition of the full Hamiltonian $\tilde H_0$, we define the
local operators
\begin{align}
    Q_I := 
    \begin{cases} 
        \omega(\mathds{1} - \Theta_p), & I = p \in \tilde{S}_\sq^P
        \\
        H_{I}^0, & I \in S_\sq
        \\
        0, & I \in \tilde{S}_\sq^E
    \end{cases},
\end{align}
such that 
\begin{align}
    \label{eq:Q}
    \tilde{H}_0 
    = \sum_{I \in \tilde{S}_\sq} Q_I\,.
\end{align}
First, notice that for all $I \in \tilde{S}$, $\text{supp}(Q_I) \subseteq
B_2(I)$, i.e., each term has constant range of support. This
establishes condition (1) [and (2)] from above.

To show frustration-freeness, note that the ground state energy
$E_0[\tilde{H}_0]$ of $\tilde{H}_0$ is lower bounded by the sum of the smallest
eigenvalues $q_{0,I}$ of $Q_I$, i.e., 
\begin{align}
    \label{eq:bound}
    E_0[\tilde{H}_0] \geq \sum_{I \in \tilde{S}_\sq} q_{0,I}.
\end{align}
Further, note that the operators $Q_I$ mutually commute, i.e., 
$[Q_I, Q_J] = 0$ for all $I,J \in \tilde{S}_\sq$. Thus there exists a common 
eigenbasis for all operators $Q_I$.
In particular, this implies that it suffices to construct one state that 
saturates the bound~\eqref{eq:bound}, because
then any other state with $Q_I \ket{\psi} = q_I \ket{\psi}$ and $q_I >
q_{0,I}$ has strictly larger energy.
To this end, consider the state
\begin{align}
	\ket{\omega} 
  := \frac{1}{|\L_{\G}|}\sum_{\vec{n} \in \L_{\G}} \ket{\vec{n}}
  \,.
\end{align}
The operators $\Theta_p(h)$ define bijective maps $\L_{\G} \rightarrow \L_{\G}$
for each $h \in G$. Hence, applying $\Theta_p(h)$ to $\ket{\omega}$ leads to a
permutation of the summands, i.e., $\Theta_p(h) \ket{\omega} = \ket{\omega}$.
This shows that for every $I = p \in \tilde{S}_\sq^P$, $\ket{\omega}$ satisfies
$Q_p \ket{\omega} = 0$. As $\mathds{1} - \Theta_p$ is a projector, it has
eigenvalues $0$ and $1$ and $Q_p$ has eigenvalues $0$ and $\omega$. It follows
that $\ket{\omega}$ is an eigenstate of $Q_p$ with minimal eigenvalue. Moreover
$\ket{\omega}$ also is an eigenstate with minimal eigenvalue of $H^0_{I}$ as it
is a linear combination of states in $\H^0_\G$. Consequently,
$\ket{\omega}$ saturates \cref{eq:bound} which proves the frustration-freeness
of the decomposition~\eqref{eq:Q}; hence condition (3) from above is
satisfied.

This construction also shows that $\tilde{H}_0$ has a spectral gap. To this
end, note that for $I \in S_\sq$, $H^0_{I}$ has a spectral gap of at least
$\min_{i}(\Delta_i)/2$ [we assume $\max_i(\Delta_i) \ll U_0$] and for $p \in
\tilde{S}_\sq^P$, $Q_p$ has a spectral gap of $\omega$. Hence
if one of the common eigenvectors of the operators $Q_I$ fails to be a ground
state of some operator $Q_I$, its energy is increased by  at least
$\min\{\min_i(\Delta_i)/2,\omega\}$, independent of the system size;
this establishes condition (4) from above.

The local-gap condition (5) from above can be verified with a similar argument.
Here, for $A = B_r(I_0)$ with $I_0 \in \tilde{\Lambda}_\sq$ and $r \geq 0$, we
must consider the localized Hamiltonian 
\begin{align}
    \tilde{H}^A_{0} := \sum_{\mathclap{\text{supp}(Q_I) \subseteq A}} Q_I\,.
\end{align}
Moreover, we define the localized Hamiltonian without plaquette fluctuations
\begin{align}
	H_0^A := \sum_{\text{supp}(H_I^0) \subseteq A} H_I^0.
\end{align}
If we denote the set of ground state configurations of $H_{0}^A$
as $\L_{\G}^A$, the same arguments as above shows
that the state
\begin{align}
    \ket{\omega}_A 
    := \frac{1}{|\L_{\G}^A|} \sum_{\vec{n} \in \L_{\G}^A} \ket{\vec{n}}
\end{align}
is a ground state minimizing the eigenvalue of all operators $Q_I$ with
$\text{supp}(Q_I) \subseteq A$, and that the spectral gap is lower bounded by
$\min\{\min_i(\Delta_i)/2, \omega\}$.

\paragraph{Local-TQO.}

\begin{figure}[tb]
    \centering
    \includegraphics[width=0.75\linewidth]{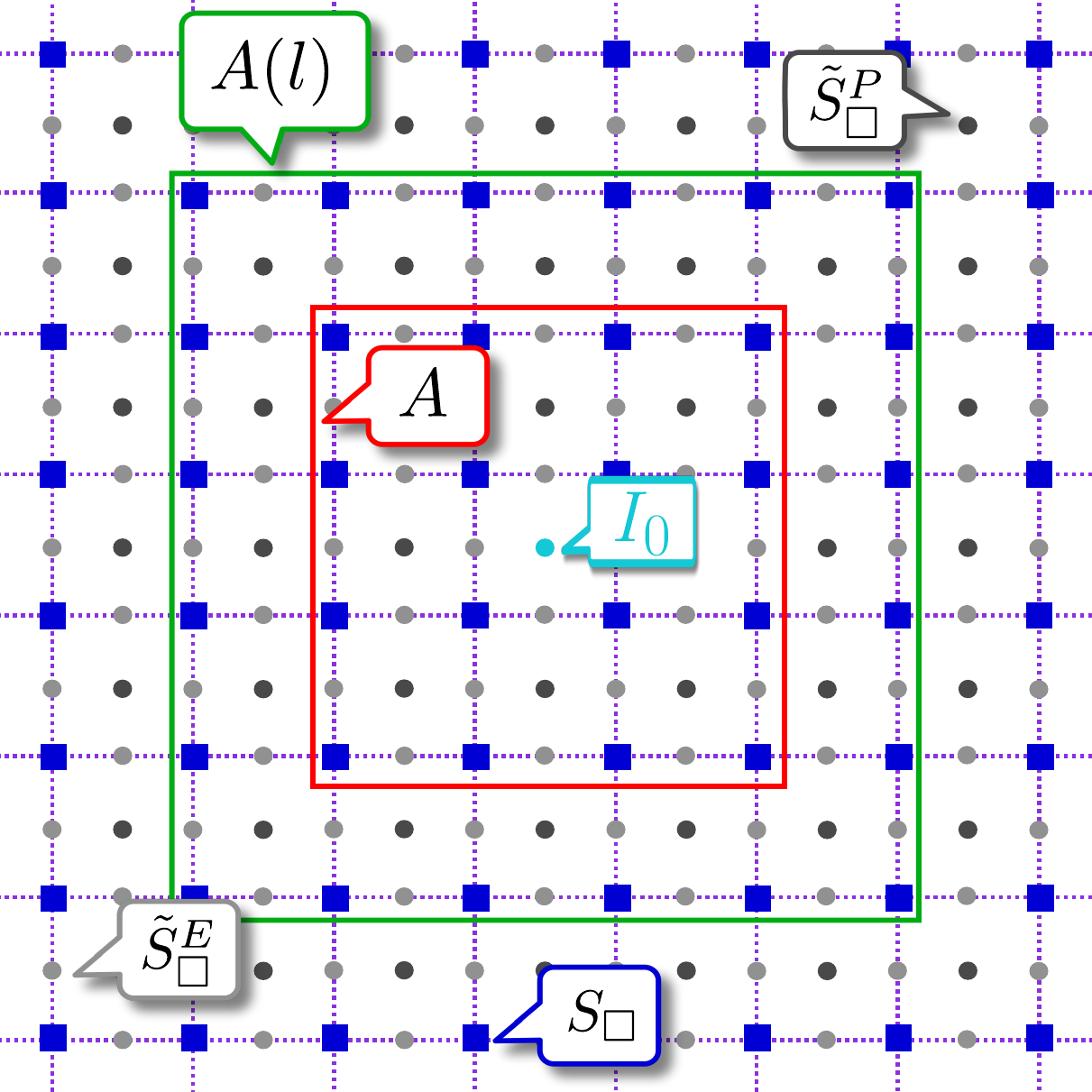}
    \caption{%
        \CaptionMark{Setting for local-TQO.}
        Summary of the construction to verify local-TQO. The blue squares
        represent the sites of the square lattice $\Lambda_\sq$ and the purple
        dotted lines its edges. The dark gray (light gray) circles represent
        the sites of the extended lattice $\tilde{\Lambda}_\sq$ that correspond
        to plaquettes (edges) of $\Lambda_\sq$. The boundary of the rectangular
        region $A = B_r(I_0)$ is highlighted by a red square, the central site
        $I_0$ is highlighted cyan. The boundary of the region $A(l) =
        B_{r+l}(I_0)$ is highlighted by a green square. The figure shows the
        configuration for $r = 3$ and $l = 2$.
    }
    \label{fig:neighborhoods}
\end{figure}

To verify the local-TQO condition~(6), we show a stronger condition, namely
that there exists an $r$-independent bound $l^{*}$ such that for $l \geq
l^{*}$, the reduced density matrices [\cref{eq:rdm2} below] are \emph{equal}. We
follow the proof of Cui~\etal~\cite[Theorem 3.1]{Cui_2020} (who proved this
condition for the original quantum double models).

Let $A = B_r(I_0)$ for some $I_0 \in \tilde{S}_\sq$ and $r \geq 0$. Moreover,
let $A(l) = B_{r+l}(I_0)$ for $l \geq 4 =: l^*$; these regions are shown in
\cref{fig:neighborhoods}. 
As before, $\tilde{H}_0^{A(l)}$ denotes the localized Hamiltonian and
$\H_\G^{A(l),0}$ its classical ground state manifold for $\omega = 0$, spanned
by excitation patterns in $\L_\G^{A(l)}$. 
For an excitation pattern $\vec{n}$, we denote its restriction to vertices that
are part of some set $X \subseteq S_\sq$ as $\vec{n}|_X$. This allows us to
define the set $\L_{\G,A}^{A(l)} := \{\vec{n}|_A \, |\, \vec{n} \in
\L_\G^{A(l)}\}$. We indicate excitation patterns of vertices in $X \subseteq
S_\sq$ by writing $\vec{n}_X$. 

With these conventions, we find that an arbitrary ground state $\ket{\omega}
\in \H_\G^{A(l),\omega}$ can be written as 
\begin{subequations}
	\begin{align}
		\ket{\omega} 
    &= \sum_{\vec{n} \in \L_\G^{A(l)}} C(\vec{n}) \ket{\vec{n}}
    \\
		&= \sum_{\vec{n}_A \in \L_{\G,A}^{A(l)}} \ket{\vec{n}_A}_A \ket{\omega_\Ac(\vec{n}_A)}_\Ac
    \,,\label{eq:SD}
	\end{align}
\end{subequations}
with some coefficients $C(\vec{n}) \in \mathbb{C}$. 
These coefficients cannot be arbitrary but have to be chosen such that
$\Theta_p(h) \ket{\omega} = \ket{\omega}$ is satisfied for every plaquette $p
\in A(l)$ and $h \in G$. Note that the implicitly defined states
$\ket{\omega_\Ac(\vec{n}_A)}$ are not normalized.

We use \cref{eq:SD} to construct a Schmidt decomposition of $\ket{\omega}$. To
this end, let $L_{\sq,A}$ denote the set of links emanating from sites in
$S_\sq \cap A$. On this finite patch of the square lattice, we can define the
quantum double Hilbert space $\H_\SG$ and operators $A_p(h)$ and $B_I$ for $h
\in G$ as usual (see Ref.~\cite{Cui_2020}; the $\square$ in $\H_\SG$
indicates that this is a quantum double on the square lattice). To define
these operators, we use the orientation conventions of $L_\sq$ as induced by
our orientation convention on $L_\hex$ (recall \cref{par:loc}). For plaquettes
with boundary edges that do not belong to $L_{\sq,A}$, the operators $A_p(h)$
are still defined as usual on the edges that are part of $L_{\sq,A}$. 

Let $\H_{\SG,A}^0$ denote the ground state manifold of the quantum double
Hamiltonian $H_{\SG}$ on the square lattice for $J_p = 0$. As this Hamiltonian
is diagonal in the group basis, we can define $\L_{\SG,A}$
such that $\H_{\SG,A}^0 =\spn{\ket{\vec{g}_A}\,|\,\vec{g}_A\in \L_{\SG,A}}$.
The operators $A_p(h)$ are
permutation matrices when represented in the group basis. Thus they induce a
group action on $\L_{\SG,A}$ by $A_p(h)\ket{\vec{g}_A} = \ket{A_p(h) \cdot
\vec{g}_A}$. By abuse of notation, we use the same symbol for the operator
$A_p(h)$ and its induced group action on $\L_{\SG,A}$.

From the condition $l \geq 4$ it follows that $\text{supp}(H_I^0) \subseteq
A(l)$ for all $I \in A$. Thus the excitation patterns $\vec{n}_A$ are
maximum-weight independent sets of the amalgamation of all site structures
$\G_s$ where the site $s \in S_\hex$ is part of a site $I \in A$. Hence,
by construction, they can by mapped
bijectively to ground state configurations of the quantum double on $L_\hex$.
For such configurations on $L_\hex$, the group elements on the edges in $L_\sq$
uniquely determine the group elements on the remaining (vertical)
edges. Thus there exists a bijection $\zeta: \L_{\G,A}^{A(l)} \rightarrow
\L_{\SG,A}$. 
This bijection satisfies
\begin{align}
	\label{eq:exchange}
	\zeta(\Theta_p(h) \cdot \vec{n}_A) 
  = A_p(h) \cdot \zeta(\vec{n}_A)\,,
\end{align}
for all $\vec{n}_A \in \L_{\G,A}^{A(l)}$, $h \in G$ and $p\in A(l)$.

For $\vec{g}_A \in \L_{\SG,A}$, we denote by $\vec{g}_A|_\dA$ the restriction of
this configuration to the links in $L_{\sq, A}$ that cross the boundary between
$A$ and $\Ac$. Cui~\etal~\cite{Cui_2020} showed that for any two configurations
$\vec{g}_A, \vec{g}'_A \in \L_{\SG,A}$ with $\vec{g}_A|_\dA = \vec{g}'_A|_\dA$,
there exists $h_p \in G$, such that $A_\text{int} := \prod_{p \in A} A_p(h_p)$
satisfies $\vec{g}'_A = A_\text{int}\cdot\vec{g}_A$.
Note that $A_\text{int}$ acts trivially on the group elements on the boundary
of $A$. It follows that for any $\vec{n}_A, \vec{n}'_A \in \L_{\G,A}^{A(l)}$
with $\zeta(\vec{n}_A)|_\dA = \zeta(\vec{n}'_A)|_\dA$, there exists $h_p \in G$
for $p \in A$, such that $\Phi_\text{int} := \prod_{p\in A} \Theta_p(h_p)$ 
satisfies $\vec{n}'_A = \Phi_\text{int}\cdot \vec{n}_A$. 
As $\ket{\omega}$ is invariant under $\Theta_p(h)$ for all $h \in G$, it
follows that 
\begin{subequations}
	\begin{align}
		\ket{\omega_\Ac(\vec{n}_A)}_\Ac 
    &= \bra{\vec{n}_A}_A \ket{\omega}\\
    &= \bra{\vec{n}_A}_A \Phi_\text{int}^\dagger \ket{\omega}\\
		&= \bra{\vec{n}'_A}_A \ket{\omega}\\
    &= \ket{\omega_\Ac(\vec{n}'_A)}_\Ac
    \,.
	\end{align}
\end{subequations}
Thus the state $\ket{\omega_\Ac(\vec{n}_A)}_\Ac$ only depends on
$\zeta(\vec{n}_A)|_\dA$. 
Define $\zeta_\dA(\vec{n}_A) := \zeta(\vec{n}_A)|_\dA$, then by abuse of
notation, we can denote this state as
$\ket{\omega_\Ac(\zeta_\dA(\vec{n}_A))}_\Ac$. 

Define the set $\L_{\SG,\dA} := \zeta_\dA(\L_{\G,A}^{A(l)})$ and for
$\vec{g}_\dA \in \L_{\SG,\dA}$, define $\L_{\G,A}^{A(l)}(\vec{g}_\dA) :=
\{\vec{n}_A \in \L_{\G,A}^{A(l)}\, |\, \zeta_\dA(\vec{n}_A) = \vec{g}_\dA\}$.  
Then \cref{eq:SD} becomes 
\begin{align}
	\label{eq:SD2}
	\ket{\omega} 
  =\;\sum_{\mathclap{\vec{g}_\dA \in \L_{\SG,\dA}}}\;
  \mathcal{N}_1(\vec{g}_\dA)\ket{\xi_A(\vec{g}_\dA)}_A \ket{\omega_\Ac(\vec{g}_\dA)}_\Ac
\end{align}
where 
\begin{align}
	\ket{\xi_A(\vec{g}_\dA)}_A 
  := \frac{1}{\mathcal{N}_1(\vec{g}_\dA)} 
  \sum_{\vec{n}_A \in \L_{\G,A}^{A(l)}(\vec{g}_\dA)} \ket{\vec{n}_A}_A
\end{align}
and $\mathcal{N}_1(\vec{g}_\dA) := \sqrt{|\L_{\G,A}^{A(l)}(\vec{g}_\dA)|}$.

We now show that \cref{eq:SD2} is a Schmidt decomposition of $\ket{\omega}$. To
this end, we must show that the states $\ket{\omega_\Ac(\vec{g}_\dA)}_\Ac$ are
orthogonal and have the same norm. 

For orthogonality, suppose that $\vec{g}_\dA$ and $\vec{g}'_\dA$ differ in link
$l$, i.e., $g_l \neq g'_l$. Let $I \in S_\sq$ be the unique site in $\Ac$ from
which $l$ emanates. As $l \geq 4 = l^*$ (be aware, that this "$l$" refers
to the length that is used to define $A(l)$ and not to a link.),
$\text{supp}(H_I^0) \subseteq A(l)$, and therefore every excitation pattern
from $\L_\G^{A(l)}$ restricted to $I$ can be mapped bijectively to group
elements on the links emanating from $I$. For every excitation pattern $\vec{n}
\in \L_{\G}^{A(l)}$, both $\vec{n}|_A$ and $\vec{n}|_I$ must associate the same
group element to link $l$, otherwise $\vec{n}$ cannot be a ground state
pattern. 
Suppose $\vec{n}_\Ac$ is an excitation pattern in the expansion of
$\ket{\omega_\Ac(\vec{g}_\dA)}_\Ac$ and $\vec{n}'_\Ac$ is an excitation pattern
in the expansion of $\ket{\omega_\Ac(\vec{g}'_\dA)}_\Ac$. Then $g_l \neq g'_l$
implies that $\vec{n}_\Ac|_I \neq \vec{n}'_\Ac|_I$ and thus
$\bra{\vec{n}_\Ac}_\Ac\ket{\vec{n}'_\Ac}_\Ac = 0$. As this holds for all
excitation pattens in the expansions of the respective states, we have shown
the desired orthogonality.

Now we show that all states $\ket{\omega_\Ac(\vec{g}_\dA)}_\Ac$ for
$\vec{g}_\dA \in \L_{\SG,\dA}$ have the same norm. Cui~\etal~\cite{Cui_2020}
showed that for any two configurations $\vec{g}_\dA, \vec{g}'_\dA \in
\L_{\SG,\dA}$, there exist group elements $h_p \in G$ such that $A_\dA :=
\prod_{p \in \dA} A_p(h_p)$ satisfies $A_\dA \cdot \vec{g}'_\dA = \vec{g}_\dA$.
Here, $p \in \dA$ denotes the condition that $p$ is neither contained in $A$
nor in $\Ac$. 
\cref{eq:exchange}, together with the previous paragraph, implies that there exist
group elements $h_p \in G$ such that for $\vec{g}_\dA, \vec{g}'_\dA \in
\L_{\SG,\dA}$, the automorphism $\Phi_\dA := \prod_{p \in \dA} \Theta_p(h_p)$
satisfies 
\begin{align}
    \label{eq:set_eq}
    \L_{\G,A}^{A(l)}(\vec{g}_\dA) 
    = \Phi_\dA \cdot \L_{\G,A}^{A(l)}(\vec{g}'_\dA)
    \,.
\end{align}
We can factor $\Phi_\dA = \Phi_\dA^A \Phi_\dA^\Ac$, where $\Phi_\dA^A$ acts
trivially on vertices in $\Ac$ and $\Phi_\dA^\Ac$ acts trivially on vertices in
$A$. Both of these permutations induce unitary operators on $\H_{\G,A}$ and
$\H_{\G,\Ac}$ respectively. By abuse of notation, we denote them with the same
symbols $\Phi_\dA^A$ and $\Phi_\dA^\Ac$. \cref{eq:set_eq} then implies that
$\Phi_\dA^A \ket{\xi_A(\vec{g}'_\dA)} = \ket{\xi_A(\vec{g}_\dA)}$. In
particular, this shows that $\mathcal{N}_1 \equiv \mathcal{N}_1(\vec{g}_\dA)$.
Moreover, from $\Phi_\dA \ket{\omega} = \ket{\omega}$ it follows that 
\begin{subequations}
	\begin{align}
		\ket{\omega_\Ac(\vec{g}'_\dA)}_\Ac &= \bra{\xi_A(\vec{g}'_\dA)}_A \ket{\omega}\\
		&= \bra{\xi_A(\vec{g}'_\dA)}_A  \Phi_\dA^\dagger \ket{\omega}\\
		&= \bra{\xi_A(\vec{g}'_\dA)}_A  (\Phi^A_\dA)^\dagger (\Phi^\Ac_\dA)^\dagger \ket{\omega}\\
		&= (\Phi^\Ac_\dA)^\dagger \bra{\xi_A(\vec{g}_\dA)}_A   \ket{\omega}\\
		&= (\Phi^\Ac_\dA)^\dagger \ket{\omega_\Ac(\vec{g}_\dA)}_\Ac\,.
	\end{align}
\end{subequations}
Since $(\Phi^\Ac_\dA)^\dagger$ is unitary, the states
$\ket{\omega_\Ac(\vec{g}_\dA)}_\Ac$ and $\ket{\omega_\Ac(\vec{g}'_\dA)}_\Ac$
have the same norm; we denote this norm as $\mathcal{N}_2$ and the states
$1/\mathcal{N}_2 \ket{\omega_\Ac(\vec{g}_\dA)}_\Ac$ are orthonormal. 

Consequently, taking the partial trace of \cref{eq:SD2} over $\Ac$ yields
\begin{align}
    \label{eq:rdm2}
    \rho_{A} 
    = (\mathcal{N}_1\mathcal{N}_2)^2 \sum_{\vec{g}_\dA 
    \in \L_{\SG,\dA}} \ket{\xi_A(\vec{g}_\dA)}_A\bra{\xi_A(\vec{g}_\dA)}_A
    \,.
\end{align}
The only terms in \cref{eq:rdm2} that could depend on the state $\ket{\omega}$
are the normalization constants $\mathcal{N}_1$ and $\mathcal{N}_2$. However,
since all states $\ket{\xi_A(\vec{g}_\dA)}_A$ are normalized, the condition $1
= \tr{\rho_A}$ implies that $(\mathcal{N}_1\mathcal{N}_2)^2 = |\L_{\SG,\dA}|$,
which is independent of $\ket{\omega}$. 

Thus we have shown that the reduced density matrix $\rho_A$ is the same for all
states $\ket{\omega} \in \H_{\G}^{A(l),\omega}$, as desired.

\section{Wilson loops}%
\label{app:wilson}

In \cref{sec:fluxlattice} we used the Wilson loop (operator) \eqref{eq:Wilsonloop}
of the quantum double model
\begin{align}
    \hat W^{R}(\gamma) 
    := \sum_{\ket{\vec g}\in \H_G} \chi_R(g_\gamma) \ket{\vec g}\bra{\vec g} 
    \label{eq:WRFourier}
\end{align}
with product $g_\gamma := \prod\nolimits_{l \in \gamma} g_{l}^{\sigma_l}$ along
a closed, oriented loop $\gamma$ on the dual lattice; the sign functions
$\sigma_l$ are defined in \cref{sec:fluxlattice} (see also \cref{fig:fig4}).
$\chi_R$ denotes the character of the irreducible representation (irrep) $R$ of
the group $G$. Note that in \cref{sec:fluxlattice} we work with the matrix elements
of the operator \eqref{eq:WRFourier} in the product basis $\ket{\vec g}$ for
simplicity.

At the fixpoint of the quantum double phase, the measurement of the Wilson loop
operators over all irreps $\hat G$ of the group $G$ uniquely determines the
enclosed flux, independent of the shape of the loop. 
Here we demonstrate this property by explicitly performing the discrete Fourier
transform for class functions
\begin{align}
    \hat W^{C}(\gamma) 
    := \frac{1}{|Z_G(r_C)|}\sum_{R \in \hat G} \chi_R^*(r_C)\hat W^{R}(\gamma)
    \label{eq:WCFourier}
\end{align}
where $Z_G(r_C)$ is the centralizer of the representative $r_C \in C$ of
conjugacy class $C \in \text{Cl}(G)$. Note that $|Z_G(\cdot)|$ is a class
function since centralizers of different representatives are isomorphic via
conjugation. By the orbit-stabilizer theorem, we can rewrite the cardinality of
the centralizer $|Z_G(r_C)| = |G|/|C|$ for any $r_C \in C$, this makes
\cref{eq:Wilsonloop2a} and \cref{eq:WCFourier} equivalent.

The \emph{completeness} of the character on the set of class functions can be
formulated as
\begin{align}
    \sum_{R \in \hat G} \chi_R^*(g)\chi_R(h) = |Z_G(g)|\delta_{g \sim h},
\end{align}
where $\delta_{g \sim h}=1$ if $g$ and $h$ are conjugate and $\delta_{g \sim
h}=0$ otherwise.
This allows us to rewrite the Fourier transform \eqref{eq:WCFourier} as
\begin{align}
    \hat W^{C}(\gamma) = \sum_{\ket{\vec g}\in \H_G} \delta_{r_C \sim g_\gamma}\ket{\vec g}\bra{\vec g};
\end{align}
where $\delta_{r_C \sim g_\gamma}=1$ if $g_\gamma \in C$ and zero and
$\delta_{r_C \sim g_\gamma}=0$ otherwise; this shows \cref{eq:Wilsonloop2b}.

As an example, consider a state $\ket{C_s}\in\H_{\tilde
\G}^0\cap\H_{\tilde\G}^S$ where one flux anyon $[C,E]$ is pinned on site $s$.
We consider loops $\gamma$ that enclose $s$. Note that the product $g_\gamma$
is conserved when reshaping the loop such that it traverses a site without flux
anyon (since then $g_0g_1g_2 = 1$ for the group elements on the three emanating
links). This allows us to contract the loop $\gamma$ to enclose only the site
$s$ with the flux anyon. Then $g_\gamma$ is just the product of the three group
elements of its emanating links, which is in $C$ for all states $\ket{\vec g}$
in $\ket{C_s}$.
Then $\bra{C_s}\hat W^R(\gamma)\ket{C_s} = \chi_R(r_C)$ for some representative
$r_C \in C$ (as proposed in \cref{sec:fluxlattice}) and $\bra{C_s}\hat
W^{C'}(\gamma)\ket{C_s} = \delta_{C',C}$ is nonzero only for $C' = C$.

\section{Braiding in $\double{S_3}$}%
\label{app:braiding}

\begin{figure}[tb]
    \centering
    \includegraphics[width=1.0\linewidth]{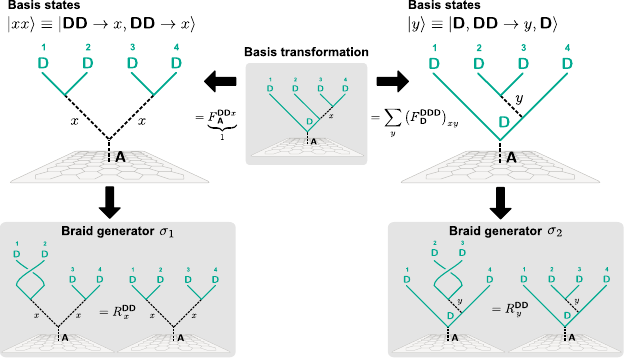}
    \caption{%
        \CaptionMark{Basis states and braid group generators}. 
        The fusion basis $\ket{xx}\equiv \ket{\anyon{DD}\rightarrow x,
        \anyon{DD}\rightarrow x}$ of $\H^\anyon{DDDD}_\anyon{A}$, which
        describes four $\anyon D$-anyons that fuse into the vacuum $\anyon A$,
        is shown on the left as a splitting diagram.  In this basis, the braid
        generators $\sigma_1$ and $\sigma_3$ are diagonal with
        $R^\anyon{DD}_x$.  Another basis $\ket{y} \equiv \ket{\anyon
        D,\anyon{DD}\rightarrow y,\anyon{D}}$ is shown on the right. In this
        basis, the braid generator $\sigma_2$ is diagonal with
        $R^\anyon{DD}_y$. The basis transformation between $\ket{xx}$ and
        $\ket{y}$ is performed via two $F$-moves: The first is trivial (since
        $F^{\anyon{DD}x}_{\,\anyon{A}}=1$) and the second gives rise to a
        non-trivial basis transformation via $\FmatD{D}$.
    }
    \label{fig:fig6-braid}
\end{figure}

The braid group $B_4$, which describes the world lines of four anyons that
start and end aligned on a row at fixed positions, is generated by three
operations: $\sigma_1$ exchanges the first anyon with the second, $\sigma_2$
the second anyon with the third, and $\sigma_3$ the third with the fourth. We
define these exchanges with an anti-clockwise orientation, as shown in
\cref{fig:fig6,fig:fig6-braid}. 
In \cref{sec:braiding} we introduced the basis
$\ket{xx}\equiv\ket{\anyon{DD}\rightarrow x,\anyon{DD}\rightarrow x}$ of the
fusion space $\H^\anyon{DDDD}_\anyon{A}$ which contains the states with four
$\anyon D$-anyons that fuse into the vacuum $\anyon A$. The fusion algebra of
$\double{S_3}$ allows for $ x \in\{ \anyon A, \anyon C, \anyon F, \anyon G,
\anyon H\}$, i.e., the first two anyons fuse into $x$ and the remaining two
anyons also fuse into $x$, which finally fuse into the vacuum $\anyon A$, see
\cref{fig:fig6-braid}. The fusion outcome of both pairs must be equal because
those fuse into the vacuum and in $\double{S_3}$ all anyons are their own
antiparticle. It follows that $\dim\H^\anyon{DDDD}_\anyon{A}=5$, which reflects
the non-abelian nature of the $\double{S_3}$ anyon theory.

Braiding anyons effects unitary operations on $\H^\anyon{DDDD}_\anyon{A}$. Due
to locality, the representations of the two braid generators $\sigma_1$ and
$\sigma_3$ are diagonal in the basis $\ket{xx}$ (braiding two anyons cannot
change their fusion channel):
\begin{subequations}
    \begin{align}
    \sigma_1:\;\ket{{\color{red}x} x} &\mapsto  R^\anyon{DD}_{\color{red}x} \ket{{\color{red}x}x}\,,\\
    \sigma_3:\;\ket{x {\color{blue}x}} &\mapsto  R^\anyon{DD}_{\color{blue}x} \ket{x{\color{blue}x}}\,,
    \end{align}
\end{subequations}
with $R^\anyon{DD}_x\in\mathbb{C}$ the $R$-matrices (phases) for braiding two
$\anyon D$-anyons that are in fusion channel $x$.
For the quantum double $\double{S_3}$, one finds \cite{Cui2015}
\begin{subequations}
    \begin{align}
        \RmatD{A} &= \RmatD{C} = \RmatD{F} = -1\,,\\
        \RmatD{G} &=- \bar{\omega}^2\,,\\
        \RmatD{H} &=- \bar{\omega}\,,
    \end{align}
\end{subequations}
with $\bar{\omega}  = e^{ 2 \pi i/3}$. 
Note that the last two equations implicitly define which anyons we call $\anyon
G$ and $\anyon H$. This was not yet fixed since we did not specify the
representations $\Gamma_{R_1}$ and $\Gamma_{R_2}$ in \cref{sec:braiding}.

The braid generator that allows us to probe the non-abelian statistics is
$\sigma_2$ -- which is \emph{not} diagonal in the $\ket{xx}$ basis since the
fusion channel of anyon 2 and 3 is not determined in this basis.
We denote the basis of $\H^\anyon{DDDD}_\anyon{A}$ in which anyon 2 and 3 fuse
into $y\in\{ \anyon A, \anyon C, \anyon F, \anyon G, \anyon H\}$ as
$\ket{y}\equiv\ket{\anyon D,\anyon{DD}\to y,\anyon D}$. The anyon $y$ then
fuses with the fourth $\anyon D$-anyon into $\anyon D$, which finally fuses
with the first $\anyon D$-anyon into the vacuum, see \cref{fig:fig6-braid}.
(The fusion channel of $y$ with the fourth $\anyon D$-anyon must be $\anyon D$
because this is the only way to fuse the first $\anyon D$-anyon into the vacuum
$\anyon A$.)

Again due to locality, in this basis the unitary representation of the braid
generator $\sigma_2$ is diagonal:
\begin{align}
    \sigma_2:\;\ket{y} \mapsto R^\anyon{DD}_y \ket{y}\,.
\end{align}

The basis change from $\ket{xx}$ to $\ket{y}$ is achieved by two $F$-moves, see
\cref{fig:fig6-braid}, where the first move is trivial because
$F^{\anyon{DD}x}_{\,\anyon{A}} = 1$. Consequently, the basis transformation is
determined by the matrix $\FmatD{D}$,
\begin{align}
    \ket{xx} = \sum_{y} \left(\FmatD{D}\right)_{xy} \ket{y},
\end{align}
with $F$-matrix given by \cite{Cui2015}
\begin{align}
    \FmatD{D}=
    \frac{1}{3} 
    \begin{pmatrix}
        1 & \sqrt{2} & \sqrt{2} & \sqrt{2}& \sqrt{2} \\
        \sqrt{2} & 2 & -1 & -1 & -1\\
        \sqrt{2}& -1 & 2 & -1 & -1\\
        \sqrt{2} & -1 & -1 & -1& 2\\
        \sqrt{2} & -1 & - 1 & 2 & -1
    \end{pmatrix}\,.
\end{align}
Here, the order of the basis states is given by $\{\anyon A, \anyon C, \anyon
F, \anyon G, \anyon H\}$. Note that $\left(\FmatD{D}\right)^2= \mathds{1}$, so
that the inverse transformation from basis $\ket{y}$ to $\ket{xx}$ is also
given by $\FmatD{D}$.

This allows us to express the matrix $U_{\sigma_2}$, which represents the braid
generator $\sigma_2$ in the basis $\ket{xx}$, by first changing to the basis
$\ket{y}$, then performing the braid $R^\anyon{DD}_y$, and finally switching
back to the basis $\ket{xx}$:
\begin{align}
    U_{\sigma_2}= \FmatD{D}\cdot R^\anyon{DD}\cdot \FmatD{D}
\end{align}
with the diagonal matrix $\left(R^\anyon{DD}\right)_{xy} =
\delta_{xy}R^\anyon{DD}_y$; this yields~\cite{Cui2015}
\begin{equation}
    U_{\sigma_2}
    =-\frac{1}{3}
    \begin{pmatrix}
        1 &  \sqrt{2} &  \sqrt{2} &   \sqrt{2}\,\bar{\omega}&   \sqrt{2}\,\bar{\omega}^2 \\
        \sqrt{2} & 2 & -1 &  - \bar{\omega}&  - \bar{\omega}^2   \\
        \sqrt{2} & -1 & 2 &  - \bar{\omega}&  - \bar{\omega}^2 \\
        \sqrt{2}\,\bar{\omega} & - \bar{\omega} &  - \bar{\omega}&  - \bar{\omega}^2 & 2 \\
        \sqrt{2}\,\bar{\omega}^2 &   - \bar{\omega}^2 &  - \bar{\omega}^2&  2 &  - \bar{\omega}
    \end{pmatrix}\,.
\end{equation}
Applying this operation on our initial state 
\begin{align}
    \BraKet{xx}{AA}=
    \begin{pmatrix}1 & 0 & 0 & 0 & 0\end{pmatrix}^T
\end{align}
leads to the result in \cref{eq:braiding} of the main text (up to a global
phase). 

\end{document}